\documentclass[a4paper,11pt]{article}
\usepackage{soul}
\usepackage[usenames,dvipsnames]{color}
\usepackage{jheppub}
\usepackage{esvect}
\usepackage{relsize}
\usepackage{comment}
\usepackage{amsmath, amssymb, slashed, epsf, color, graphicx, latexsym, bm}
\usepackage{mathrsfs}   
\usepackage{tensor}
\usepackage{slashed}
\usepackage{fixme}
\usepackage{subcaption}
\fxsetup{status=draft}
\usepackage{physics}
\usepackage{tikz}
\usepackage{quiver}
\usepackage{tikz-cd}
\usetikzlibrary{positioning, calc}
\usepackage{dsfont}
\usepackage{pgfplots}
\pgfplotsset{compat=1.18}
\usepackage{epsfig}
\usepackage{graphics}
\usepackage[T1]{fontenc}
\usepackage{mathtools}
\usepackage{fixme}
\fxsetup{status=draft}
\usepackage{caption}
\usepackage{float}
\usepackage{bbm}
\usepackage{subcaption}
\usepackage{enumerate}
\usepackage{verbatim}
\usepackage{blindtext}
\makeatletter
\renewcommand*\env@matrix[1][*\c@MaxMatrixCols c]{%
  \hskip -\arraycolsep
  \let\@ifnextchar\new@ifnextchar
  \array{#1}}
\makeatother
\DeclareFontEncoding{LS1}{}{}
\DeclareFontSubstitution{LS1}{stix}{m}{n}
\DeclareSymbolFont{symbols2}{LS1}{stixfrak}{m}{n}
\DeclareMathSymbol{\typecolon}{\mathbin}{symbols2}{"25}
\usepackage[mathlines]{lineno}

\newcommand{\be}{\begin{equation}}
\newcommand{\ee}{\end{equation}}
\usepackage{environ}

\NewEnviron{eqsp}{%
  \begin{equation}
    \begin{split}
      \BODY
    \end{split}
  \end{equation}
}
\def\bM{\bm{M}}
\def\bN{\bm{N}}

\def\CB{{\cal B}}
\def\CC {{\cal C}}

\def\CF {{\cal F}}

\def\CL {{\cal L}}
\def\CM {{\cal M}}
\def\CN {{\cal N}}
\def\CO {{\cal O}}

\def\CR {{\cal R}}

\def\CO {{\cal O}}

\def\CB {{\cal B}}

\def\CS {{\cal S}}

\def\IC{\mathbb{C}}

\def\IH{\mathbb{H}}

\def\IN{\mathbb{N}}

\def\IQ{\mathbb{Q}}
\def\IR{{\mathbb{R}}}

\def\IX{{\mathbb{X}}}
\def\IZ{{\mathbb{Z}}}

\def\F\IX{\mathfrak{x}}

\def\F\IX{\mathfrak{X}}

\def\r0{\lceil r_0\rceil}
\def\ben{\begin{eqnarray}}
\def\een{\end{eqnarray}}
\def\eps{\epsilon}

\newcommand{\refb}[1]{(\ref{#1})}

\newcommand{\mr}[1]{\mathrm{#1}}

\newcommand{\wt}[1]{\widetilde{#1}}
\newcommand{\wh}[1]{\widehat{#1}}

\begin{document}
\title{Generating Function of single-centered Black Hole Index in CHL Models}
\author{Ranveer Kumar Singh}
\affiliation[]{NHETC and Department of Physics and Astronomy, Rutgers University, 
126 Frelinghuysen Rd., Piscataway NJ 08855, USA}
\emailAdd{ranveersfl@gmail.com}

\abstract{We present the construction of the generating function of single-centered black hole index in general $\mathbb{Z}_N$ CHL models. This is done by subtracting from the index of quarter BPS dyons, described by a meromorphic Siegel modular form, the generating function for the index of two-centered black holes. We use black hole bound state metamorphosis in CHL models for the construction of the generating function of two-centered black hole index. We prove the convergence of the generating function for the cases $N=2,3$. } 
\def\rem{$\clubsuit$}
\newcommand{\cas}[1]{\marginpar{\raggedright \tiny \rem AS: #1 \rem}} % comments by Ashoke
\newcommand{\caj}[1]{\marginpar{\raggedright \tiny \textcolor{blue}{\rem AB: #1 \rem}}}    % comments by Ajit
\newcommand{\cran}[1]{\marginpar{\raggedright \tiny \textcolor{OliveGreen}{\rem RKS: #1 \rem}}}    % comments by Ranveer

\maketitle

\section{Introduction and summary}
Black holes in string theory provide a perfect playground to test the consistency of quantum gravity. One of the most prominent examples is the matching of classical black hole entropy by counting microstates in string theory \cite{Strominger:1996sh,Sen:1995in}. A well-studied instance of this matching is in the 4d $\CN=4$ supersymmetric compactification of string theory \cite{Dijkgraaf:1996it,Dijkgraaf:1996xw,Verlinde:1997fm,LopesCardoso:2004law,Shih:2005uc}. More precisely, one considers the compactification of type II string theory on $K3\times T^2$. The resulting low energy effective theory is a 4d $\CN=4$ supergravity. The entropy of quarter BPS black holes in this $\CN=4$ supergravity can then be reproduced by counting configurations of non-perturbative objects in type II string theory on $K3\times T^2$, see \cite{Sen:2007qy} for a comprehensive review. 
\\\\
In this paper, we will be concerned with a general class of $\CN=4$ compactifications of type II string theory called CHL compactifications \cite{Chaudhuri:1995fk,Chaudhuri:1995bf}. CHL models are obtained by compactifying type II string theory on $\mathcal{M}\times T^2$, with $\mathcal{M}=K3, T^4$, with a $\IZ_N$ orbifold of the worldsheet CFT. The S-duality group of this theory is the following subgroup of $\mr{SL}(2,\IZ)$:
\begin{eqsp}
    \Gamma_1(N):=\left\{\begin{pmatrix}
        a&b\\c&d
    \end{pmatrix}\in\mr{SL}(2,\IZ):c\equiv 0\bmod N, a\equiv d\equiv 1\bmod N\right\}~.
\end{eqsp}
The low energy $\CN=4$ supergravity action up to two derivatives is completely characterized by the number of U(1) gauge fields. At generic points in the moduli space, there are $22+r$ U(1) gauge fields in the supergravity, where $r$ depends on $N$. The theory admits single-centered and two-centered BPS black hole solutions characterized by a $(22+r)$ dimensional electric charge vector $Q$ and a $(22+r)$ dimensional magnetic charge vector $P$. These charges satisfy the following quantization condition:
\begin{eqsp}
    \frac{Q^2}{2}\in\frac{1}{N}\IZ~,\quad \frac{P^2}{2}\in\IZ~,\quad Q\cdot P\in \IZ~.
\end{eqsp}
The moduli space of the theory is parametrized by a complex parameter $\tau_\infty$, called the axio-dilaton moduli and an $r\times r$ matrix valued scalar. For a given charge $(Q,P)$ and for fixed value of the scalar moduli, the moduli space parametrized by $\tau_\infty$ is divided into various chambers bounded by walls of marginal stability \cite{Sen:2007vb}. While single-centered black holes are insensitive to the walls of marginal stability, two-centered black holes may appear or disappear when we cross a wall of marginal stability \cite{Sen:2007pg,Cheng:2007ch}. Thus, the natural choice for testing quantum gravity is in matching the entropy of single-centered black holes by counting microstates. 
\\\\
The relevant object to compute in string theory is the index of quarter BPS dyons. The index of quarter BPS states carrying charges $Q,P$ is defined as \cite{Bachas:1996bp,Gregori:1997hi}
\begin{eqsp}
B_6(Q,P)=\frac{1}{6!} \mr{Tr}_{\mathscr{H}_{Q,P}}(-1)^F(2h)^{6}~,   
\end{eqsp}
where the trace is taken over the space $\mathscr{H}_{Q,P}$ of all states carrying
charges $Q,P$, $F$ is the fermion number operator and $h$ is the third component of the angular
momentum carried by the state. One can show that this  index gets contributions only from quarter BPS states. It was argued using $\mr{AdS_2/CFT_1}$ correspondence in \cite{Sen:2009vz} that the index is equal to negative of the degeneracy of the ensemble represented by the horizon of the single-centered black hole with charge $Q,P$. Hence the index can be used to compute the entropy of black holes.   
This index has been computed for a large class of quarter BPS dyons in CHL models. 
A dyon with charge $(Q,P)$ is called torsion 1 if the charge satisfies 
\be
\gcd\{ Q_i P_j - Q_j P_i\} = 1\, ,
\ee
in some basis for the charge lattice. 
The BPS index of torsion 1 quarter BPS dyons is known to be given by the Fourier
coefficients of the inverse of a Siegel modular form $\Phi_{k}$ of weight $k$ for a subgroup $G$ of 
$\mr{Sp}(2,\mathbb{Z})$ \cite{Jatkar:2005bh,David:2006ji,David:2006ru,David:2006ud,David:2006yn}. More precisely, $\Phi_k$ is a holomorphic function on the Siegel upper half space $\IH_2$ defined by 
\begin{eqsp}
    \IH_2:=\left\{\Omega:=\begin{pmatrix}
        \tau&z\\z&\sigma
    \end{pmatrix}:\tau,\sigma\in\IH,z\in\IC,\mr{Im}\,\Omega>0\right\}~,
\end{eqsp}
where $\IH:=\{x+iy:y>0\}$ is the upper half plane. The positive definiteness condition $\mr{Im}\,\Omega>0$ can be written more simply as follows: let us define 
\be
(\tau_2,\sigma_2,z_2)= {\rm Im}(\tau,\sigma,z), \qquad (\tau_1,\sigma_1,z_1)
= {\rm Re}(\tau,\sigma,z)\, ,
\ee
Then we can write 
\begin{eqsp}
    \IH_2=\left\{\Omega:=\begin{pmatrix}
        \tau&z\\z&\sigma
    \end{pmatrix}:\tau_2,\sigma_2>0,\tau_2\sigma_2-z_2^2>0\right\}~.
\end{eqsp}
$\Phi_k$ satisfies the modularity property
\begin{eqsp}
    \Phi_k((A\Omega+B)(C\Omega+D)^{-1})=\mr{det}(C\Omega+D)^{k}\Phi_k(\Omega)~,\quad \begin{pmatrix}
        A&B\\C&D
    \end{pmatrix}\in G\subset\mr{Sp}(2,\IZ)~.
\end{eqsp}
The weight $k$ and the subgroup $G$ depends on $N$ and the 4-manifold $\mathcal{M}$. 
The BPS index can then be expressed as appropriate Fourier integrals
of $1/\Phi_{k}$.
It turns out that the index of torsion 1 dyons 
is a function of only the quadratic 
combinations $Q^2,P^2$ and $Q\cdot P$ \cite{Dabholkar:2007vk,Banerjee:2007sr}. Let us introduce the notation
\be
m:= Q^2N / 2, \qquad n := P^2 / 2, \qquad \ell :=Q\cdot P\, , \qquad m, n,\ell\in \mathbb{Z}\,.
\ee
The index, which from now on we will denote\footnote{Note that $B_6(Q,P)=-d(m,n,\ell)$. We will abuse terminology and call $d(m,n,\ell)$ the index.} by $d(m,n,\ell)$, is then given by 
\be \label{edmnl}
d(m,n,\ell) = \frac{(-1)^{\ell+1}}{N} \int_{\mathcal{C}}
d\tau d\sigma dz\, e^{-2 \pi i\left(m\tau+n\sigma/N+\ell z\right)}
\frac{1}{\Phi_{k}(\Omega)}\, .
\ee
The integration contour $\CC$ is given by 
\begin{eqsp}
0\leq \tau_1,z_1<1~,\quad 0\leq \sigma_1<N~,    
\end{eqsp}
with a fixed $\tau_2,\sigma_2,z_2$.
\\\\
The Siegel modular form $\Phi_k^{-1}$ has double poles on the hypersurface 
\begin{eqsp}\label{eq:polesn2n0}
    n_2(\tau\sigma-z^2)-m_1\tau+n_1\sigma+jz+m_2~,\quad &m_1/N,m_2,n_1,n_2,j\in\IZ~,\\&m_1n_1+m_2n_2+\frac{j^2}{4}=\frac{1}{4}~.
\end{eqsp}
This causes the index to jump whenever we deform the contour of integration $\CC$ across a pole. This is related to the fact that for a fixed charge, the index generically counts microstates of both single and two-centered black holes and when we deform the contour $\CC$ across a pole of $\Phi_{k}^{-1}$, it corresponds to crossing a wall of marginal stability in the moduli space \cite{Sen:2007vb,Dabholkar:2007vk}. Thus, the $(\tau_2,\sigma_2,z_2)$ space labeling the integration contour 
can also be divided into chambers bounded by the hypersurfaces \eqref{eq:polesn2n0} and the
contour integrals performed in different chambers give different results. In fact, for a given set of charges,
there is a one-to-one correspondence between the
chambers in the moduli space and the chambers in the $(\tau_2,\sigma_2, z_2)$ space
in the region where
\be
\det \rm Im\, \Omega>{1\over 4} \, .
\ee
In this region the walls separating the different chambers  in the $(\tau_2,\sigma_2, z_2)$ space
lie along
\be \label{ewall}
-m_1\tau_2 + n_1\sigma_2 + j\, z_2 = 0, \qquad m_1/N,n_1,j\in\mathbb{Z}, \qquad
m_1n_1+{j^2\over 4}
={1\over 4}\, .
\ee   
The Fourier coefficient $d(m,n,\ell)$ of $\Phi_{k}^{-1}$
in a given chamber in the $(\tau_2,\sigma_2,z_2)$ space, defined by
\refb{edmnl} with $\tau_2,\sigma_2,z_2$ lying inside that chamber,  
gives the index in the
corresponding chamber in the moduli space. 
\\\\
As discussed before, the index typically includes contribution
from both single-centered black holes and two-centered black holes. However,
for a given $(m,n,\ell)$ satisfying
\be \label{erangepositive}
m\ge 0, \qquad n\ge 0, \qquad 4mn-\ell^2N\ge 0\, ,
\ee
there is a special chamber in the $(\tau_2,\sigma_2,z_2)$ space, known as the
attractor chamber, 
where the contribution from two-centered black holes vanish and we get single-centered
black hole index $d^*(m,n,\ell)$ carrying charges $(m,n,\ell)$. 
The precise expression for $d^*(m,n,\ell)$ is \cite{Cheng:2007ch}
\begin{equation}\label{edstarTint}
 d^{*}(m,n,\ell) =\left\{ \begin{split}
&\frac{(-1)^{\ell+1}}{N} \int_{\mathcal{C}_{m,n,\ell}} d\tau d\sigma dz\, e^{-2 \pi i\left(m\tau+n\sigma/N+\ell z\right)} \frac{1}{\Phi_{k}(\Omega)}~, \ \  \begin{array}{cc}
 m\ge 0, n\ge 0~,\\ 4mn-\ell^2N\ge 0~,     
\end{array} \\
& 0~, \qquad \hbox{\rm otherwise}~, \,
    \end{split}\right. 
\end{equation}
where the contour $\mathcal{C}_{m,n,\ell}$ is given as 
\begin{equation}
\begin{split}
\mathcal{C}_{m,n,\ell}:\quad 
\tau_2=\frac{2n}{N\varepsilon}, \quad \sigma_2=\frac{2m}{\varepsilon}, \quad z_2=-\frac{\ell}{\varepsilon},\qquad
0 \leq \tau_1,z_1<1~,\quad 0\leq \sigma_1<N~.
\end{split}
\label{eq 3.3}
\end{equation}
Here $\varepsilon>0$ is a real positive number, sufficiently small so that  
$\det \rm Im \, \Omega$ is 
larger than $1/4$.  
Following the same arguments as in \cite{Bhand:2025ghn},  for zero discriminant states satisfying $4mn-\ell^2N=0$, $m,n\ge 0$,
one can deform the contour to
\be \label{edeformedcontour}
\tau_2=\frac{2n+\varepsilon_1}{N\varepsilon}, \quad \sigma_2=\frac{2m+\varepsilon_2}{\varepsilon}, \quad 
z_2=-\frac{\ell+\varepsilon_3}{\varepsilon}, \qquad 0<|\varepsilon_1|,|\varepsilon_2|,|\varepsilon_3|\ll 1\, ,
\ee
for appropriately chosen small parameters $\varepsilon,\varepsilon_1,\varepsilon_2,\varepsilon_3$ to define $d^*(m,n,\ell)$ using \refb{edstarTint} unambiguously. 
Let us introduce the notation
\be\label{eq:T_def}
T  := \begin{pmatrix}
    m & \ell/2\\ \ell/2 & n/N
\end{pmatrix}\, .
\ee
Equation \refb{edstarTint} can be written as
\begin{equation}\label{edstarTintOmega}
d^{*}(T) =\left\{\begin{split}
    &\frac{(-1)^{2T_{12}+1}}{N} \int_{\CC_T} 
d^3( {\rm Re}\,\Omega)\, e^{-2 \pi i{\rm Tr} (\Omega T)} \frac{1}{\Phi_{k}(\Omega)}~,\quad 
\hbox{for $T\ge 0$}~, \\
&0~, \qquad \hbox{otherwise}\, ,
    \end{split}\right. 
\end{equation}
where $\CC_T$ is defined in \eqref{eq 3.3} with the identification \eqref{eq:T_def} and $T\ge 0$ means that $T$ has non-negative eigenvalues. Using the symmetries of $\Phi_k$, one can show that 
% It follows from \refb{edstarTintOmega} and the identities
% \be
% \Phi_{k} (\gamma\Omega \gamma^t) = \Phi_{k}(\Omega), \quad
% \gamma \gamma_0 \gamma^t=\gamma_0, \quad \gamma^t \gamma_0 \gamma=\gamma_0~,\quad d^3(\mr{Re}\,\Omega)=d^3(\mr{Re}\,\gamma\Omega\gamma^t)~, 
% \quad \hbox{for} \ \gamma\in \Gamma_1(N)\, ,
% \ee
$d^*$ is invariant under an $\Gamma_1(N)$ transformation \cite{Jatkar:2005bh}: 
\be\label{esl2zT}
d^*(\gamma^t T\gamma) = d^*(T), \qquad \gamma\in \Gamma_1(N)\, .
\ee
For $\gamma=\begin{pmatrix} a & b\cr c & d\end{pmatrix}$ the transformation
$T\mapsto \gamma^t T \gamma$ takes the form:
\be\label{e2}
\begin{pmatrix}
    m\\ n\\ \ell
\end{pmatrix}\mapsto \begin{pmatrix}
    a^2 m + c^2 n/N + ac \ell\cr b^2 m + d^2 n + bdN\ell\cr
2 ab m + 2 cd n/N + (ad + bc)\ell
\end{pmatrix}\, .
\ee
Since the index $d^*(m,n,\ell)$ is counting states, it must be a positive integer \cite{Sen:2011ktd}. This precise mathematical statement about the Fourier coefficients of $\Phi_k^{-1}$ in appropriate chamber is called the Sen's conjecture. 
To study such properties, it is desirable to construct the generating function of single-centered degeneracies. The first such attempt was made in \cite{Dabholkar:2012nd}, and later generalized in \cite{Chattopadhyaya:2018xvg,Bhand:2023rhm,Banerjee:2025bqi}. It was shown that the degeneracy of single-centered black holes with a given set of charges (with some restrictions discussed in \cite{Bhand:2023rhm,Banerjee:2025bqi}) is the Fourier coefficient of a holomorphic function which special transformation properties defining what are called mock Jacobi forms \cite{Dabholkar:2012nd,bringmann2017harmonic}. The basic idea was to decompose the coefficient $\phi_m(\sigma,z)$ of $e^{2\pi im\tau}$ in the expansion of $\Phi_k^{-1}$ in a given chamber into  holomorphic and polar part $\phi_m=\phi^F_m+\phi^P_m$ and starting from the Fourier integral \eqref{edstarTint},  relate the degeneracy $d^*(m,n,\ell)$ to the appropriate Fourier coefficient of $\phi^F_m$. In physical terms, the polar part $\phi^P$ contains the contributions to the index from two-centered black holes in the chamber where $\Phi_k^{-1}$ was expanded to obtain $\phi_m$. The holomorphic part $\phi_m^F$ is the thus the natural candidate for the generating function of single-centered degeneracies for fixed electric charge $m$. The generating function in terms of mock Jacobi form was used to prove Sen's conjecture for some cases \cite{Bringmann:2012zr,Rossello:2024qyi}. On the other hand, there are also some shortcomings of the generating function in terms of mock Jacobi forms:
\begin{enumerate}
    \item The holomorphic part $\phi_m^F$  contains Fourier coefficients corresponding to negative discriminant charges, i.e., charge $m,n,\ell$ with $4mn-\ell^2N<0$, for which no single-centered black holes exist.
    \item The range of charges for which one can determine the degeneracy of single-centered black holes from $\phi_m^F$ is rather restricted \cite{Bhand:2023rhm,Banerjee:2025bqi}.
    \item The S-duality invariance \eqref{esl2zT} of single-centered degeneracy is obscure from mock Jacobi forms\footnote{Part of S-duality invariance corresponding to $\begin{pmatrix}
        1&r\\0&1
    \end{pmatrix}\in \mr{SL}(2,\IZ)$, called the spectral flow invariance follows from the elliptic transformation of mock Jacobi forms.}. 
\end{enumerate}  
All three of these issues were resolved in \cite{Bhand:2025ghn} for the CHL model corresponding to $\CM=K3,N=1$. The main idea was to subtract the contribution of \textit{all} two-centered black holes from the index in a given chamber. The main goal of this paper is to generalize the results of \cite{Bhand:2025ghn} to other CHL models. To this end,  
following \cite{Bhand:2025ghn}, we define the generating function of single-centered
black hole index by:
\be\label{edefFO}
F_k(\Omega) = N\sum_{m,n,\ell\in {\mathbb{Z}}} d^*(m,n,\ell)\, (-1)^{\ell+1}\, 
e^{2 \pi i\left(m\tau+n\sigma/N+\ell z\right)}=N\sum_T (-1)^{2 T_{12}+1}\, d^*(T) e^{2\pi i {\rm Tr}(\Omega T)}\, .
\ee
If we can prove the absolute convergence of $F_k$, then it follows from \refb{edefFO} and \refb{esl2zT} that 
$F_k(\Omega)$ is also $\Gamma_1(N)$ invariant:
\be
F_k(\gamma\Omega\gamma^t)=F_k(\Omega)\, ,\quad \gamma\in\Gamma_1(N)~.
\ee
The convergence of a subseries of $F_k(\Omega)$ where $m,n,\ell$ grow together is identical to the proof in \cite{Bhand:2025ghn} using the known growth of $d^*(m,n,\ell)$ \cite{Jatkar:2005bh,Sen:2007qy}. 
\par 
As discussed above, we can also define the single-centered index by subtracting the index of two-centered black holes from the total index computed from \refb{edmnl}. Let $d(m,n,\ell)$ be the total index in a given chamber in the moduli space and $d^{\text{two}}(m,n,\ell)$ denote the contribution to the index from
two-centered black holes in the same
chamber of the moduli space. Then we can define the single-centered index by:
\begin{eqsp}\label{eq:deg_multi_single}
\wt d^*(m,n,\ell)=d(m,n,\ell)-d^{\text{two}}(m,n,\ell)~.    
\end{eqsp}
Let us write the generating function for $\wt d^*$ as
\be\label{edeftildeF}
\wt F_k(\Omega) = N\sum_T \, (-1)^{2 T_{12}+1} \, \wt d^*(T) \, e^{2\pi i \mr{Tr}(\Omega T)}\, .
\ee
By analyzing the bound state metamorphosis \cite{Dabholkar:2007vk,Sen:2007pg,Cheng:2007ch,1104.1498,1210.4385,
Chowdhury:2019mnb,LopesCardoso:2020pmp} of two-centered black holes, we derive the generating function $\wt F_k(\Omega)$. We get
the following expression for $\wt F_k(\Omega)$:
\begin{eqsp}
\label{eguessfinintro}
\wt F_k(\Omega) &= {1\over \Phi_{k}(\Omega)}
-\varepsilon_N \sum_{\big{(}\begin{smallmatrix} a & b\cr c & d\end{smallmatrix}\big{)}\in 
\mr{P}\Gamma_1(N)}
\left(e^{\pi i \{ac\tau + bd\sigma + (ad+bc)z\}} - e^{-\pi i \{ac\tau + bd\sigma + (ad+bc)z\}}\right)^{-2}  \\ & 
\hskip 1in 
\times\ f_+(a^2\tau +b^2\sigma +2abz) \ g_+(c^2\tau+d^2\sigma+2cd z) \\ 
&-\varepsilon_N 
\sum_{\big{(}\begin{smallmatrix} a & b\cr c & d\end{smallmatrix}\big{)}\in \mr{P}\Gamma_1(N)}
\sum_{r>0}  g_{-1} r  H(ac\tau_2 + bd\sigma_2 + (ad+bc)z_2) \\ &\hskip 1in \times
H\left(-ac\tau_2 - bd\sigma_2 - (ad+bc)z_2 + rNa^2\tau_2 + rNb^2\sigma_2+2rNab z_2
\right) \,
\\&\hskip 1in  \times 
e^{2\pi i \{-(c^2\tau+d^2\sigma+2cdz)/N+r( ac\tau+bd\sigma+(ad+bc)z )\}}f_+(a^2\tau+b^2\sigma+2abz) 
\\ &-\varepsilon_N 
\sum_{\big{(}\begin{smallmatrix} a & b\cr c & d\end{smallmatrix}\big{)}\in \mr{P}\Gamma_1(N)}
\sum_{r>0}  f_{-1} r  H(ac\tau_2 + bd\sigma_2 + (ad+bc)z_2) \\ &\hskip 1in \times\
H\left(-ac\tau_2 - bd\sigma_2 - (ad+bc)z_2 - rc^2\tau_2 + rd^2\sigma_2+2rcd z_2
\right) \,
\\&\hskip 1in  \times 
e^{2\pi i \{-(a^2\tau+b^2\sigma+2abz)+r( ac\tau+bd\sigma+(ad+bc)z )\}}g_+(c^2\tau+d^2\sigma+2cdz) \\ 
&-  f_{-1}g_{-1} \sum_{r>0} r\, 
\sum_{\big{(}\begin{smallmatrix} a & b\cr c & d\end{smallmatrix}\big{)}\in G^N_r\backslash \mr{P}\Gamma_1(N)^+}\hskip .1in 
 \bigg\{\prod_{n=-\infty}^\infty 
H(a_nc_n\tau_2 + b_nd_n\sigma_2 + (a_nd_n+b_nc_n)z_2)
\bigg\}
\\ & \hskip 1in \times \
e^{2\pi i \{(-a^2-c^2/N+r ac)\tau+(-b^2-d^2/N+r bd)\sigma+
(-2ab-2cd/N + r(ad+bc)) z\}}  \, .
\end{eqsp}
where
\begin{eqsp}
    \mr{P}\Gamma_1(N)=\begin{cases}
        (\Gamma_1(N)\sqcup-\Gamma_1(N))/\{\pm 1\},&N>2~,
        \\
        \Gamma_1(N)/\{\pm 1\},&N=1,2~,
    \end{cases}
\end{eqsp}
\begin{eqsp}
    \varepsilon_N=\begin{cases}
        1,&N>1~,
        \\
        \frac{1}{2},&N=1~,
    \end{cases}
\end{eqsp}
$H$ is the Heaviside function
\be
H(x) = \begin{cases} \hbox{1 for $x > 0$}\cr
\hbox{0 for $x\le 0$} \end{cases} \, ,
\ee
the $f_p,g_q$ are the Fourier coefficients of the inverse of $f^{(k)},g^{(k)}$ defined in \eqref{eq:fk-gk-def1}--\eqref{eq:fk-gk-def4}:
\begin{eqsp}\label{eq:fk-gk-exp}
   \frac{1}{f^{(k)}(\tau)}=:\sum_{p\geq -1}f_pe^{2\pi ip\tau}~,\quad  \frac{1}{g^{(k)}(\sigma)}=:\sum_{q\geq -1}g_qe^{2\pi iq\sigma/N}~,
\end{eqsp}
and the functions $f_+(\tau),g_+(\sigma)$ are defined as
\begin{eqsp}
f_+(\tau)=\sum_{p\geq 0}f_pe^{2\pi ip\tau}~,\quad  g_+(\sigma)=\sum_{q\geq 0}g_qe^{2\pi iq\sigma/N}~,    
\end{eqsp}
% \be \label{e122}
% \eta(\tau)^{-24} = \sum_{p=-1}^\infty f_p \, e^{2\pi ip \tau}, \qquad f_+(\tau) 
% = \sum_{p=0}^\infty f_p \, e^{2\pi i p\tau}\, ,
% \ee
% where $\eta$ is the Dedekind eta function, 
$G^N_r$ is the cyclic group generated by 
\be
\begin{pmatrix}0 & -1/\sqrt{N}\cr \sqrt{N} & -r\sqrt{N} \end{pmatrix}~,
\ee
$\mathrm{P}\Gamma_1(N)^+=\mathrm{P}\Gamma_1(N)\cup\gamma_N\mathrm{P}\Gamma_1(N)$, where 
\begin{eqsp}
    \gamma_N:=\begin{pmatrix}
        0&-\frac{1}{\sqrt{N}}\\\sqrt
        N&0
    \end{pmatrix}~,
\end{eqsp}
and 
\be\label{edefanbnintro}
\begin{pmatrix} a_n & b_n\cr c_n & d_n\end{pmatrix}:=\begin{pmatrix}0 & -1/\sqrt{N}\cr \sqrt{N} & -r\sqrt{N} \end{pmatrix}^n
\begin{pmatrix} a & b\cr c & d\end{pmatrix}\, .
\ee
The generating function simplifies slightly for CHL models which posses the addition S-duality symmetry corresponding to $\gamma_N$,
that is, the Siegel modular form satisfies 
\begin{eqsp}\label{eq:gamma_N_inv_Phik}
    \Phi_k(\gamma_N\Omega\gamma_N^t)=\Phi_k(\Omega)~.
\end{eqsp}
This holds for example for $\CM=K3$ \cite{David:2006yn,Persson:2015jka}. 
The generating function for such CHL model takes the form 
\begin{eqsp}
\label{eguessfinintroK3}
\wt F_k(\Omega) &= {1\over \Phi_{k}(\Omega)}
-\frac{1}{2} \sum_{\big{(}\begin{smallmatrix} a & b\cr c & d\end{smallmatrix}\big{)}\in 
\mr{P}\Gamma_1(N)^+}
\left(e^{\pi i \{ac\tau + bd\sigma + (ad+bc)z\}} - e^{-\pi i \{ac\tau + bd\sigma + (ad+bc)z\}}\right)^{-2}  \\ & 
\hskip 1in 
\times\ f_+(a^2\tau +b^2\sigma +2abz) \ g_+(c^2\tau+d^2\sigma+2cd z) \\ 
&-
\sum_{\big{(}\begin{smallmatrix} a & b\cr c & d\end{smallmatrix}\big{)}\in \mr{P}\Gamma_1(N)^+}
\sum_{r>0}  g_{-1} r  H(ac\tau_2 + bd\sigma_2 + (ad+bc)z_2) \\ &\hskip 1in \times
H\left(-ac\tau_2 - bd\sigma_2 - (ad+bc)z_2 + rNa^2\tau_2 + rNb^2\sigma_2+2rNab z_2
\right) \,
\\&\hskip 1in  \times 
e^{2\pi i \{-(c^2\tau+d^2\sigma+2cdz)/N+r( ac\tau+bd\sigma+(ad+bc)z )\}}f_+(a^2\tau+b^2\sigma+2abz) 
\\  
&-  f_{-1}g_{-1} \sum_{r>0} r\, 
\sum_{\big{(}\begin{smallmatrix} a & b\cr c & d\end{smallmatrix}\big{)}\in G^N_r\backslash \mr{P}\Gamma_1(N)^+}\hskip .1in 
 \bigg\{\prod_{n=-\infty}^\infty 
H(a_nc_n\tau_2 + b_nd_n\sigma_2 + (a_nd_n+b_nc_n)z_2)
\bigg\}
\\ & \hskip 1in \times \
e^{2\pi i \{(-a^2-c^2/N+r ac)\tau+(-b^2-d^2/N+r bd)\sigma+
(-2ab-2cd/N + r(ad+bc)) z\}}  \, .
\end{eqsp}
The main technical challenge is in proving the convergence of the series \eqref{eguessfinintro}. We present the steps in the proof of the convergence of \eqref{eguessfinintro} for general $N$ wherever possible but achieve this goal fully only for $N=2,3$. The proof for $N>3$ presents challenges which probably requires novel methods to solve. We now summarize the main results of this paper. 
\begin{enumerate}
\item In Section \ref{section3}, we provide a physical derivation of \eqref{eguessfinintro} based on bound state metamorphosis \cite{Dabholkar:2007vk,Sen:2007pg,Cheng:2007ch,1104.1498,1210.4385,Chowdhury:2019mnb}. Bound state metamorphosis has been reviewed in Appendix \ref{app:BSM}.
\item In Section \ref{stildeFconverge}, in Theorem \ref{thm:S_conv}, we show for $N=2,3$, that the sum over $a,b,c,d,r$ in each of the terms in
\refb{eguessfinintro} converges absolutely and uniformly on compact subsets of
$\IH_2$ except on the subspaces
\be\label{epoleint}
-m_1\tau + n_1\sigma + m_2 + j\, z = 0, \qquad
m_1/N,n_1,m_2,j\in \mathbb{Z}, \qquad m_1n_1  +
{j^2\over 4} = {1\over 4}\, .
\ee
\item In Section \ref{ssection5}, we show that on the subspaces \refb{epoleint} 
the sums in \refb{eguessfinintro}  have double poles and these poles precisely cancel the poles
of $1/\Phi_{k}(\Omega)$ at \refb{epoleint}.
We show that despite the presence of the Heaviside functions in its definition,
$\wt F_k(\Omega)$ admits a meromorphic continuation to $\IH_2$
with double poles at 
\begin{eqnarray}
&& n_2 (\tau\sigma - z^2) - m_1\tau + n_1\sigma + m_2 + j\, z = 0~, \nonumber \\
&&
m_1/N,n_1,m_2,n_2,j\in \mathbb{Z}, \qquad n_2\ge 1, \qquad m_1n_1 + m_2 n_2 +
{j^2\over 4} = {1\over 4}\, .
\end{eqnarray}
The behavior at these poles coincide with those of $1/\Phi_{k}(\Omega)$. 
However $1/\Phi_{k}(\Omega)$ has additional poles for 
$n_2=0$ which are absent in $\wt F_k(\Omega)$.
The proof of the fact that that the two different definitions of single-centered index and their generating functions
coincide:
\be\label{etildeFFequality}
d^*(T)=\wt d^*(T)\quad \ \forall \ T, \qquad F_k(\Omega)=\wt F_k(\Omega)\, ,
\ee
including contributions from states where $T$ has a negative eigenvalue is similar to \cite{Bhand:2025ghn}.
Since both $F_k(\Omega)$ and $\wt F_k(\Omega)$ are defined by analytic continuation from their
domains of convergence, a more precise statement will be that the analytic continuations
of $F_k(\Omega)$ and $\wt F_k(\Omega)$ agree.
\item In Section \ref{sec:conc}, we list some future directions. 
\end{enumerate}

\section{Construction of the generating function  $\wt F_k(\Omega)$} \label{section3}

%The analysis in Section \ref{sgenerating} focussed on the construction of thegenerating function of single-centered black holes directly from the coefficients$d^*(m,n,\ell)$. For different $(m,n,\ell)$, these coefficients have to be calculated by working in different chambers in the $(\tau_2,\sigma_2,z_2)$ space and then wehave to use them to build the generating function. 
Our goal in this section is
to construct the generating function of the index of single-centered black holes by taking the difference between the complete generating function for black holes 
and the generating function for two-centered black holes. We
begin with the partition function $\Phi_{k}^{-1}$ of the CHL model obtained by the $\IZ_N$ orbifold of the worldsheet CFT of type II string theory compactified on $\CM\times T^2$, where $\CM=K3,T^4$. We will work in a fixed chamber, and then subtract from it, the 
generating function of two-centered black holes in the same chamber.

To introduce $\Phi_{k}^{-1}$, we need to introduce the elliptic genus of the $\IZ_N$-orbifold of the sigma model with target $\CM$: suppose $\widetilde{g}$ generates the $\IZ_N$ action on $\CM$ which we are orbifolding, then the elliptic genus is given by 
\begin{equation}
F^{(r,s)}(\tau, z) :=\frac{1}{N} \mr{Tr}_{\mr{RR};\widetilde{g}^r} \left( \widetilde{g}^s (-1)^{F_L+F_R} e^{2\pi i \tau (L_0-\frac{c_L}{24})} e^{-2\pi i \bar{\tau} (\bar{L}_0-\frac{c_R}{24})} e^{2\pi i F_L z} \right), \quad 0 \leq r,s \leq N-1~,
\end{equation}
where $\mr{Tr}_{\mr{RR};\widetilde{g}^r}$ denotes trace over all the RR sector states twisted by $\wt g^r$ before we project on to $\wt g^s$ invariant states, $L_0,\bar{L}_0$ denotes the left and rightmoving Virasoro
generators with central charge $c_L,c_R$ respectively, and $F_L$ and $F_R$ denote the worldsheet fermion numbers associated with left and rightmoving sectors in this SCFT. Due
to the insertion of $(-1)^{F_R}$ factor in the trace the contribution to $F^{(r,s)}$
comes only from the $\bar{L}_0-c_R/24=0$
states. As a result $F^{(r,s)}$ does not depend on $\overline{\tau}$. Let us expand the elliptic genus as 
\begin{eqsp}
    F^{(r, s)}(\tau, z) = \sum_{b=0}^1 \sum_{\substack{\ell \in 2 \mathbb{Z}+b \\ n \in \mathbb{Z} / N}} c_b^{(r, s)}\left(4 n-\ell^2\right) q^n \zeta^{\ell}, \quad q=e^{2\pi i \tau},\quad \zeta=e^{2\pi i z}~,
\end{eqsp}
with the coefficients $c_b^{(r, s)}$ taking different values for different CHL models. The partition function is given by \cite{David:2006ji,Sen:2007qy}
\begin{equation}
\begin{aligned}
& \Phi_k(\tau, z, \sigma)=e^{2 \pi i(\widetilde{\alpha} \tau+\widetilde{\gamma} \sigma+z)} \\
& \quad \times \prod_{b=0}^1 \prod_{r=0}^{N-1} \prod_{\substack{k^{\prime} \in \mathbb{Z}+\frac{r}{N}, j \in 2 \IZ+b, l \in \mathbb{Z} \\
k^{\prime}, l \geq 0, j<0 \text { for } k^{\prime}=l=0}}\left(1-\exp \left(2 \pi i\left(k^{\prime} \sigma+l \tau+j z\right)\right)\right)^{\sum_{s=0}^{N-1} e^{-2 \pi i s l / N} c_b^{(r, s)}\left(4 k^{\prime} l -j^2\right)}~,
\end{aligned}
\end{equation}
where
\begin{equation}
\widetilde{\alpha}:=\frac{1}{24 N} Q_{0,0}-\frac{1}{2 N} \sum_{s=1}^N Q_{0, s} \frac{e^{2 \pi i s / N}}{\left(1-e^{2 \pi i s / N}\right)^2}, \quad \widetilde{\gamma}:=\frac{1}{24 N} Q_{0,0}=\frac{1}{24 N} \chi(\mathcal{M})
\end{equation}
with $Q_{r, s}$ defined by
\begin{equation}
Q_{r, s}:=N\left(c_0^{(r, s)}(0)+2 c_1^{(r, s)}(-1)\right)~ .
\end{equation}
Here $\chi(\mathcal{M})$ denotes the Euler number of the manifold $\mathcal{M}$. One can show that $\tilde{\alpha}, N \tilde{\gamma} \in \mathbb{Z}$ \cite{Sen:2007qy}.
The function $\Phi_k$ transforms as a Siegel modular form of weight $k$ for a subgroup $G$ of $\mr{Sp}(2,\IZ)$. Explicitly, $k$ is given by
\begin{equation}
k=\frac{1}{2} \sum_{s=0}^{N-1} c_0^{(0, s)}(0)~.
\end{equation}
For $\CM=K3$, $k$ has a simple formula in terms of $N$ for $N\leq 8$:
\begin{equation}
    k=\begin{cases}
        \frac{24}{N+1}-2,&N=1,2,3,5,7;\\3,&N=4;\\2,& N=6;\\1,& N=8~.
    \end{cases}
\end{equation}
Note that for $N=1$ and $\CM=K3$, the CHL partition function $\Phi_k^{-1}$ is exactly the inverse of the Igusa cusp form $\Phi_{10}$ and was analyzed in detail in \cite{Bhand:2025ghn}.

$1/\Phi_{k}(\Omega)$ has double poles at \cite{LopesCardoso:2004law,Sen:2007qy} 
\be \label{epoles}
n_2(\tau\sigma-z^2) - m_1\tau +n_1\sigma + jz
+ m_2=0, \quad
m_1/N,n_1,m_2,n_2,j\in \IZ,
\quad m_1 n_1+m_2n_2 + {j^2\over 4}={1\over 4}\, .
\ee
% For $n_2\neq 0$, we can rewrite the hypersurface \eqref{epoles} as 
% \begin{eqsp}\label{eq:poles_new}
%     n_2\left[\left(\tau+\frac{n_1}{n_2}\right)\left(\sigma-\frac{m_1}{n_2}\right)-\left(z-\frac{j}{2n_2}\right)^2\right]+\frac{1}{4n_2}=0~.
% \end{eqsp}
% Let us write 
% \begin{eqsp}
%     \tilde{\tau}_1=\tau_1+\frac{n_1}{n_2}~,\quad \tilde{\sigma}_1=\sigma_1-\frac{m_1}{n_2}~,\quad \tilde{z}_1=z_1-\frac{j}{2n_2}~.
% \end{eqsp}
% Then the imaginary part of \eqref{eq:poles_new} gives 
% \begin{eqsp}\label{eq:im_poles_new}
%     \tilde{\tau}_1=\frac{2\tilde{z}_1z_2-\tau_2\tilde{\sigma}_1}{\sigma_2}~,
% \end{eqsp}
% and the real part of \eqref{eq:poles_new} gives 
% \begin{eqsp}\label{eq:re_poles_new}
%     n_2\left[\tilde{\tau}_1\tilde{\sigma}_1-\tau_2\sigma_2-\tilde{z}_1^2+z_2^2\right]+\frac{1}{4n_2}=0~.
% \end{eqsp}
% Combining \eqref{eq:im_poles_new} and \eqref{eq:re_poles_new} we get 
% \begin{eqsp}\label{eq:im_re_poles_new}
% {\tau_2\over \sigma_2}\tilde{\sigma}_1^2 + \tilde{z}_1^2 - 2\, {z_2\over \sigma_2}\,
% \tilde{\sigma}_1 \tilde{z}_1 = {1\over 4n_2^2} - (\tau_2\sigma_2 - z_2^2)~.    
% \end{eqsp}
% The left hand side of \eqref{eq:im_re_poles_new}
% is a positive definite quadratic form in the variables $\tilde{\sigma}_1,\tilde{\tau}_1$ for $\tau_2\sigma_2>z_2^2$. Therefore,
% a solution to these equations exists only when the right hand side is positive, i.e.
% when $\det \mathrm{Im} \,\Omega= \tau_2\sigma_2-z_2^2 \le 1/4n_2^2$. In particular, 
Following the same argument as in \cite{Bhand:2025ghn}, a sufficient condition for absence of poles with $n_2\neq 0$ is $\det \mathrm{Im} \,\Omega= \tau_2\sigma_2-z_2^2 > 1/4$.

This shows that, for  $\det\rm Im\,\Omega > 1/4$ we
only need to examine the poles corresponding to $n_2=0$ in \refb{epoles}. 
One of these, corresponding
to $m_1=m_2=n_1=0$, $j=1$, is the pole at $z=0$.
Near the double pole $z\to 0$, we have
\be\label{epolestructure}
{1\over \Phi_{k}(\Omega)} = \frac{(e^{\pi i z}- e^{-\pi i z})^{-2}}{f^{(k)}(\tau) g^{(k)}(\sigma)}\, ,
\ee
where 
\begin{align}\label{eq:fk-gk-def1}
f^{(k)}(\tau)&=e^{2 \pi i \widetilde{\alpha} \tau} \prod_{r=1}^{\infty}\left(1-e^{2 \pi i r \tau}\right)^{s_r}~, \\
g^{(k)}(\sigma)&=e^{2 \pi i \widetilde{\gamma} \sigma} \prod_{r=0}^{N-1} \prod_{\substack{k^{\prime} \in \mathbb{Z}+r / N \\
k^{\prime}>0}}\left(1-e^{2 \pi i k^{\prime} \sigma}\right)^{t_r}~, \\
s_r=\frac{1}{N} \sum_{s^{\prime}=0}^{N-1} e^{-2 \pi i r s^{\prime} / N} Q_{0, s^{\prime}}&=\sum_{s^{\prime}=0}^{N-1} e^{-2 \pi i r s^{\prime} / N}\left(c_0^{(0, s^{\prime})}(0)+2 c_1^{(0, s^{\prime})}(-1)\right)~, \\\label{eq:fk-gk-def4}
t_r=\frac{1}{N} \sum_{s=0}^{N-1} Q_{r, s}&=\sum_{s=0}^{N-1}\left(c_0^{(r, s)}(0)+2 c_1^{(r, s)}(-1)\right)~ .
\end{align}
The functions $f^{(k)},g^{(k)}$ have the interpretation of being the inverse of the partition function of the purely magnetic and purely electric half-BPS states respectively. We will expand these functions as in \eqref{eq:fk-gk-exp}.

The behavior of $1/\Phi_{k}$ near other
double poles of the type given in \refb{epoles} with $n_2=0$ is obtained by $\Gamma_1(N)$-transformation of \refb{epolestructure} given by $\Omega\to\gamma\Omega\gamma^t$ since one can
show, by arguments analogous to \cite[Lemma 5.1]{Bhand:2025ghn}, that every other pole for $n_2=0$ can be related to the pole at $z=0$ using
a $\Gamma_1(N)$-transformation. 
The Fourier coefficients of the RHS of \eqref{epolestructure} in a given chamber are known to give the contribution of two-centered
black holes to the index and are responsible for the jump in the index
across the walls of marginal stability \cite{Sen:2007pg,Cheng:2007ch,Sen:2007qy}. 
Thus the naive expectation would be that to get the generating function of
single-centered index, we should remove from $1/\Phi_{k}$ the contribution on the right
hand side of \refb{epolestructure} and its $\Gamma_1(N)$-images.
The naive guess for the generating function is then:
\begin{eqnarray}
\label{eguess}
&& {1\over \Phi_{k}(\Omega)}
-\varepsilon_N \sum_{\big{(}\begin{smallmatrix} a & b\cr c & d\end{smallmatrix}\big{)}\in \mr{P}\Gamma_1(N)}
\left(e^{\pi i \{ac\tau + bd\sigma + (ad+bc)z\}} - e^{-\pi i \{ac\tau + bd\sigma + (ad+bc)z\}}\right)^{-2} \nonumber \\ && 
\hskip 1in 
\times\ f^{(k)}(a^2\tau +b^2\sigma +2abz) \ g^{(k)}(c^2\tau+d^2\sigma+2cd z) \nonumber \\ 
&=&  {1\over \Phi_{k}(\Omega)}
- \varepsilon_N\sum_{\big{(}\begin{smallmatrix} a & b\cr c & d\end{smallmatrix}\big{)}\in \mr{P}\Gamma_1(N)}
\left(e^{\pi i \{ac\tau + bd\sigma + (ad+bc)z\}} - e^{-\pi i \{ac\tau + bd\sigma + (ad+bc)z\}}\right)^{-2}  \nonumber \\ &&
\hskip 1in \times \ \sum_{p,q=-1}^\infty f_p g_q \,
 e^{2\pi i \, p\{a^2\tau +b^2\sigma +2abz\}} \,  
e^{2\pi i \frac{q}{N}\{c^2\tau+d^2\sigma+2cd z\}} \, ,
\end{eqnarray}
where 
\begin{eqsp}
    \varepsilon_N:=1-\frac{1}{2}\delta_{N,1}=\begin{cases}
        1,&N>1~,\\\frac{1}{2},&N=1~,    \end{cases}
\end{eqsp}
for $N=1$ compensates for the fact that
the summand in (\ref{eguess}) remains 
invariant under $(a,b,c,d)\to (-c,-d,a,b)$ and hence each term effectively appears
twice.  
%$[~]_F$ involves certain restriction on the terms in the sum that will be explained below. 
% It is easy to see however that this cannot be the correct result.
% For example, if we take $(a,b,c,d)=(1,0,c,1)$ with $c\in\mathbb{Z}$,
% then for $q=-1$, the term in the third and fourth line of \refb{eguess} will involve a sum of the form:
% \be \label{e314}
% \sum_{c\in\mathbb{Z}} (e^{\pi i(c\tau+z)}-e^{-\pi i (c\tau+z)})^{-2} 
% e^{-2\pi i (c^2\tau+\sigma + 2cz)}\,e^{2\pi ip\tau}\, .
% \ee
% Since $\tau$ has positive imaginary part, the sum over $c$ will diverge. 
% As we shall see, the correct result will not suffer from such divergences.
For the same reason as explained in \cite{Bhand:2025ghn}, this cannot be the correct generating function since the sum diverges. 

We shall now carry out a more detailed analysis on why \refb{eguess} might 
differ from the partition function of two-centered black hole states.
First consider the coefficient
of the terms involving $f_p f_q$ for $p,q\ge 0$. By expanding 
$\left(e^{\pi i \{ac\tau + bd\sigma + (ad+bc)z\}} - e^{-\pi i \{ac\tau + bd\sigma + (ad+bc)z\}}\right)^{-2}$
in a power series expansion in $e^{\pm\pi i \{ac\tau + bd\sigma + (ad+bc)z\}}$ in the respective domains
of convergence, we can express the term subtracted from $1/\Phi_{k}$ in \eqref{eguess} as
\begin{eqnarray}\label{e216xy}
&& -
\varepsilon_N\sum_{\big{(}\begin{smallmatrix} a & b\cr c & d\end{smallmatrix}\big{)}\in \mr{P}\Gamma_1(N)} 
 \sum_{r>0} r \bigg\{
e^{2\pi i r \{ac\tau + bd\sigma + (ad+bc)z\}} H(ac\tau_2 + bd\sigma_2 + (ad+bc)z_2) 
\nonumber \\ && \hskip 1in +\
e^{-2\pi i r (ac\tau + bd\sigma + (ad+bc)z)} H(-(ac\tau_2 + bd\sigma_2 + (ad+bc)z_2))\bigg\} 
\nonumber \\ && \hskip 1in \times\sum_{p,q\ge 0}f_p g_q e^{2\pi i p(a^2\tau+b^2\sigma+2abz)} e^{2\pi i \frac{q}{N}(c^2\tau+d^2\sigma+2cdz)}
~. 
\end{eqnarray} 

This can be reorganized as
\begin{eqnarray} \label{ereorg}
&& -
\varepsilon_N\sum_{\big{(}\begin{smallmatrix} a & b\cr c & d\end{smallmatrix}\big{)}\in \mr{P}\Gamma_1(N)} 
\sum_{r>0}\sum_{p,q\ge 0} r  H(ac\tau_2 + bd\sigma_2 + (ad+bc)z_2) \,
f_p q_q \nonumber \\ && \hskip 1in \times 
e^{2\pi i \{(pa^2+qc^2/N+r ac)\tau+(pb^2+qd^2/N+r bd)\sigma+
(2pab+2qcd/N + r(ad+bc) )z\}}  \nonumber \\ &&
-
\varepsilon_N\sum_{\big{(}\begin{smallmatrix} a & b\cr c & d\end{smallmatrix}\big{)}\in \mr{P}\Gamma_1(N)}
\sum_{r>0} \sum_{p,q\ge 0}r  H(-ac\tau_2 - bd\sigma_2 - (ad+bc)z_2) \,
f_p g_q \nonumber \\ && \hskip 1in \times 
e^{2\pi i \{(pa^2+qc^2/N-r ac)\tau+(pb^2+qd^2/N-r bd)\sigma+
(2pab+2qcd/N - r(ad+bc) )z\}} \, .
\end{eqnarray}
The term in the first two lines of \refb{ereorg} may be identified as the subtraction of the
contribution from a bound state of the
half-BPS states carrying charges $(b\bM, a\bM)$ and $(d\bN,c\bN)$, with
\be
\bM^2=2p, \quad \bN^2=2q/N, \quad \bM\cdot \bN=r, \quad r>0\, ,
\ee
so that the total charge $(Q,P)=(b\bM+d\bN,a\bM+c\bN)$
satisfies
\be
P^2 = 2 (pa^2+qc^2/N+r ac), \quad Q^2 = 2(pb^2+qd^2/N+r bd), \quad Q\cdot P =2pab+2qcd/N + r(ad+bc)\, .
\ee
Such bound states are known to exist in the chamber $ac\tau_2 + bd\sigma_2 + (ad+bc)z_2>0$ for
$r>0$ \cite{1104.1498}. Similarly the term in the last two lines of \refb{ereorg} 
may be identified as representing the subtraction of a bound state of the
half BPS states carrying charges $(b\bM, a\bM)$ and $(d\bN,c\bN)$, with
\be
\bM^2=2p, \quad \bN^2=2q/N, \quad \bM\cdot \bN=-r, \quad r>0\, ,
\ee
so that the total charge $(Q,P)=(b\bM+d\bN,a\bM+c\bN)$
satisfies
\be
P^2 = 2 (pa^2+qc^2/N-r ac), \quad Q^2 = 2(pb^2+qd^2/N-r bd), \quad Q\cdot P =(2pab+2qcd/N - r(ad+bc)\, .
\ee
Such bound states are known to exist in the chamber $ac\tau_2 + bd\sigma_2 + (ad+bc)z_2<0$ for
$r>0$ \cite{1104.1498}. Thus these subtraction terms remove from the Fourier expansion of $1/\Phi_{k}$
the contributions from two-centered bound states carrying such charges in any chamber of the
moduli space. Hence for these terms, \eqref{e216xy} gives the correct subtraction.

The terms corresponding to $(p,q)=(-1,0),(0,-1),(-1,-1)$ require special attention. Let us first consider the term proportional to $f_p g_{-1}$ with $p\ge 0$ and $f_{-1}g_q$ term with $q\geq 0$. In this case, %if we ignore the $[~]_F$operation, then,
following the same steps that led to \refb{ereorg}, 
the subtraction terms in \eqref{eguess} can be organized as:
\begin{align} \label{ereorgtwo}
&- 
\varepsilon_N\sum_{\big{(}\begin{smallmatrix} a & b\cr c & d\end{smallmatrix}\big{)}\in \mr{P}\Gamma_1(N)} 
\sum_{r>0} \sum_{p\ge 0} f_p g_{-1}r  H(ac\tau_2 + bd\sigma_2 + (ad+bc)z_2) \,
\\\label{ereorgtwo2}& \hskip 1in \times 
e^{2\pi i \{(pa^2-c^2/N+r ac)\tau+(pb^2-d^2/N+r bd)\sigma+
(2pab-2cd/N + r(ad+bc) )z\}}  
\\\label{ereorgtwo3}
&- 
\varepsilon_N\sum_{\big{(}\begin{smallmatrix} a & b\cr c & d\end{smallmatrix}\big{)}\in \mr{P}\Gamma_1(N)} 
\sum_{r>0} \sum_{q\ge 0} f_{-1} g_{q}r  H(ac\tau_2 + bd\sigma_2 + (ad+bc)z_2) \,
\\\label{ereorgtwo4} & \hskip 1in \times 
e^{2\pi i \{(-a^2+qc^2/N+r ac)\tau+(-b^2+qd^2/N+r bd)\sigma+
(-2ab+2qcd/N + r(ad+bc) )z\}}  
\\\label{ereorgtwo5} 
&- 
\varepsilon_N\sum_{\big{(}\begin{smallmatrix} a & b\cr c & d\end{smallmatrix}\big{)}\in \mr{P}\Gamma_1(N)}
\sum_{r>0} \sum_{p\ge 0}f_p g_{-1}r  H(-ac\tau_2 - bd\sigma_2 - (ad+bc)z_2) \,
\\\label{ereorgtwo6} & \hskip 1in \times 
e^{2\pi i \{(pa^2-c^2/N-r ac)\tau+(pb^2-d^2/N-r bd)\sigma+
(2pab-2cd/N - r(ad+bc) )z\}}
\\\label{ereorgtwo7}
&- 
\varepsilon_N\sum_{\big{(}\begin{smallmatrix} a & b\cr c & d\end{smallmatrix}\big{)}\in \mr{P}\Gamma_1(N)}
\sum_{r>0} \sum_{q\ge 0}f_{-1} g_{q}r  H(-ac\tau_2 - bd\sigma_2 - (ad+bc)z_2) \,
\\\label{ereorgtwo8}& \hskip 1in \times 
e^{2\pi i \{(-a^2+qc^2/N-r ac)\tau+(-b^2+qd^2/N-r bd)\sigma+
(-2ab+2qcd/N - r(ad+bc) )z\}}\, .
\end{align}
% Note that we have used the $(a,b,c,d)\to (c,d,-a,-b)$, $p\leftrightarrow q$ 
% symmetry to remove the overall
% factor of $1/2$ and drop the $f_{-1}f_q$ term.
The terms in \eqref{ereorgtwo}--\eqref{ereorgtwo8} have the same interpretation as in the case of the
terms proportional to $f_p g_{q}$ with $(p,q)$ replaced by $(-1,0),(0,-1)$, 
namely they represent subtraction of the contribution
from the bound states of charges
$(b\bM, a\bM)$ and $(d\bN,c\bN)$, with
\be\label{eq:p,q=-1MN}
\bM^2=2p, \quad \bN^2=-2/N, \quad \bM\cdot \bN=\pm r, \quad r>0\, ,
\ee
and 
\be\label{eq:p=-1,qMN}
\bM^2=-2, \quad \bN^2=2q/N, \quad \bM\cdot \bN=\pm r, \quad r>0\, ,
\ee
so that the total charge $(Q,P)=(b\bM+d\bN,a\bM+c\bN)$
satisfies
\be\label{ebcount}
P^2 = 2 (pa^2-c^2/N\pm r ac), \quad Q^2 = 2(pb^2-d^2/N\pm r bd), \quad Q\cdot P =2pab-2cd/N \pm r(ad+bc)\, ,
\ee
in the case of \eqref{eq:p,q=-1MN} and 
\be\label{ebcount}
P^2 = 2 (-a^2+qc^2/N\pm r ac), \quad Q^2 = 2(-b^2+qd^2/N\pm r bd), \quad Q\cdot P =-2ab+2qcd/N \pm r(ad+bc)\, ,
\ee
in the case of \eqref{eq:p=-1,qMN}.
However, it was shown in \cite{1104.1498,1210.4385} that 
this overcounts the bound state contribution, this has been reviewed in Appendix \ref{app:BSM}. To see this, let us
make a change of variable $(a,b,c,d)\to (a,b,c-raN,d-rbN)$ to express 
\eqref{ereorgtwo5} and \eqref{ereorgtwo6} as
\begin{eqsp} \label{emetamor}
&- 
\varepsilon_N\sum_{\big{(}\begin{smallmatrix} a & b\cr c & d\end{smallmatrix}\big{)}\in \mr{P}\Gamma_1(N)}
\sum_{r>0} \sum_{p\ge 0}f_p g_{-1} r  H\left(-ac\tau_2 - bd\sigma_2 - (ad+bc)z_2 + rNa^2\tau_2 + rNb^2\sigma_2+2rNab z_2
\right) \,
\\ & \hskip 1in \times 
e^{2\pi i \{(pa^2-c^2/N+r ac)\tau+(pb^2-d^2/N+r bd)\sigma+
(2pab-2cd/N + r(ad+bc)) z\}}\, .
\end{eqsp}
We now note that the exponent in the summand has exactly the same form as the exponent in \refb{ereorgtwo2} and hence represent the same charges.
Naively these two contributions will get added. However, it was shown in 
\cite{1104.1498,1210.4385} that these two contributions should be regarded as coming
from the same bound states and 
furthermore, for given $(\tau_2,\sigma_2,z_2)$, the bound state 
exists only when the Heaviside functions appearing in \eqref{emetamor} and \eqref{ereorgtwo} are both non-zero. Similarly, changing $(a,b,c,d)\to (a-cr,b-dr,c,d)$ to \eqref{ereorgtwo7} and \eqref{ereorgtwo8}, we get 
\begin{eqsp}
    - 
\varepsilon_N\sum_{\big{(}\begin{smallmatrix} a & b\cr c & d\end{smallmatrix}\big{)}\in \mr{P}\Gamma_1(N)}
\sum_{r>0} \sum_{q\ge 0}&f_{-1} g_{q}r  H(-ac\tau_2 - bd\sigma_2 - (ad+bc)z_2+rc^2\tau_2+rd^2\sigma_2+2rcdz_2) \,
\\& \times 
e^{2\pi i \{(-a^2+qc^2/N+r ac)\tau+(-b^2+qd^2/N+r bd)\sigma+
(-2ab+2qcd/N + r(ad+bc) )z\}}\, ,
\end{eqsp}
which gives the same contribution as \eqref{ereorgtwo3} and \eqref{ereorgtwo4}. 
Thus \refb{ereorgtwo}--\eqref{ereorgtwo8} should be
replaced by,
\begin{eqsp} \label{erewritetwo}
- 
\varepsilon_N\sum_{\big{(}\begin{smallmatrix} a & b\cr c & d\end{smallmatrix}\big{)}\in \mr{P}\Gamma_1(N)}
\sum_{r>0}&  g_{-1} r  H(ac\tau_2 + bd\sigma_2 + (ad+bc)z_2) \\ &\times\
H\left(-ac\tau_2 - bd\sigma_2 - (ad+bc)z_2 + rNa^2\tau_2 + rNb^2\sigma_2+2rNab z_2
\right) \,
\\& \times 
e^{2\pi i \{-(c^2\tau+d^2\sigma+2cdz)/N+r( ac\tau+bd\sigma+(ad+bc)z )\}}f_+(a^2\tau+b^2\sigma+2abz) 
\\ - 
\varepsilon_N\sum_{\big{(}\begin{smallmatrix} a & b\cr c & d\end{smallmatrix}\big{)}\in \mr{P}\Gamma_1(N)}
\sum_{r>0}&  f_{-1} r  H(ac\tau_2 + bd\sigma_2 + (ad+bc)z_2) \\ &\times\
H\left(-ac\tau_2 - bd\sigma_2 - (ad+bc)z_2 - rc^2\tau_2 + rd^2\sigma_2+2rcd z_2
\right) \,
\\& \times 
e^{2\pi i \{-(a^2\tau+b^2\sigma+2abz)+r( ac\tau+bd\sigma+(ad+bc)z )\}}g_+(c^2\tau+d^2\sigma+2cdz) .
\end{eqsp}
% We can verify that in this form, the divergence of the type mentioned in  \refb{e314}
% disappears. Indeed, if we set $(a,b,c,d)=(1,0,c,1)$, the sum in 
% \refb{erewritetwo} takes the form
% \be \label{erewritecheck}
% - 
% \sum_{c\in \IZ} 
% \sum_{r>0} \sum_{p\ge 0} r   \,
% f_p f_{-1} e^{2\pi i \{(p-c^2+r c)\tau-\sigma+
% (-2c + r)z\}} 
% H(c\tau_2  + z_2)  H\left(-c\tau_2 -z_2 +
% r\tau_2 \right)
% \, . 
% \ee
% We now see that for fixed $r$, the sum over $c$ is restricted, and hence the
% sum over $c$ no longer diverges.
In Section \ref{stildeFconverge} we shall prove the convergence of the 
full sum over $a,b,c,d$, $r$ for $N=2,3$.
% With a little more effort one can also show the absence of divergence
% in the sum over $r$. For definiteness let us suppose that we have take $|z_2|\ll\tau_2$. In that 
% case, if 
% $z_2>0$ then the sum over $c$ is restricted to the range $0\le c\le (r-1)$ and the dominant 
% contribution to the summand, coming from the $c=(r-1)$ term, has the form
% \be
% e^{2\pi i \{(p+r-1)\tau-\sigma+
% (2 - r)z\}}\, .
% \ee
% The sum over $r$ clearly converges since $\tau_2\gg|z_2|>0$. On the other hand, if $z_2<0$, then
% the sum over $c$ in \refb{erewritecheck} is restricted to the range $1\le c\le r$. Now the
% dominant contribution to the summand comes from the $c=r$ term, and the summand is
% proportional to
% \be
% e^{2\pi i (p\tau-\sigma
% - rz)}\, .
% \ee
% The sum over $r$ still remains convergent since now $z_2<0$.

Finally, we turn to the terms proportional to $f_{-1}g_{-1}$. 
We start from 
\begin{align} \label{ef1analysis1}
-\varepsilon_Nf_{-1}g_{-1}
{\sum_{\big{(}\begin{smallmatrix} a & b\cr c & d\end{smallmatrix}\big{)}\in \mr{P}\Gamma_1(N)}}
&\sum_{r>0} r  H(ac\tau_2 + bd\sigma_2 + (ad+bc)z_2) \,
\\\label{ef1analysis2} &\times 
e^{2\pi i \{(-a^2-c^2/N+r ac)\tau+(-b^2-d^2/N+r bd)\sigma+
(-2ab-2cd/N + r(ad+bc) )z\}}  \\\label{ef1analysis3} 
-\varepsilon_N f_{-1}g_{-1}
{\sum_{\big{(}\begin{smallmatrix} a & b\cr c & d\end{smallmatrix}\big{)}\in \mr{P}\Gamma_1(N)}}
&\sum_{r>0} r  H(-ac\tau_2 - bd\sigma_2 - (ad+bc)z_2) \,
\\\label{ef1analysis4} &\times 
e^{2\pi i \{(-a^2-c^2/N-r ac)\tau+(-b^2-d^2/N-r bd)\sigma+
(-2ab-2cd/N - r(ad+bc) )z\}} \, .
\end{align}
Note that applying the transformation 
\begin{eqsp}
    \begin{pmatrix}
        a&b\\c&d
    \end{pmatrix}\to \begin{pmatrix}
        0&-\frac{1}{\sqrt{N}}\\\sqrt{N}&0
    \end{pmatrix}\begin{pmatrix}
        a&b\\c&d
    \end{pmatrix}=\begin{pmatrix}
        -c/\sqrt{N}&-d/\sqrt{N}\\a\sqrt{N}&b\sqrt{N}
    \end{pmatrix}~,
\end{eqsp}
to \eqref{ef1analysis1} and \eqref{ef1analysis2}, we obtain \eqref{ef1analysis3} and \eqref{ef1analysis4} respectively. Thus, we can write \eqref{ef1analysis1}--\eqref{ef1analysis4} as 
\begin{eqsp}\label{eq:f-1g-1_pre}
-f_{-1}g_{-1}
{\sum_{\big{(}\begin{smallmatrix} a & b\cr c & d\end{smallmatrix}\big{)}\in \mr{P}\Gamma_1(N)^+}}
&\sum_{r>0} r  H(ac\tau_2 + bd\sigma_2 + (ad+bc)z_2) \,
\\ &\times 
e^{2\pi i \{(-a^2-c^2/N+r ac)\tau+(-b^2-d^2/N+r bd)\sigma+
(-2ab-2cd/N + r(ad+bc) )z\}} ~,     
\end{eqsp}
where 
\begin{eqsp}
\mr{P}\Gamma_1(N)^+:=\mr{P}\Gamma_1(N)\cup \gamma_N\mr{P}\Gamma_1(N)~,   
\end{eqsp}
and 
\begin{eqsp}
    \gamma_N:=\begin{pmatrix}
        0&-\frac{1}{\sqrt{N}}\\\sqrt{N}&0
    \end{pmatrix}~.
\end{eqsp}
For $N=1$, the factor of $\varepsilon_N=\frac{1}{2}$ is removed by the fact that $\mr{P}\Gamma_1(1)^+=\mr{PSL}(2,\IZ)$.
Now, it is clear that the transformations 
\begin{eqsp}\label{eq:identification_rN}
    \begin{pmatrix}
        a&b\\c&d
    \end{pmatrix}\to \begin{pmatrix}
        0&-1/\sqrt{N}\\\sqrt{N}&0
    \end{pmatrix}\begin{pmatrix}
        1&-r\\0&1
    \end{pmatrix}\begin{pmatrix}
        a&b\\c&d
    \end{pmatrix}=\begin{pmatrix}
        0&-1/\sqrt{N}\\\sqrt{N}&-r\sqrt{N}
    \end{pmatrix}\begin{pmatrix}
        a&b\\c&d
    \end{pmatrix}~,
    \\
    \begin{pmatrix}
        a&b\\c&d
    \end{pmatrix}\to \begin{pmatrix}
        0&-1/\sqrt{N}\\\sqrt{N}&0
    \end{pmatrix}\begin{pmatrix}
        1&0\\-rN&1
    \end{pmatrix}\begin{pmatrix}
        a&b\\c&d
    \end{pmatrix}=\begin{pmatrix}
        r\sqrt{N}&-1/\sqrt{N}\\\sqrt{N}&0
    \end{pmatrix}\begin{pmatrix}
        a&b\\c&d
    \end{pmatrix}~,
\end{eqsp}
leaves the exponent in \eqref{eq:f-1g-1_pre} invariant.  
By the same argument as before, the bound states corresponding to 
matrices $(\begin{smallmatrix} a & b\cr c & d\end{smallmatrix})$ related by 
these transformations represent the same physical state and
should be counted only once \cite{1104.1498,1210.4385}, see Appendix \ref{app:BSM} for details of this identification. Thus in the sum over $a,b,c,d$
in \eqref{eq:f-1g-1_pre}, we must identify the matrices 
$(\begin{smallmatrix} a & b\cr c & d\end{smallmatrix})$ related by these transformations. 
Since the matrices multiplying $\begin{pmatrix}
    a&b\\c&d
\end{pmatrix}$ in \eqref{eq:identification_rN} are inverses of each other up to an overall multiplication by $-1$, we need to sum over all $\mr{P}\Gamma_1(N)^+$ matrices up to left multiplication by the group $G^N_r$ generated by the matrix 
\begin{eqsp}
G^N_r:=\left\langle\begin{pmatrix} 0 & -1/\sqrt{N}\\\sqrt{N}& -r\sqrt{N}\end{pmatrix}\right\rangle~.    
\end{eqsp}
Furthermore, 
these bound states exist in only those chambers of the moduli space where the 
Heaviside function in \eqref{eq:f-1g-1_pre}, and all the other Heaviside functions
related to the one in \eqref{eq:f-1g-1_pre} by replacing the matrix 
$(\begin{smallmatrix} a & b\cr c & d\end{smallmatrix})$ by an element of $G^N_r$
multiplying this matrix from the left, are non-zero. Thus
the
the Heaviside function in \eqref{eq:f-1g-1_pre} must be replaced by the product of infinite number of Heaviside functions where in  the argument we replace $(a,b,c,d)$ by $(a_n,b_n, c_n, d_n)$, where
\be\label{edefanbn}
\begin{pmatrix} a_n & b_n\cr c_n & d_n\end{pmatrix}:=\begin{pmatrix} 0 & -1/\sqrt{N}\\\sqrt{N}& -r\sqrt{N}\end{pmatrix}^n
\begin{pmatrix} a & b\cr c & d\end{pmatrix}\,,\quad n\in\IZ .
\ee
Thus the net contribution may be expressed as
\begin{eqnarray} \label{ereorgthreen}
&-&  f_{-1}g_{-1} \sum_{r>0} r 
\sum_{\big{(}\begin{smallmatrix} a & b\cr c & d\end{smallmatrix}\big{)}\in G^N_r\backslash \mr{P}\Gamma_1(N)^+}\hskip .1in 
 \bigg\{\prod_{n=-\infty}^\infty 
H(a_nc_n\tau_2 + b_nd_n\sigma_2 + (a_nd_n+b_nc_n)z_2)
\bigg\}
\nonumber \\ && \hskip 1in \times \
e^{2\pi i \{(-a^2-c^2/N+r ac)\tau+(-b^2-d^2/N+r bd)\sigma+
(-2ab-2cd/N + r(ad+bc)) z\}}  \, .
\end{eqnarray}
The sum is well defined since the summand is invariant under a change in the coset representative. 
Note that since $G^N_r$ depends on $r$, we are forced to exchange the order of sum over $r$ and sum over $a,b,c,d$. This is part of the prescription in the definition of $\wt{F}_k(\Omega)$.

Using \refb{e216xy}, \refb{erewritetwo} and \refb{ereorgthreen}, 
we obtain 
\begin{eqsp}
\label{eguessfin}
\wt F_k(\Omega) &= {1\over \Phi_{k}(\Omega)}
-\varepsilon_N \sum_{\big{(}\begin{smallmatrix} a & b\cr c & d\end{smallmatrix}\big{)}\in 
\mr{P}\Gamma_1(N)}
\left(e^{\pi i \{ac\tau + bd\sigma + (ad+bc)z\}} - e^{-\pi i \{ac\tau + bd\sigma + (ad+bc)z\}}\right)^{-2}  \\ & 
\hskip 1in 
\times\ f_+(a^2\tau +b^2\sigma +2abz) \ g_+(c^2\tau+d^2\sigma+2cd z) \\ 
&-\varepsilon_N 
\sum_{\big{(}\begin{smallmatrix} a & b\cr c & d\end{smallmatrix}\big{)}\in \mr{P}\Gamma_1(N)}
\sum_{r>0}  g_{-1} r  H(ac\tau_2 + bd\sigma_2 + (ad+bc)z_2) \\ &\hskip 1in \times
H\left(-ac\tau_2 - bd\sigma_2 - (ad+bc)z_2 + rNa^2\tau_2 + rNb^2\sigma_2+2rNab z_2
\right) \,
\\&\hskip 1in  \times 
e^{2\pi i \{-(c^2\tau+d^2\sigma+2cdz)/N+r( ac\tau+bd\sigma+(ad+bc)z )\}}f_+(a^2\tau+b^2\sigma+2abz) 
\\ &-\varepsilon_N 
\sum_{\big{(}\begin{smallmatrix} a & b\cr c & d\end{smallmatrix}\big{)}\in \mr{P}\Gamma_1(N)}
\sum_{r>0}  f_{-1} r  H(ac\tau_2 + bd\sigma_2 + (ad+bc)z_2) \\ &\hskip 1in \times\
H\left(-ac\tau_2 - bd\sigma_2 - (ad+bc)z_2 - rc^2\tau_2 + rd^2\sigma_2+2rcd z_2
\right) \,
\\&\hskip 1in  \times 
e^{2\pi i \{-(a^2\tau+b^2\sigma+2abz)+r( ac\tau+bd\sigma+(ad+bc)z )\}}g_+(c^2\tau+d^2\sigma+2cdz) \\ 
&-  f_{-1}g_{-1} \sum_{r>0} r\, 
\sum_{\big{(}\begin{smallmatrix} a & b\cr c & d\end{smallmatrix}\big{)}\in G^N_r\backslash \mr{P}\Gamma_1(N)^+}\hskip .1in 
 \bigg\{\prod_{n=-\infty}^\infty 
H(a_nc_n\tau_2 + b_nd_n\sigma_2 + (a_nd_n+b_nc_n)z_2)
\bigg\}
\\ & \hskip 1in \times \
e^{2\pi i \{(-a^2-c^2/N+r ac)\tau+(-b^2-d^2/N+r bd)\sigma+
(-2ab-2cd/N + r(ad+bc)) z\}}  \, .
\end{eqsp}
The sum over $a,b,c,d$ will be organized by first summing over $a,b,c,d$ subject to the condition $|a|,|b|,|c|,|d|\leq K$ for some positive integer $K$ so that the sum is finite and then summing over all positive integers $K$.  
In Section \ref{stildeFconverge} we shall prove the convergence of the sum given in
\refb{eguessfin} for $N=2,3$.
\par
Let us briefly discuss the simplification of the generating function \eqref{eguessfin} for CHL models where we have the additional $S$-duality symmetry corresponding to $\begin{pmatrix}
    0&-1/\sqrt{N}\\\sqrt{N}&0
\end{pmatrix}$, for example when $\CM=\mr{K3}$ \cite{David:2006ud,Persson:2015jka}. For such cases, the Siegel modular form satisfies \eqref{eq:gamma_N_inv_Phik}. As a result, since  $\gamma_N$-transformation acts on $\Omega$ as
\begin{eqsp}
\Omega\to\gamma_N\Omega\gamma_N^t=\begin{pmatrix}
    \sigma/N&-z\\-z&\tau
\end{pmatrix}~,    
\end{eqsp}
using this symmetry of $\Phi_{k}^{-1}$ near the $z=0$ pole with the behavior \eqref{epolestructure}, we find that 
\begin{eqsp}\label{eq:fkgk_relK3}
    g^{(k)}(\sigma)=f^{(k)}(\sigma/N)~.
\end{eqsp} 
The first effect of this symmetry is to extend the sum over a larger S-duality group generated by $\Gamma_1(N)$ and $\begin{pmatrix}
    0&-1/\sqrt{N}\\\sqrt{N}&0
\end{pmatrix}\Gamma_1(N)$ and\footnote{Note that the transformation by $\gamma_N$ does not preserve the charge lattice but acts on the quadratic T-duality invariants $Q^2,P^2,P\cdot Q$ as 
\begin{eqsp}
Q^2\to P^2/N,\quad P^2\to NQ^2,\quad P\cdot Q \to -P\cdot Q~,   
\end{eqsp}
and the invariance of the Siegel modular form $\Phi_k$ under $\gamma_N$ ensures that the full index $d(m,n,\ell)$ as well as the single-centered index $d^*(m,n,\ell)$ is invariant under this transformation:
\begin{eqsp}
    d(m,n,\ell)=d(n,m,-\ell)~,\quad d^*(m,n,\ell)=d^*(n,m,-\ell)~.
\end{eqsp}
} include a factor of $\frac{1}{2}$ owing to the invariance of the summand under $(a,b,c,d)\to (-c/\sqrt{N},-d/\sqrt{N},a\sqrt{N},b\sqrt{N})$. 
Using \eqref{eq:fkgk_relK3}, we can combine the terms in the sixth, seventh and eight line of \eqref{eguessfin} with the terms in the third, fourth and fifth lines of \eqref{eguessfin}:
\begin{eqsp}
    &-\frac{1}{2} 
\sum_{\big{(}\begin{smallmatrix} a & b\cr c & d\end{smallmatrix}\big{)}\in \mr{P}\Gamma_1(N)^+}
\sum_{r>0}  g_{-1} r  H(ac\tau_2 + bd\sigma_2 + (ad+bc)z_2) \\ &\hskip 1in \times
H\left(-ac\tau_2 - bd\sigma_2 - (ad+bc)z_2 + rNa^2\tau_2 + rNb^2\sigma_2+2rNab z_2
\right) \,
\\&\hskip 1in  \times 
e^{2\pi i \{-(c^2\tau+d^2\sigma+2cdz)/N+r( ac\tau+bd\sigma+(ad+bc)z )\}}f_+(a^2\tau+b^2\sigma+2abz) 
\\ &-\frac{1}{2}
\sum_{\big{(}\begin{smallmatrix} a & b\cr c & d\end{smallmatrix}\big{)}\in \mr{P}\Gamma_1(N)^+}
\sum_{r>0}  f_{-1} r  H(ac\tau_2 + bd\sigma_2 + (ad+bc)z_2) \\ &\hskip 1in \times\
H\left(-ac\tau_2 - bd\sigma_2 - (ad+bc)z_2 - rc^2\tau_2 + rd^2\sigma_2+2rcd z_2
\right) \,
\\&\hskip 1in  \times 
e^{2\pi i \{-(a^2\tau+b^2\sigma+2abz)+r( ac\tau+bd\sigma+(ad+bc)z )\}}g_+(c^2\tau+d^2\sigma+2cdz)
\\&=-\sum_{\big{(}\begin{smallmatrix} a & b\cr c & d\end{smallmatrix}\big{)}\in \mr{P}\Gamma_1(N)^+}
\sum_{r>0}  g_{-1} r  H(ac\tau_2 + bd\sigma_2 + (ad+bc)z_2) \\ &\hskip 1in \times
H\left(-ac\tau_2 - bd\sigma_2 - (ad+bc)z_2 + rNa^2\tau_2 + rNb^2\sigma_2+2rNab z_2
\right) \,
\\&\hskip 1in  \times 
e^{2\pi i \{-(c^2\tau+d^2\sigma+2cdz)/N+r( ac\tau+bd\sigma+(ad+bc)z )\}}f_+(a^2\tau+b^2\sigma+2abz) 
\\&=-\sum_{\big{(}\begin{smallmatrix} a & b\cr c & d\end{smallmatrix}\big{)}\in \mr{P}\Gamma_1(N)^+}
\sum_{r>0}  f_{-1} r  H(ac\tau_2 + bd\sigma_2 + (ad+bc)z_2) \\ &\hskip 1in \times\
H\left(-ac\tau_2 - bd\sigma_2 - (ad+bc)z_2 - rc^2\tau_2 + rd^2\sigma_2+2rcd z_2
\right) \,
\\&\hskip 1in  \times 
e^{2\pi i \{-(a^2\tau+b^2\sigma+2abz)+r( ac\tau+bd\sigma+(ad+bc)z )\}}g_+(c^2\tau+d^2\sigma+2cdz)~.
\end{eqsp}
Thus we obtain the generating function given in \eqref{eguessfinintroK3}.
\par

\section{Convergence of $\wt F_k(\Omega)$} \label{stildeFconverge}

In this section we shall prove the convergence of the sum in \eqref{eguessfin}. 
Let us define 
\begin{align} 
    \CF_{1}&:=\sum_{\big{(}\begin{smallmatrix} a & b\cr c & d\end{smallmatrix}\big{)}\in 
\mr{P}\Gamma_1(N)}
\left(e^{\pi i \{ac\tau + bd\sigma + (ad+bc)z\}} - e^{-\pi i \{ac\tau + bd\sigma + (ad+bc)z\}}\right)^{-2} \nonumber \\ & 
\hskip 1in 
\times\ f_+(a^2\tau +b^2\sigma +2abz) \ g_+(c^2\tau+d^2\sigma+2cd z)~,
\label{ef1}\\
\CF_2&:=
\sum_{\big{(}\begin{smallmatrix} a & b\cr c & d\end{smallmatrix}\big{)}\in \mr{P}\Gamma_1(N)}
\sum_{r>0}  g_{-1} r  H(ac\tau_2 + bd\sigma_2 + (ad+bc)z_2) \nonumber\\ &\hskip 1in \times
H\left(-ac\tau_2 - bd\sigma_2 - (ad+bc)z_2 + rNa^2\tau_2 + rNb^2\sigma_2+2rNab z_2
\right) \,
\nonumber\\&\hskip 1in  \times 
e^{2\pi i \{-(c^2\tau+d^2\sigma+2cdz)/N+r( ac\tau+bd\sigma+(ad+bc)z )\}}f_+(a^2\tau+b^2\sigma+2abz) 
\label{ef2}\\ \CF_3&:= 
\sum_{\big{(}\begin{smallmatrix} a & b\cr c & d\end{smallmatrix}\big{)}\in \mr{P}\Gamma_1(N)}
\sum_{r>0}  f_{-1} r  H(ac\tau_2 + bd\sigma_2 + (ad+bc)z_2)\nonumber \\ &\hskip 1in \times\
H\left(-ac\tau_2 - bd\sigma_2 - (ad+bc)z_2 - rc^2\tau_2 + rd^2\sigma_2+2rcd z_2
\right) \,
\nonumber\\&\hskip 1in  \times 
e^{2\pi i \{-(a^2\tau+b^2\sigma+2abz)+r( ac\tau+bd\sigma+(ad+bc)z )\}}g_+(c^2\tau+d^2\sigma+2cdz) \label{ef3}\\ 
\CF_4&:=  f_{-1}g_{-1} \sum_{r>0} r\, 
\sum_{\big{(}\begin{smallmatrix} a & b\cr c & d\end{smallmatrix}\big{)}\in G^N_r\backslash \mr{P}\Gamma_1(N)^+}\hskip .1in 
 \bigg\{\prod_{n=-\infty}^\infty 
H(a_nc_n\tau_2 + b_nd_n\sigma_2 + (a_nd_n+b_nc_n)z_2)
\bigg\}\nonumber
\\ & \hskip 1in \times \
e^{2\pi i \{(-a^2-c^2/N+r ac)\tau+(-b^2-d^2/N+r bd)\sigma+
(-2ab-2cd/N + r(ad+bc)) z\}}  \, ,
\label{ef4}
\end{align}
so that \eqref{eguessfin} can be written as
\begin{equation}\label{eq:tildeFCF123}
    \widetilde{F}_k(\Omega)=\frac{1}{\Phi_{k}}-\varepsilon_N(\CF_1+\CF_2+\CF_3)-\CF_4~.
\end{equation}
%Thus the proving Theorem \ref{thm:S_conv} is equivalent to proving the convergence of $\CF_{1},\CF_{2},\CF_3$. 

% \begin{figure}
% \begin{center}
% \figfour
% \end{center}

% \vskip -.3in

% \caption{Chambers in the $x$-$y$ plane separated by pole locations of $1/\Phi_{k}$. \textbf{Comment out if figure below is better.}} \label{figfour}
% \end{figure}
\subsection{The fundamental chamber $\CR_N$}
Having constructed the generating function, it will be useful to develop a geometric picture for the matrices
$\begin{pmatrix} a & b\cr c & d\end{pmatrix}\in \Gamma_1(N)$ as in \cite{Bhand:2025ghn}. Let us define
\be
x := {z_2\over \tau_2}, \qquad y := {\sigma_2\over \tau_2}\, .
\ee
Then the imaginary parts of \refb{epoles} for $n_2=0$ take the form:
\be \label{epolesimaginary}
- m_1 +n_1\, y + j\, x
=0, \quad
m_1/N,n_1,j\in \IZ,
\quad m_1 n_1 + {j^2\over 4}={1\over 4}\, .
\ee
The fact that $\Omega\in\IH_2$ implies that 
\be\label{eparabola}
y > x^2\, .
\ee
Again one can show that the lines represented by \eqref{epolesimaginary} intersect only on the parabola and hence the vertices 
of these chambers either lie on the parabola or at $y=\infty$. Let us see how these chambers are related to the chambers in the $\tau_\infty$-plane. The $\CR_N$-chamber in the $\tau_\infty$-plane is described as follows \cite{Sen:2007vb}: the $\CR_N$-chamber is bounded by the lines connecting $(i\infty,0)$ and $(1,i\infty)$ on the left and right and a set of circle segments in the bottom. The shape of these circle segments  depends on the charge but the points where it intersects the real $\tau_\infty$-axis is universal and is given by $\mr{SL}(2,\IZ)$-matrices $\begin{pmatrix}
    a&b\\c&d
\end{pmatrix}$ with $cd\in N\IZ$. A circle segment corresponding to $\begin{pmatrix}
    a&b\\c&d
\end{pmatrix}$ intersects the real $\tau_\infty$-axis at $\frac{a}{c},\frac{b}{d}$. We will denote this circle segment by $(\frac{a}{c},\frac{b}{d})$. The first circle segment bounding $\CR_N$ from below begins at 0 and ends at some point
$P_1$ on the positive real $\tau_\infty$-axis, the second one begins at $P_1$ and ends at some point $P_2$ to
the right of $P_1$ and so on with the final segment ending at 1. The matrices corresponding to these circle segments is determined by requiring that it travels maximum possible
distance to the right starting from a given vertex $P$. Any other segment that travels less distance will lie beneath the maximal segment and hence is not a boundary of $\CR_N$.  One can follow the procedure illustrated in \cite{Sen:2007vb} to determine the walls and vertices. Some of the vertices in counterclockwise direction is given by 
\begin{eqsp}
    i\infty,0,\frac{1}{N},\frac{1}{N-1},\frac{N-1}{N(N-2)},\cdots,1-\frac{N-1}{N(N-2)},1-\frac{1}{N-1},1-\frac{1}{N},1~.
\end{eqsp}
Clearly, for $N=1,2,3$ there are finitely many walls with vertices at 
\begin{eqsp}
    i\infty,0,1,\quad &N=1~,
    \\
    i\infty,0,\frac{1}{2},1,\quad &N=2~,
    \\
    i\infty,0,\frac{1}{3},\frac{1}{2},\frac{2}{3},1,\quad &N=3~.
\end{eqsp}
For $N\geq 4$, there are infinitely many walls and vertices \cite{Sen:2007vb}. This is easily seen as follows:
as noted in \cite{Sen:2007vb}, the $\mr{SL}(2,\IR)$-matrix $\begin{pmatrix}
    \sqrt{N}&-1/\sqrt{N}\\\sqrt{N}&0
\end{pmatrix}$, via its M\"{o}bius action on $\tau_\infty$, shifts the vertices by one step in clockwise direction mapping $i\infty\to 1\to1-\frac{1}{N}$ and so on. But it has fixed points at 
\begin{eqsp}
    v_c:=\frac{1}{2}\left(1+\sqrt{1-\frac{4}{N}}\right)~,\quad \text{and}\quad v_c^{-1}/N=\frac{1}{2}\left(1-\sqrt{1-\frac{4}{N}}\right)~.
\end{eqsp}
Thus the action of $\begin{pmatrix}
    \sqrt{N}&-1/\sqrt{N}\\\sqrt{N}&0
\end{pmatrix}$ has accumulation points at $v_c,v_c^{-1}/N$ and hence there are infinitely many vertices. A vertex in the range
\begin{eqsp}\label{eq:vc_range}
v_c^{-1}/N<\tau_{\infty}<v_c~,    
\end{eqsp}
is never mapped outside this range by $\begin{pmatrix}
    \sqrt{N}&-1/\sqrt{N}\\\sqrt{N}&0
\end{pmatrix}$.  
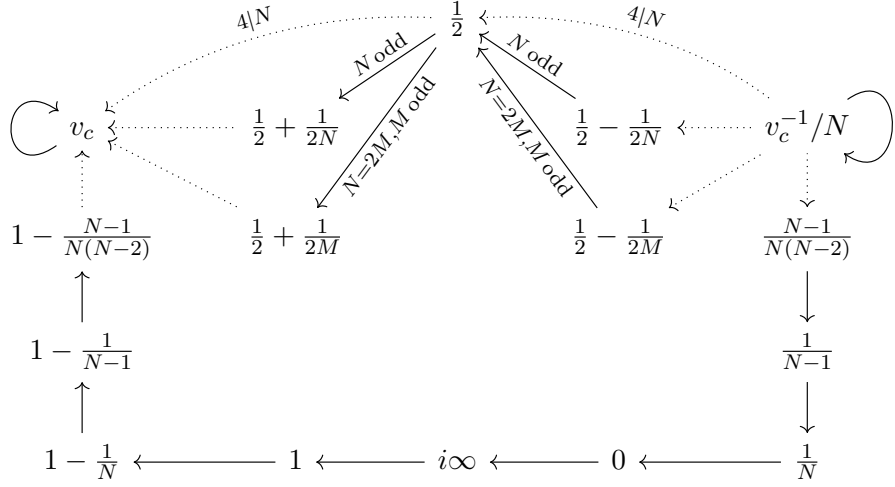
\begin{figure}
    \centering
\[\begin{tikzcd}
	&& {\frac{1}{2}} && \\
	{v_c} & {\frac{1}{2}+\frac{1}{2N}} && {\frac{1}{2}-\frac{1}{2N}} & {v_c^{-1}/N} \\
	{1-\frac{N-1}{N(N-2)}} & {\frac{1}{2}+\frac{1}{2M}} && {\frac{1}{2}-\frac{1}{2M}} & {\frac{N-1}{N(N-2)}} \\
	{1-\frac{1}{N-1}} &&&& {\frac{1}{N-1}} \\
	{1-\frac{1}{N}} & 1 & {i\infty} & 0 & {\frac{1}{N}}
	\arrow[tail reversed,no head,midway,sloped,"{N\,\text{odd}}", from=1-3, to=2-4]
	\arrow[tail reversed,no head,midway,sloped,"{4|N}", curve={height=-18pt}, dotted, from=1-3, to=2-5]
	\arrow[tail reversed,no head,midway,sloped,"{N=2M,M\,\text{odd}}"', from=1-3, to=3-4]
	\arrow[tail reversed,no head,midway,sloped,"{4|N}", curve={height=-18pt}, dotted, from=2-1, to=1-3]
	\arrow[from=2-1, to=2-1, loop, in=145, out=215, distance=10mm]
	\arrow[tail reversed,no head,dotted, from=2-1, to=2-2]
	\arrow[tail reversed,no head,dotted, from=2-1, to=3-2]
	\arrow[tail reversed,no head,midway,sloped,"{N\,\text{odd}}", from=2-2, to=1-3]
	\arrow[tail reversed,no head,dotted, from=2-4, to=2-5]
	\arrow[from=2-5, to=2-5, loop, in=325, out=35, distance=10mm]
	\arrow[dotted, from=3-1, to=2-1]
	\arrow[tail reversed,no head,midway,sloped,"{N=2M,M\,\text{odd}}"', from=3-2, to=1-3]
	\arrow[tail reversed,no head,dotted, from=3-4, to=2-5]
	\arrow[dotted, tail reversed, no head, from=3-5, to=2-5]
	\arrow[from=4-1, to=3-1]
	\arrow[tail reversed, no head, from=4-5, to=3-5]
	\arrow[from=5-1, to=4-1]
	\arrow[from=5-2, to=5-1]
	\arrow[from=5-3, to=5-2]
	\arrow[from=5-4, to=5-3]
	\arrow[tail reversed, no head, from=5-5, to=4-5]
	\arrow[from=5-5, to=5-4]
\end{tikzcd}\]
\caption{The action of the $\mr{SL}(2,\IR)$-matrix $\begin{pmatrix}
    \sqrt{N}&-1/\sqrt{N}\\\sqrt{N}&0
\end{pmatrix}$ on the vertices of the fundamental chamber $\CR_N$. The vertices connected by the dotted arrows on both sides are accumulation points.}
    \label{fig:act_SL2r_mat_RN}
\end{figure}
As shown in \cite{Sen:2007vb}, $\frac{1}{2}$ is in the range \eqref{eq:vc_range} and is an accumulation point if $4|N$. For $N$ odd or $N=2M,M$ odd, the vertices to the left and to the right of $\frac{1}{2}$ has been calculated in \cite{Sen:2007vb}. The full picture is summarized in Figure \ref{fig:act_SL2r_mat_RN}. In this paper, we are only concerned with $N=1,2,3$, the $\CR_N$-chamber for these two cases is shown in Figure
\ref{fig:RN-chamber1}-\ref{fig:RN-chamber3}.

\begin{figure}[htbp]
    \centering
    % === LEFT SUBFIGURE ===
    \begin{subfigure}[b]{0.48\textwidth}
        \centering
        \resizebox{\textwidth}{!}{% <-- Scales TikZ to fit the 0.48 width exactly
            \begin{tikzpicture}
                % Axes
                \draw[-{Stealth[scale=1.2]}] (-1,0) -- (5.5,0)node[right] {$\mr{Re}\,\tau_\infty$};
                \draw[-{Stealth[scale=1.2]}] (0,0) -- (0,3.5)node[above] {$\mr{Re}\,\tau_\infty$};
            
                % Axis labels
                \node[below] at (0,0) {0};
                \node[below] at (4,0) {1};
            
                % Curves bounding the region R
                \draw[thick,red] (0,0) -- (1,3.2);
                \draw[thick,red] (0,0) to[out=25, in=155] (4,0);
                \draw[thick,red] (4,0) -- (5.2,2.8);
            
                % Label for Region R
                \node at (2.4,1.8) {$\mathcal{R}_1$};
            \end{tikzpicture}% <-- Percent sign prevents extra space issues
        }
        \caption{The fundamental chamber $\mathcal{R}_1$ in the $\tau_\infty$-plane bounded by walls of marginal stability.} % Fixed \CR to \mathcal{R}
        \label{fig:sub-chamber1}
    \end{subfigure}
    \hfill % === Pushes them to the outer margins perfectly ===
    % === RIGHT SUBFIGURE ===
    \begin{subfigure}[b]{0.48\textwidth}
        \centering
        \resizebox{\textwidth}{!}{% <-- Scales the large scale=3 diagram down to fit
            \begin{tikzpicture}[scale=3]
                % axes
                \draw[->] (-1.6,0) -- (1.6,0) node[right] {$x:=\frac{z_2}{\tau_2}$};
                \draw[->] (0,0) -- (0,2.2) node[above] {$y:=\frac{\sigma_2}{\tau_2}$};
            
                % parabola
                \draw[thick,domain=-1.5:1.5,samples=200] plot (\x,{\x*\x}) node[right]{} ;
            
                % slanted lines from origin
                \draw[thick,red] (0,0) -- (-1,1);
                \draw[thick,red] (0,0) -- (1,1);
            
                % vertical lines starting at (+-1,1) and going up
                \draw[thick,red] (-1,1) -- (-1,2);
                \draw[thick,red] (1,1) -- (1,2);
                \draw[thick,red] (0,0) -- (0.5,0.25);
                \draw[thick,red] (0.5,0.25) -- (1,1);
                \draw[thick,red] (0,0) -- (0,2);
                
                % labels for regions
                \node at (-0.5,1.4) {$\mathcal{R}_1$};
                \node at (0.5,1.4) {$\mathcal{L}_1$};
                \node at (-1.1,1) {$-1$};
                \node at (1.1,1) {$1$};
                \node at (0,-0.1) {$0$};
                \node at (0.65,0.25) {$\frac{1}{2}$};
            \end{tikzpicture}%
        }
        \caption{Chambers in the $x$-$y$ plane separated by pole locations of $1/\Phi_{k}$.}
        \label{fig:sub-chamber2}
    \end{subfigure}
    
    % Main figure caption
    \caption{Fundamental $\CR_1$-chamber in the $\tau_\infty$ and $x$-$y$ plane.}
    \label{fig:RN-chamber1}
\end{figure}
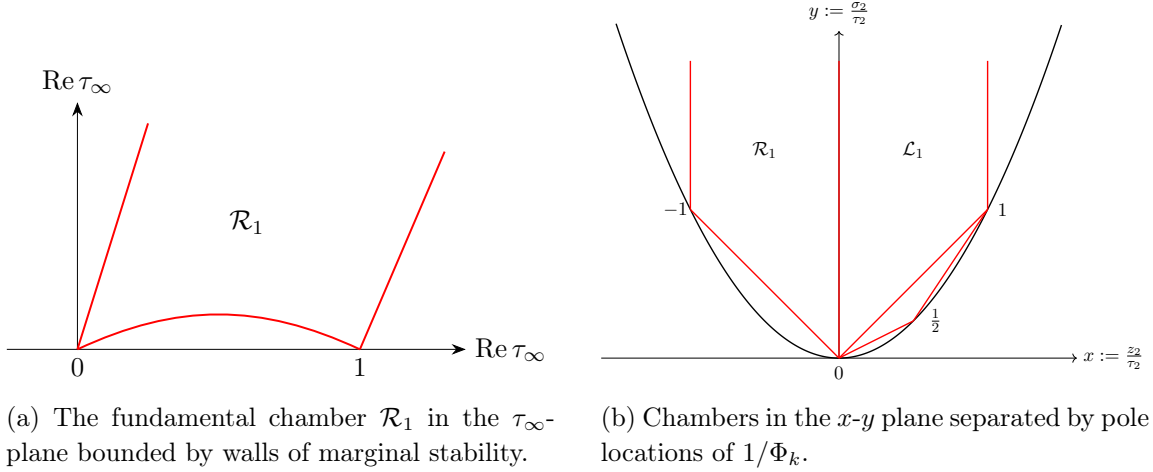
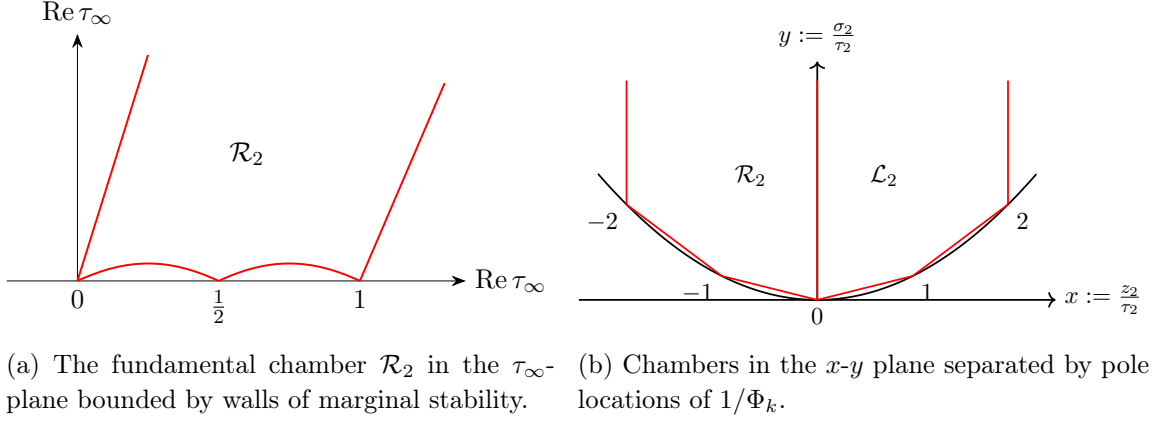
\begin{figure}[htbp]
    \centering
    % === LEFT SUBFIGURE ===
    \begin{subfigure}[b]{0.48\textwidth}
        \centering
        \resizebox{\textwidth}{!}{% <-- Scales TikZ to fit the 0.48 width exactly
            \begin{tikzpicture}

    % Axes
    \draw[-{Stealth[scale=1.2]}] (-1,0) -- (5.5,0)node[right] {$\mr{Re}\,\tau_\infty$};
    \draw[-{Stealth[scale=1.2]}] (0,0) -- (0,3.5)node[above] {$\mr{Re}\,\tau_\infty$};

    % Axis labels
    \node[below] at (0,0) {0};
    \node[below] at (2,0) {$\frac{1}{2}$};
    \node[below] at (4,0) {1};

    % Left straight line boundary starting at (0,0)
    \draw[thick,red] (0,0) -- (1,3.2);
    
    % Bottom curved boundary from 0 to 1/2
    \draw[thick,red] (0,0) to[out=25, in=155] (2,0);
    
    % Bottom curved boundary from 1/2 to 1
    \draw[thick,red] (2,0) to[out=25, in=155] (4,0);
    
    % Right straight line boundary starting at (1,0)
    \draw[thick,red] (4,0) -- (5.2,2.8);

    % Label for Region R
    \node at (2.4,1.8) {$\mathcal{R}_2$};

\end{tikzpicture}% <-- Percent sign prevents extra space issues
        }
        \caption{The fundamental chamber $\mathcal{R}_2$ in the $\tau_\infty$-plane bounded by walls of marginal stability.} % Fixed \CR to \mathcal{R}
        \label{fig:sub-chamber1}
    \end{subfigure}
    \hfill % === Pushes them to the outer margins perfectly ===
    % === RIGHT SUBFIGURE ===
    \begin{subfigure}[b]{0.5\textwidth}
        \centering
        \resizebox{\textwidth}{!}{% <-- Scales the large scale=3 diagram down to fit
            \begin{tikzpicture}[scale=1.5]
        % Axes
        \draw[->, thick] (-2.5,0) -- (2.5,0) node[right] {$x := \frac{z_2}{\tau_2}$};
        \draw[->, thick] (0,0) -- (0,2.5) node[above] {$y := \frac{\sigma_2}{\tau_2}$};
        \node[below] at (0,0) {$0$};

        % Parabola (y = 0.25 * x^2 to spread out x=1 and x=2 beautifully)
        \draw[thick, domain=-2.3:2.3, samples=200] plot (\x, {0.25*\x*\x});

        % --- Red Polygonal Path Lines ---
        % Right side vertices: (0,0) -> (1, 0.25) -> (2, 1.0) -> vertical up
        \draw[thick, red] (0,0) -- (1, 0.25);
        \draw[thick, red] (1, 0.25) -- (2, 1.0);
        \draw[thick, red] (2, 1.0) -- (2, 2.3); % Vertical line

        % Left side vertices (Symmetric counterparts)
        \draw[thick, red] (0,0) -- (-1, 0.25);
        \draw[thick,red] (0,0) -- (0,2.3);
        \draw[thick, red] (-1, 0.25) -- (-2, 1.0);
        \draw[thick, red] (-2, 1.0) -- (-2, 2.3); % Vertical line

        % --- Text Labels and Annotations ---
        % Chamber labels
        \node at (-0.7, 1.3) {$\mathcal{R}_2$};
        \node at (0.7, 1.3) {$\mathcal{L}_2$};

        % Numerical markers along the curve
        \node[below right] at (1, 0.25) {$1$};
        \node[below right] at (2, 1.0) {$2$};
        \node[below left] at (-1, 0.25) {$-1$};
        \node[below left] at (-2, 1.0) {$-2$};
        
    \end{tikzpicture}%
        }
        \caption{Chambers in the $x$-$y$ plane separated by pole locations of $1/\Phi_{k}$.}
        \label{fig:sub-chamber2}
    \end{subfigure}
    
    % Main figure caption
    \caption{Fundamental $\CR_2$-chamber in the $\tau_\infty$ and $x$-$y$ plane.}
    \label{fig:RN-chamber2}
\end{figure}
\begin{figure}[htbp]
    \centering
    % === LEFT SUBFIGURE ===
    \begin{subfigure}[b]{0.48\textwidth}
        \centering
        \resizebox{\textwidth}{!}{% <-- Scales TikZ to fit the 0.48 width exactly
            \begin{tikzpicture}

    % Axes
    \draw[-{Stealth[scale=1.2]}] (-1,0) -- (6.5,0)node[right] {$\mr{Re}\,\tau_\infty$};
    \draw[-{Stealth[scale=1.2]}] (0,0) -- (0,3.5)node[above] {$\mr{Im}\,\tau_\infty$};

    % Axis labels
    \node[below] at (0,0) {0};
    \node[below] at (1.8,0) {$\frac{1}{3}$};
    \node[below] at (2.7,0) {$\frac{1}{2}$};
    \node[below] at (3.6,0) {$\frac{2}{3}$};
    \node[below] at (5.1,0) {1};

    % Left straight line boundary starting at (0,0)
    \draw[thick,red] (0,0) -- (1,3.2);
    
    % Bottom curved boundaries (Farey arcs/modular-like regions)
    \draw[thick,red] (0,0) to[out=35, in=145] (1.8,0);
    \draw[thick,red] (1.8,0) to[out=45, in=135] (2.7,0);
    \draw[thick,red] (2.7,0) to[out=45, in=135] (3.6,0);
    \draw[thick,red] (3.6,0) to[out=35, in=145] (5.1,0);
    
    % Right straight line boundary starting at (1,0)
    \draw[thick,red] (5.1,0) -- (6.3,2.8);

    % Label for Region R
    \node at (3.1,2.0) {$\mathcal{R}_3$};

\end{tikzpicture}% <-- Percent sign prevents extra space issues
        }
        \caption{The fundamental chamber $\mathcal{R}_3$ in the $\tau_\infty$-plane bounded by walls of marginal stability.} % Fixed \CR to \mathcal{R}
        \label{fig:sub-chamber1}
    \end{subfigure}
    \hfill % === Pushes them to the outer margins perfectly ===
    % === RIGHT SUBFIGURE ===
    \begin{subfigure}[b]{0.5\textwidth}
        \centering
        \resizebox{\textwidth}{!}{% <-- Scales the large scale=3 diagram down to fit
            \begin{tikzpicture}[scale=1.5]
        % Axes
        \draw[->, thick] (-3.5,0) -- (3.5,0) node[right] {$x := \frac{z_2}{\tau_2}$};
        \draw[->, thick] (0,0) -- (0,2.5) node[above] {$y := \frac{\sigma_2}{\tau_2}$};
        \node[below] at (0,0) {$0$};

        % Parabola (y = 0.2 * x^2 to fit up to x=3 beautifully)
        \draw[thick, domain=-3.2:3.2, samples=200] plot (\x, {0.2*\x*\x});

        % --- Red Polygonal Path Lines ---
        % Right side vertices: (0,0) -> (1, 0.2) -> (1.5, 0.45) -> (2, 0.8) -> (3, 1.8) -> up
        \draw[thick, red] (0,0) -- (1, 0.2);
        \draw[thick, red] (1, 0.2) -- (1.5, 0.45);
        \draw[thick, red] (1.5, 0.45) -- (2, 0.8);
        \draw[thick, red] (2, 0.8) -- (3, 1.8);
        \draw[thick, red] (3, 1.8) -- (3, 2.3); % Slanted near-vertical line going up

        % Left side vertices (Symmetric counterparts)
        \draw[thick, red] (0,0) -- (-1, 0.2);
        \draw[thick, red] (-1, 0.2) -- (-1.5, 0.45);
        \draw[thick, red] (-1.5, 0.45) -- (-2, 0.8);
        \draw[thick, red] (-2, 0.8) -- (-3, 1.8);
        \draw[thick,red] (0,0) -- (0,2.3);
        \draw[thick, red] (-3, 1.8) -- (-3, 2.3); % Slanted near-vertical line going up

        % --- Text Labels and Annotations ---
        % Chamber labels
        \node at (-0.9, 1.4) {$\mathcal{R}_3$};
        \node at (0.9, 1.4) {$\mathcal{L}_3$};

        % Right side numerical markers
        \node[below, scale=0.9] at (1, 0.2) {$1$};
        \node[below right, scale=0.9] at (1.5, 0.45) {$\frac{3}{2}$};
        \node[right, scale=0.9] at (2, 0.8) {$2$};
        \node[below right, scale=0.9] at (3, 1.8) {$3$};

        % Left side numerical markers
        \node[below left, scale=0.9] at (-1, 0.2) {$-1$};
        \node[below left, scale=0.9] at (-1.5, 0.4) {$-\frac{3}{2}$};
        \node[left, scale=0.9] at (-2, 0.8) {$-2$};
        \node[left, scale=0.9] at (-3, 1.8) {$-3$};
        
    \end{tikzpicture}%
        }
        \caption{Chambers in the $x$-$y$ plane separated by pole locations of $1/\Phi_{k}$.}
        \label{fig:sub-chamber2}
    \end{subfigure}
    
    % Main figure caption
    \caption{Fundamental $\CR_3$-chamber in the $\tau_\infty$ and $x$-$y$ plane.}
    \label{fig:RN-chamber3}
\end{figure}
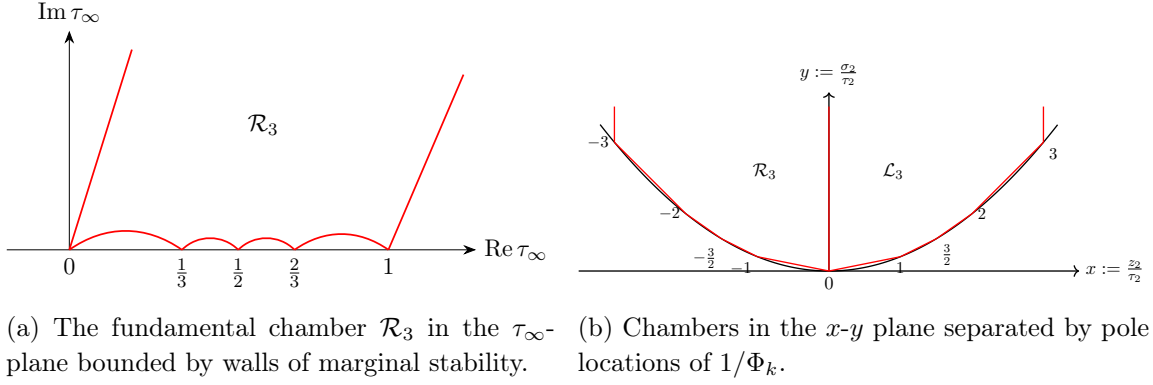
\par 
Let us now map the $\CR_N$-chamber to the $x$-$y$ plane. The chambers are mapped to the $x$-$y$ plane by the following correspondence \cite{1104.1498}:
\begin{eqsp}
    \left(\frac{b}{d}, \frac{a}{c}\right) \longleftrightarrow c d+a b y +(a d+b c) x=0 ~,
\end{eqsp}
where $\left(\frac{b}{d}, \frac{a}{c}\right) $ denotes the circle segment connecting $\frac{b}{d}, \frac{a}{c}$ in the $\tau_\infty$-plane. The image of the $\CR_N$-chamber in the $x$-$y$ plane is also denoted by the $\CR_N$. It is bounded by lines $x=0,x=-N$ on right and left and a bunch of straight lines intersecting the parabola from below. The straight line $c d+a b y +(a d+b c) x=0$ intersects the parabola $y=x^2$ at the points 
\begin{eqsp}
    x=-\frac{c}{a},-\frac{d}{b}.
\end{eqsp}
This determines the vertices of the $\CR_N$-chamber in the $x$-$y$ plane from the vertices in the $\tau_\infty$-plane. Thus in the increasing order of the $x$-values, the vertices of the $\CR_N$-chamber in the $x$-$y$ plane are given by 
\begin{eqsp}\label{eq:vert_RN_xy}
    -N,-N+1,-\frac{N(N-2)}{N-1},\cdots,-Nv_c,\cdots,-v_c^{-1},\cdots-\frac{N(N-2)}{N^2-3N+1},-\frac{N-1}{N-2},-\frac{N}{N-1},-1,0,\infty~.
\end{eqsp}
Explicitly we have 
\begin{eqsp}
N=1:\quad &\{0,-1,\infty\}~,\\
    N=2:\quad & \left\{0,-1,-2,\infty\right\}~,
    \\
    N=3:\quad & \left\{0,-1,-\frac{3}{2},-2,-3,\infty\right\}~,
    \\
    N=4:\quad & 
\left\{
0,\;
-1,\;
-\frac{4}{3},\;
-\frac{3}{2},\;
-\frac{8}{5},\;
-\frac{5}{3},\;
-\frac{12}{7},\;
-\frac{7}{4},\;
-\frac{16}{9},\;
-\frac{9}{5},\cdots,
-2,\cdots,\right.
\\&\left.\quad\quad-\frac{20}{9},\;
-\frac{9}{4},\;
-\frac{16}{7},\;
-\frac{7}{3},\;
-\frac{12}{5},\;
-\frac{5}{2},\;
-\frac{8}{3},\;
-3,\;
-4,\;
\infty
\right\}~.
\end{eqsp}
Equivalently, we have 
\begin{eqsp}
\CR_1:\quad &z_2<0,\quad z_2>-\tau_2,\quad z_2>-\sigma_2~,
\\
    \CR_2:\quad &z_2<0,\quad z_2>-2\tau_2,\quad z_2>-\sigma_2,\quad 3z_2+2\tau_2>-\sigma_2~,
    \\
    \CR_3:\quad & z_2<0,\quad z_2>-3\tau_2,\quad z_2>-\sigma_2,\quad  5z_2+3\tau_2>-2\sigma_2,\quad 7z_2+6\tau_2>-2\sigma_2,\\&5z_2+6\tau_2>-\sigma_2~.
\end{eqsp}
$\begin{pmatrix}
    a&b\\c&d
\end{pmatrix}\in\mathrm{P}\Gamma_1(N)^+$ acts on points in the $x$-$y$ plane as:
\be \label{e471y}
x  \longmapsto  {ac + bd y + (ad + bc) x
\over a^2 + b^2 y + 2ab x}, 
\qquad  y \longmapsto  {c^2 + d^2 y + 2cd x
\over a^2 + b^2 y + 2ab x}\, .
\ee
This action is the same as the $\mr{P}\Gamma_1(N)^+$-action on $\Omega$ given by $\Omega\to\gamma\Omega\gamma^t$ but now written in variables $x=z_2/\tau_2,y=\sigma_2/\tau_2$. 
The image of the action of $\mathrm{P}\Gamma_1(N)^+$ on a given chamber can be determined by the image of the vertices of the chamber. In particular, a $\mathrm{P}\Gamma_1(N)^+$ matrix can either permute the vertices of a given chamber 
(in which case, it preserves the chamber) or maps the vertices of one chamber to the vertices of another chamber. For example, the matrix $\begin{pmatrix}
    0&-1/\sqrt{N}\\\sqrt
    N&0
\end{pmatrix}$ acts as on a vertex $(x,y=x^2)$ as
\begin{eqsp}\label{eq:xtoN/x}
    x\mapsto -\frac{N}{x}~.
\end{eqsp}
Thus, under $\begin{pmatrix}
    0&-1/\sqrt{N}\\\sqrt
    N&0
\end{pmatrix}$, the $\CR_N$ chamber is mapped to the $\CL_N$ chamber, bounded by lines $x=0,x=N$ from left and right respectively and a bunch of lines intersecting the parabola from below. The $\CR_N$-chamber and $\CL_N$-chamber is shown in Figure \ref{fig:RN-chamber1}--\ref{fig:RN-chamber3} for $N=1,2,3$. Similarly, the matrix $\begin{pmatrix}
    1&0\\kN&1
\end{pmatrix}$ acts on a vertex  $(x,y=x^2)$ as
\begin{eqsp}\label{esl2zonxy}
\begin{pmatrix}
    1&0\\kN&1
\end{pmatrix}:x\mapsto x+kN~.    
\end{eqsp}
We will again specify a chamber in the $x$-$y$ plane by its vertices. For
example, the chamber labeled as $\CR_N$ in Figure \ref{fig:RN-chamber1}--\ref{fig:RN-chamber3}
will be labeled as $(-N,...,0,\infty)$ and the chamber
labeled as $\CL_N$ in Figure \ref{fig:RN-chamber1}--\ref{fig:RN-chamber3} is labeled as $(0,...,N,\infty)$. 
Since in our expression for $\CF_{i},~i=1,2,3,4$, replacing $\Omega$ by $\gamma\Omega
\gamma^t$  has the effect of multiplying the $\Gamma_1(N)$ matrix
$\begin{pmatrix} a&b\cr c&d\end{pmatrix}$ by $\gamma$ from the right, and we
are summing over all $\Gamma_1(N)$ matrices,
$\CF_{i},~i=1,2,3,4$ are formally satisfies
\be \label{esl3zoffi}
\CF_i(\gamma\Omega\gamma^t)=\CF_i(\Omega),\qquad\gamma\in\Gamma_1(N)~,\quad i=1,2,3,4\, .
\ee
From the expression for $\CF_4$, we see that it is formally invariant under the larger group $\mathrm{P}\Gamma_1(N)^+$. 
We will prove the convergence of $\CF_i$ for $(\tau_2,\sigma_2,z_2)$
in a particular chamber and then use \refb{esl3zoffi} to define
$\CF_i(\Omega)$ in any other chamber.
Hence, we can restrict the original choice of $(x,y)$ to one particular
chamber for proving convergence of the sum. We shall take this to be the 
chamber\footnote{We are assuming that $\CR_N$ is a fundamental chamber, in the sense that any other chamber can be reached from $\CR_N$ by the action of some $\Gamma_1(N)$ transformation. This has been proved for $N\leq 3$ in \cite{Sen:2007vb}.} $\CR_N$. 
We will prove the following theorem:
\begin{thm}\label{thm:S_conv}
For $N=2,3$, the sum over $a,b,c,d$ and $r$ in \eqref{eguessfin} converges absolutely and uniformly on compact subsets of $\CR_N$. 
\end{thm}
We will prove this by showing that the sums $\CF_1,\CF_2,\CF_3,\CF_4$ converges absolutely and uniformly on compact subsets of $\CR_N$.  
\subsection{Convergence of $\CF_1,\CF_2,\CF_3$}\label{sec:conv_CF1}
As discussed above, we shall work in the chamber
\begin{eqsp}\label{echamber}
    \CR_2:\quad &z_2<0,\quad z_2>-2\tau_2,\quad z_2>-\sigma_2,\quad 3z_2+2\tau_2>-\sigma_2~,
    \\
    \CR_3:\quad & z_2<0,\quad z_2>-3\tau_2,\quad z_2>-\sigma_2,\quad  5z_2+3\tau_2>-2\sigma_2,\quad 7z_2+6\tau_2>-2\sigma_2,\\&5z_2+6\tau_2>-\sigma_2~.
\end{eqsp}
We will write
\begin{equation} \label{edefprimenew}
    \tau'_2:=a^2\tau_2+b^2\sigma_2+2abz_2,\quad \sigma'_2:=c^2\tau_2+d^2\sigma_2+2cdz_2,\quad z_2':=ac\tau_2+bd\sigma_2+(ad+bc)z_2~,
\end{equation}
for the transformed variables. 
The absolute and uniform convergence of $\CF_1$ is identical to \cite{Bhand:2025ghn} with the following improvement for bound on $f_+(\tau'),g_+(\sigma')$. In the $\CR_N$-chamber, we have
\begin{eqsp}\label{eq:bound_tau2N2}
N=2:\quad    a^2\tau_2+b^2\sigma_2+2ab z_2&=
a^2(\tau_2+z_2)+b^2(\sigma_2+z_2)+(a-b)^2(-z_2)
\\
&=(a-b)^2(2\tau_2+z_2)+b^2(3z_2+2\tau_2+\sigma_2)+(a-2b)^2(-\tau_2-z_2)~.
\end{eqsp}
\begin{eqsp}
N=2:\quad   
c^2\tau_2+d^2\sigma_2+2cd z_2&=
c^2(\tau_2+z_2)+d^2(\sigma_2+z_2)+(c-d)^2(-z_2)
\\
&=(c-d)^2(2\tau_2+z_2)+d^2(3z_2+2\tau_2+\sigma_2)+(c-2d)^2(-\tau_2-z_2)~.
\end{eqsp}
\begin{eqsp}
N=2:\quad  ac\tau_2+bd\sigma_2+(ad+bc)z_2&=ac(\tau_2 + z_2) + bd(\sigma_2 + z_2) + (a - b)(c -d)(-z_2)
\\&= bd(3z_2+2\tau_2+\sigma_2) +(a - b)(c - d)(2\tau_2+z_2) \\&+ (a - 2b)(c - 2d)(-\tau_2-z_2)~.
\end{eqsp}
Thus, for $\Omega\in\CR_2$ we have 
\begin{eqsp}
a^2\tau_2+b^2\sigma_2+2ab z_2 &\geq C_2(\Omega)~,
\\
c^2\tau_2+d^2\sigma_2+2cd z_2 &\geq C_2(\Omega)~,   
\end{eqsp}
where 
\begin{eqsp}\label{eq:def_C2}
   C_2(\Omega) := \begin{cases}
       {\rm Min} (\tau_2+z_2, \sigma_2+z_2, -z_2),& -\tau_2\leq z_2< 0,
       \\
       {\rm Min} (2\tau_2+z_2, 3z_2+2\tau_2+\sigma_2, -\tau_2-z_2),& -2\tau_2< z_2\leq  -\tau_2~.
   \end{cases} 
\end{eqsp}
Similarly
\begin{eqsp}\label{eq:bound_tau2N3}
    N=3:\quad    a^2\tau_2+b^2\sigma_2+2ab z_2&=
a^2(\tau_2+z_2)+b^2(\sigma_2+z_2)+(a-b)^2(-z_2)
\\
&=\frac{1}{2}b^2(5z_2+3\tau_2+2\sigma_2)+(a-b)^2(2z_2+3\tau_2)\\&+\frac{1}{2}(2a-3b)^2(-z_2-\tau_2)
\\
&=\frac{1}{2}b^2(7z_2+6\tau_2+2\sigma_2)+\frac{1}{2}(2a-3b)^2(z_2+2\tau_2)\\&+(a-2b)^2(-2z_2-3\tau_2)
\\
&=b^2(5z_2+6\tau_2+\sigma_2) + (a-3b)^2(-z_2-2\tau_2) \\&+ (a-2b)^2(z_2+3\tau_2)
\end{eqsp}
\begin{eqsp}
N=3:\quad   c^2\tau_2+d^2\sigma_2+2cd z_2&=
c^2(\tau_2+z_2)+d^2(\sigma_2+z_2)+(c-d)^2(-z_2)
\\
&=\frac{1}{2}d^2(5z_2+3\tau_2+2\sigma_2)+(c-d)^2(2z_2+3\tau_2)\\&+\frac{1}{2}(2c-3d)^2(-z_2-\tau_2)
\\
&=\frac{1}{2}d^2(7z_2+6\tau_2+2\sigma_2)+\frac{1}{2}(2c-3d)^2(z_2+2\tau_2)\\&+(c-2d)^2(-2z_2-3\tau_2)
\\
&=d^2(5z_2+6\tau_2+\sigma_2) + (c-3d)^2(-z_2-2\tau_2) \\&+ (c-2d)^2(z_2+3\tau_2)
\end{eqsp}
\begin{eqsp}
N=3:\quad   
ac\tau_2+bd\sigma_2+(ad+bc)z_2&=ac(\tau_2 + z_2) + bd(\sigma_2 + z_2) + (a - b)(c -d)(-z_2)
\\&=\frac{1}{2}bd(5z_2+3\tau_2+2\sigma_2) + (a - b)(c - d)(2z_2+3\tau_2)  \\&+\frac{1}{2}(2a - 3b)(2c - 3d)(-\tau_2-z_2)
\\
&=\frac{1}{2}bd(7z_2+6\tau_2+2\sigma_2) + \frac{1}{2}(2a - 3b)(2c - 3d))(z_2+2\tau_2) \\&+ (a - 2b)(c - 2d)(-2z_2-3\tau_2)
\\
&=bd(5z_2+6\tau_2+\sigma_2) + (a - 3b)(c - 3d)(-z_2-2\tau_2) \\&+ (a - 2b)(c - 2d)(z_2+3\tau_2)
\end{eqsp}
Thus, for $\Omega\in\CR_3$ we have 
\begin{eqsp}\label{eq:tau2,sig2_bound}
a^2\tau_2+b^2\sigma_2+2ab z_2 &\geq C_3(\Omega)~,
\\
c^2\tau_2+d^2\sigma_2+2cd z_2 &\geq C_3(\Omega)~,   
\end{eqsp}
where 
\begin{eqsp}\label{eq:def_C3}
   C_3(\Omega) := \begin{cases}
       {\rm Min} (\tau_2+z_2, \sigma_2+z_2, -z_2),& -\tau_2\leq z_2< 0,
       \\
       {\rm Min} (5z_2+3\tau_2+2\sigma_2, 3\tau_2+2z_2, -\tau_2-z_2),& -\frac{3}{2}\tau_2\leq  z_2\leq  -\tau_2,
       \\
       {\rm Min} (7z_2+6\tau_2+2\sigma_2, 2\tau_2+z_2, -3\tau_2-2z_2),&-2\tau_2\leq z_2\leq -\frac{3}{2}\tau_2,
       \\
      {\rm Min} (5z_2+6\tau_2+\sigma_2, 3\tau_2+z_2, -2\tau_2-z_2), &-3\tau_2<z_2\leq -2\tau_2~.
   \end{cases} 
\end{eqsp}
Arguing in the same way as in \cite{Bhand:2025ghn}, we see that the coefficients appearing in the expressions for $z_2'$ have the same sign and hence we have, for $N=2,3$, the bound
\begin{eqsp}\label{eq:z2'_bound}
    |ac\tau_2+bd\sigma_2+(ad+bc)z_2|\geq KC_N(\Omega)\mu(a,b,c,d)~,\quad \Omega\in\CR_N~,
\end{eqsp}
where $K>0$ is a constant and
\begin{eqsp}\label{eq:def_mu}
    \mu(a,b,c,d):=\mathrm{Max}(|a|, |b|, |c|, |d|)~.
\end{eqsp}
Clearly, $C_N(\Omega)$ is continuous and positive function. 
Then we have 
\begin{eqsp}\label{eq:f+_bound}
    |f_+(a^2\tau+b^2\sigma+2ab z)|\leq \sum_{p\geq 0}|f_p|e^{-2\pi(a^2\tau_2+b^2\sigma_2+2ab z_2)}\leq \tilde{f}_+(C_N(\Omega))~,
    \\
    |g_+(c^2\tau+d^2\sigma+2cd z)|\leq \sum_{q\geq 0}|g_q|e^{-2\pi(c^2\tau_2+d^2\sigma_2+2cd z_2)}\leq \tilde{g}_+(C_N(\Omega))~,
\end{eqsp}
where 
\begin{eqsp}
\tilde{f}_+(x):=\sum_{p\geq 0}|f_p|e^{-2\pi x}~,\quad   \tilde{g}_+(x):=\sum_{q\geq 0}|g_q|e^{-2\pi x}~,  
\end{eqsp}
and we used the fact that 
\begin{enumerate}
    \item the Fourier expansion for $f_+,g_+$ converges absolutely and uniformly on $\IH$, so that $\tilde{f}_+(x),\tilde{g}_+(x)$
are continuous functions on $\CR_N$,
\item $\tilde{f}_+(x),\tilde{g}_+(x)$ are monotonically increasing functions. 
\end{enumerate}
Since $C_N(\Omega)$ is also a continuous function, on any compact subset $\CO\subset\CR_N$, there exists constants  $C_{\text{upper}}^{f_+}(\CO),C_{\text{upper}}^{g_+}(\CO)>0$ such that 
\begin{eqsp}\label{eq:f+_bound_cpt_set}
    \tilde{f}_+(C_N(\Omega))\leq C_{\text{upper}}^{f_+}(\CO)~,\quad \tilde{g}_+(C_N(\Omega))\leq C_{\text{upper}}^{g_+}(\CO)~,\quad \Omega\in\CO~.
\end{eqsp} 
Following the proof of \cite[Proposition 4.1, Proposition 4.2]{Bhand:2023rhm}, we get the holomorphicity of $\CF_1$ in the $\CR_N$-chamber.
\par For the convergence of $\CF_2$, we follow the same strategy as in \cite{Bhand:2025ghn}. Following the steps in  \cite{Bhand:2025ghn}, we get 
\begin{eqsp}
    |\CF_2|\leq \CF_2^{>}+\CF_2^{<}~,
\end{eqsp}
where 
\begin{equation}
\begin{aligned}
\mathcal{F}_{2}^{<} &:= C_{\text{upper}}^{f_{+}}(\mathcal{O}) \sum_{\big{(}\begin{smallmatrix} a & b \\ c & d \end{smallmatrix}\big{)} \in \mathrm{P}\Gamma_1(N)} \lceil r_{0} \rceil H(z_{2}^{\prime}) e^{2\pi(\tau_{2}\sigma_{2}-z_{2}^{2})/N\tau_{2}^{\prime}} e^{-2\pi (\lceil r_{0} \rceil - r_{0}/N)|z_{2}^{\prime}|} \\
\mathcal{F}_{2}^{>} &:= C_{\text{upper}}^{f_{+}}(\mathcal{O}) \sum_{\big{(}\begin{smallmatrix} a & b \\ c & d \end{smallmatrix}\big{)}\in \mathrm{P}\Gamma_1(N)} \sum_{\substack{r \ge \lceil r_{0} \rceil + 1 \\ r \in \mathbb{Z}}} r H(z_{2}^{\prime}) e^{2\pi(\tau_{2}\sigma_{2}-z_{2}^{ 2})/N\tau_{2}^{\prime}} e^{-2\pi (r - r_{0}/N)|z_{2}^{\prime}|} ~,
\end{aligned}
\end{equation}
where 
\begin{eqsp}
    r_0:=\frac{z_2'}{\tau_2'}>0~,
\end{eqsp}
and $\r0$ is the lowest integer larger than or equal to $r_0$.
For the proof of convergence of $\CF_2^{>}$, we use the following bound 
\begin{eqsp}
    ac\tau_2+bd\sigma_2+(ad+bc)z_2\leq 2L\mu(a,b,c,d)^2E_N(\Omega)~,
\end{eqsp}
where $L$ is a constant $E_N(\Omega)$ is defined by replacing Min by Max in the definition of $C_N(\Omega)$. Note that $E_N(\Omega)$ is only piecewise continuous on $\CR_N$ but it has a finite number of jump discontinuities and hence is still bounded on compact subsets of $\CR_N$. With this input, the proof of the uniform convergence of $\CF_2^{>}$ is identical to \cite[Proposition 4.3]{Bhand:2025ghn}. To prove the convergence of $\CF_2^{<}$, we again define for $0<\eps<\frac{1}{2}$ 
\begin{equation}
\begin{aligned}
\mathcal{F}_{2}^{<\eps} &:= C_{\text{upper}}^{f_{+}}(\mathcal{O}) \sum_{\substack{\big{(}\begin{smallmatrix} a & b \\ c & d \end{smallmatrix}\big{)}\in \mathrm{P}\Gamma_1(N)\\r_0/\r0< \eps}} \lceil r_{0} \rceil H(z_{2}^{\prime}) e^{2\pi(\tau_{2}\sigma_{2}-z_{2}^{2})/N\tau_{2}^{\prime}} e^{-2\pi (\lceil r_{0} \rceil - r_{0}/N)|z_{2}^{\prime}|} \\
\mathcal{F}_{2}^{>\eps} &:= C_{\text{upper}}^{f_{+}}(\mathcal{O}) \sum_{\substack{\big{(}\begin{smallmatrix} a & b \\ c & d \end{smallmatrix}\big{)}\in \mathrm{P}\Gamma_1(N)\\r_0/\r0\geq \eps}} \r0 H(z_{2}^{\prime}) e^{2\pi(\tau_{2}\sigma_{2}-z_{2}^{ 2})/N\tau_{2}^{\prime}} e^{-2\pi (\r0 - r_{0}/N)|z_{2}^{\prime}|} ~,
\end{aligned}
\end{equation}
so that 
\begin{eqsp}
    \CF_2^{<}= \CF_2^{>\eps}+\CF_2^{<\eps}~.
\end{eqsp}
The proof of the convergence of $\CF_2^{<\eps}$ is identical to \cite[Proposition 4.4]{Bhand:2025ghn}. To prove the convergence of $\CF_2^{>\eps}$, we follow the same steps as in \cite[Proposition 4.5]{Bhand:2025ghn}. We put $n:=\r0$ and parametrize $\mathrm{P}\Gamma_1(N)$-matrices as $c=N(an-k),d=bNn-\ell$ requiring $-a\ell+bkN=1$. Then for $\frac{r0}{\r0}\geq \eps$, we find the bound 
\begin{eqsp}
   \left(\r0-\frac{r_0}{N}\right)&=r_0\tau_2\left(n-\frac{r_0}{N}\right) 
   \\&\geq \frac{\eps n\tau_2'}{N}\left(\frac{akN\tau_2+b\ell\sigma_2+(a\ell+bkN)z_2}{a^2\tau_2+b^2\sigma_2+2abz_2}\right)
   \\
   &=\frac{\eps n}{N}(akN\tau_2+b\ell\sigma_2+(a\ell+bkN)z_2)~.
\end{eqsp}
Then we get
\begin{eqsp}\label{eq:ser_>eps}
\CF_{2}^{>\eps}&\leq C^{f_+}_{\mr{upper}}(\CO)\exp\left(\frac{2\pi (\tau_2\sigma_2 - z_2^2)}{NC_N(\Omega)}\right)
\sum_{\big{(}\begin{smallmatrix} a & b\cr -kN & -\ell\end{smallmatrix}
    \big{)}\in \mr{P}\Gamma_1(N)} \sum_{n=1}^{\infty}
n ~H(akN\tau_2+b\,\ell\sigma_2+(bkN+a\,\ell)z_2) \, \\ & \hskip 2.5in
e^{-2\pi \, \frac{\eps\,n}{N} (a\,kN\tau_2+b\,\ell\sigma_2+(b\,kN+a\,\ell)z_2)}
\\&\leq C^{f_+}_{\mr{upper}}(\CO)\exp\left(\frac{2\pi (\tau_2\sigma_2 - z_2^2)}{NC_N(\Omega)}\right)\\&\times \sum_{\big{(}\begin{smallmatrix} a & b\cr -kN & -\ell\end{smallmatrix}
    \big{)}\in \mr{P}\Gamma_1(N)}
\, \left(e^{\pi \, \frac{\eps}{N}\, (a\,kN\tau_2+b\,\ell\sigma_2+(b\,kN+a\,\ell)z_2)}-
e^{-\pi \, \frac{\eps}{N}\, (a\,kN\tau_2+b\,\ell\sigma_2+(b\,kN+a\,\ell)z_2)}\right)^{-2}
\\&<\infty~.
\end{eqsp}
The proof of convergence of $\CF_3$ is similar.
\subsection{Convergence of ${\cal F}_4$}\label{sec:conv_CF3}

We shall now analyze the convergence of 
${\cal F}_4$ given in \refb{ef4}:
\begin{equation} \label{ef3exp}
\begin{split}
\CF_4&:=  f_{-1}g_{-1} \sum_{r>0} r\, 
\sum_{\big{(}\begin{smallmatrix} a & b\cr c & d\end{smallmatrix}\big{)}\in G^N_r\backslash \mr{P}\Gamma_1(N)^+}\hskip .1in 
 \bigg\{\prod_{n=-\infty}^\infty 
H(a_nc_n\tau_2 + b_nd_n\sigma_2 + (a_nd_n+b_nc_n)z_2)
\bigg\}
\\ & \hskip 1in \times \
e^{2\pi i \{(-a^2-c^2/N+r ac)\tau+(-b^2-d^2/N+r bd)\sigma+
(-2ab-2cd/N + r(ad+bc)) z\}}  \, ,
\end{split}    
\end{equation}
The summand in this expression
has absolute value $e^{2\pi B}$ where
\be \label{eBdef}
B = \sigma_2'/N+\tau_2'-r\, z_2'\, ,
\ee
and $\tau_2',\sigma_2',z_2'$ have been defined in \refb{edefprimenew}. The Heaviside functions are analytic in the interior of the $\CR_N$-chamber. Consequently, our goal is to show that the sum over $r$ and $a,b,c,d$ in the expression \refb{ef3exp} for ${\cal F}_4$ converges absolutely and uniformly on compact subsets of the $\CR_N$-chamber. 
We begin by analyzing the
constraints imposed by the Heaviside functions.

\begin{lemma}\label{lemma:Heaviside_matr_const}
Fix $N,r\in\IN$. Let $\begin{pmatrix} a & b\cr c & d\end{pmatrix}\in
\mathrm{P}\Gamma_1(N)^+$, $r\geq 1$ and $\begin{pmatrix} a_n & 
b_n\cr c_n & d_n\end{pmatrix}=\begin{pmatrix} 0 & -1/\sqrt{N}\cr \sqrt{N}& -r\sqrt{N}\end{pmatrix}^n\begin{pmatrix} a & b\cr c & d
\end{pmatrix}$. Let $a(n)$ be the sequence defined by the 
recursion relation:
\begin{eqsp} \label{e461x}
   a(0)=0, \quad a(1)=1,\quad a(n)=r\sqrt{N}a(n-1)-a(n-2)~,\quad n\geq 2~.
\end{eqsp}
Define the sets 
\begin{eqsp}\label{eq:def_S123}
    S_1&:=\left\{\frac{\sqrt{N}a(n)}{a(n+1)}:n\geq 0\right\}=\left\{0,{1\over r}, {rN\over r^2N-1}, {r^2N-1\over r(r^2N-2)}, \cdots \right\}~,
\\
    S_2&:=\left\{\frac{\sqrt{N}a(n+1)}{a(n)}:n\geq 0\right\}=\left\{\infty , rN, {r^2N-1\over rN}, {r(r^2N-2)\over r^2N-1}, \cdots \right\}~,
    \\
    S_3&:= \begin{cases} \{x: u_c<x<u_c' \}, \quad {\rm for} \ r> 2/\sqrt{N}\cr
    \varnothing, \quad{\rm for}  \ r\leq 2/\sqrt{N}\end{cases}\,~ ,
\end{eqsp}
where,
\begin{equation} \label{eucdef}
    u_c=\frac{N}{2}\left(r-\sqrt{r^2-\frac{4}{N}}\right)~,\quad u_c':=Nu_c^{-1}=\frac{N}{2}\left(r+\sqrt{r^2-\frac{4}{N}}\right)~.
\end{equation}
Then for $\Omega\in\CR_N$, if 
\begin{eqsp}\label{eq:heaviside_const_n}
    H(a_nc_n\tau_2 + b_nd_n\sigma_2 + (a_nd_n+b_nc_n)z_2)=1\quad\text{for all}\quad n\in\IZ
\end{eqsp}
then
\ben\label{erequirement}
&&  {{}_nc\over {}_na}, \ {{}_nc+2\,{}_nd\over {}_na+2\,{}_nb}, \in S_1\cup S_2\cup S_3~,
\een
where 
\begin{eqsp}
    \begin{pmatrix}
        {}_na&{}_nb\\{}_nc&{}_nd
    \end{pmatrix}=\begin{pmatrix}
        a&b\\c&d
    \end{pmatrix}\begin{pmatrix}
        0&1/\sqrt{N}\\-\sqrt{N}&-\sqrt{N}
    \end{pmatrix}^n~,\quad n\in \IZ~.
\end{eqsp}
\end{lemma}
\begin{proof}
We will give a geometric proof of this lemma. For $n=0$, \eqref{eq:heaviside_const_n} implies
\begin{eqnarray}\label{ez2zero}
    z_2'>0~,
\end{eqnarray}
where $\tau_2',\sigma_2',z_2'$ have been defined in \refb{edefprimenew}. 
Let us now define 
\be \label{expypdef}
x' := {z_2'\over \tau_2'}= {ac\tau_2+bd\sigma_2+(ad+bc)z_2\over
a^2\tau_2+b^2\sigma_2+2abz_2}, 
\qquad y' := {\sigma_2'\over \tau_2'}= {c^2\tau_2+d^2\sigma_2+2cdz_2 
\over a^2\tau_2+b^2\sigma_2+2abz_2}\, .
\ee
\refb{ez2zero} can then be rewritten as
\be
x'>0\, .
\ee
For all the Heaviside functions in \eqref{eq:heaviside_const_n} to be nonzero,  
the constraint $x'>0$ on $\tau_2',\sigma_2',z_2'$ must be satisfied if we replace 
$a,b,c,d$ in the expression for $x'$ by $a_n,b_n,c_n,d_n$ respectively. 
We can obtain these restrictions by mapping the line $x'=0$ in the $x'$-$y'$ plane by
the transformations
\be \label{e555}
g_{r,N} :=\begin{pmatrix} 0 & -1/\sqrt{N}\cr \sqrt{N} & -r\sqrt{N}\end{pmatrix}, \qquad
g_{r,N}^{-1} := \begin{pmatrix} 0 & -1/\sqrt{N}\cr \sqrt{N} & -r\sqrt{N}\end{pmatrix}^{-1}~.
\ee
Under the transformation
\be
\begin{pmatrix} \tau_2' & z_2'\cr z_2' & \sigma_2' \end{pmatrix}
\longmapsto  \tilde{g} \, \begin{pmatrix} \tau_2' & z_2'\cr z_2' & \sigma_2' \end{pmatrix}\, 
\tilde{g}^t, 
\qquad \tilde{g}:= \begin{pmatrix} \alpha & \beta\cr \gamma & \delta \end{pmatrix}\in\mathrm{P}\Gamma_1(N)^+\, ,
\ee
$(x',y')$ maps to
\be \label{e471y}
x ' \longmapsto  {\alpha\gamma + \beta \delta y' + (\alpha\delta + \beta\gamma) x'
\over \alpha^2 + \beta^2 y' + 2\alpha\beta x'}, 
\qquad  y' \longmapsto  {\gamma^2 + \delta^2 y' + 2\gamma \delta x'
\over \alpha^2 + \beta^2 y' + 2\alpha\beta x'}\, .
\ee
Using the condition $y'\ge (x')^2$ for points in $\IH_2$, we can see that 
\begin{eqsp}
\gamma^2 + \delta^2 y' + 2\gamma \delta x'>0~,\quad \alpha^2 + \beta^2 y' + 2\alpha\beta x'>0~.    
\end{eqsp}
Therefore for any $\tilde{g}$ of the form $g_{r,N}^n$, $n\in\mathbb{Z}$, 
we must have
\be \label{e472x}
\alpha\gamma + \beta \delta y' + (\alpha\delta + \beta\gamma) x' > 0\, .
\ee
These boundaries are
a subset of the straight lines of the form \eqref{epolesimaginary} some of which are shown in Figure~\ref{fig:RN-chamber1}--\ref{fig:RN-chamber3} except that now
we are drawing them in the $(x',y')$ plane instead of in the $(x,y)$ plane.

To solve the problem of determining the allowed region in the $(x',y')$ plane, we
label the lines obtained by setting the LHS of \refb{e472x} to zero 
by the points where they intersect the parabola 
$y'=(x')^2$. These can be specified by specifying the $x'$ value of the point, which
is the convention we shall follow. 
The only exception is the $x'=$~constant lines whose one end is at $y'=\infty$. We
shall denote this point by $\infty$. In this convention, the original $x'=0$ line,
connecting $\infty$ to 0,
will be represented as $[\infty,0]$ and the allowed region is to the left of this line.
The transformation \refb{e471y} maps this line to
$[\delta/\beta, \gamma/\alpha]$.

It is easy to prove recursively that 
\begin{eqsp} \label{e473x}
    g_{r,N}^n&=(-1)^{n+1}\begin{pmatrix}
        a(n-1)&-a(n)/\sqrt{N}\\a(n)\sqrt{N}&-a(n+1)
    \end{pmatrix}\in\mathrm{P}\Gamma_1(N)^+,\quad n\geq 1~,
    \\
    g_{r,N}^{-n}&=(-1)^{n+1}\begin{pmatrix}
        -a(n+1)&a(n)/\sqrt{N}\\-a(n)\sqrt{N}&a(n-1)
    \end{pmatrix}\in\mathrm{P}\Gamma_1(N)^+,\quad n\geq 1~,
\end{eqsp}
where the $a(n)$'s have been defined in \refb{e461x}.
Under \refb{e473x} the line $[\infty,0]$ maps to 
\be \label{e474x}
\left[\sqrt{N}{a(n+1)\over a(n)}, \sqrt{N}{a(n)\over a(n-1)}\right] \quad {\rm and} \quad 
\left[\sqrt{N}{a(n-1)\over a(n)}, \sqrt{N}{a(n)\over a(n+1)}\right]\, ,
\ee
respectively. These are lines connecting two successive points in the sets
$S_2$ and $S_1$ respectively. To understand the $n\to\infty$ limits of these
points 
we need
to find the fixed points of the transformation \refb{e471y} under $g_{r,N}$. 
Demanding that  $(x', y'=x^{\prime 2})$ maps to 
$(x', y'=x^{\prime 2})$ under the map \refb{e471y} for $\tilde{g}=g_{r,N}$, we get
\be
x' =  \frac{r\, x^{\prime 2} - x'}{\frac{x^{\prime 2}}{N}}\, .
\ee
This gives a quadratic equation for $x'$ with two solutions
\be 
x' = u_c, u_c'\, , 
\ee
with $u_c,u_c'$ given in \refb{eucdef}.

\begin{figure}
    \centering
\begin{tikzpicture}[scale=1.5, xscale=1.4] 

        % --- Axes ---
        % Horizontal Axis (x')
        \draw[->, thick] (-0.1, 0) -- (5, 0) node[right] {$x' := \frac{z'_2}{\tau'_2}$};
        \node[below] at (0,0) {$0$};
        
        % Vertical Axis (y')
        \draw[->, thick] (0, 0) -- (0, 5.5) node[above] {$y' := \frac{\sigma'_2}{\tau'_2}$};
        \draw[thick, red] (0, 0) -- (0, 5.3);

        % --- The Continuous Blue Curve ---
        \draw[blue, thick, domain=0:4.8, samples=200] plot (\x, {0.18*\x*\x});

        % --- Red Polygonal Path Segments ---
        \draw[thick, red] (0,0) -- (0.8, 0.115);
        \draw[thick, red] (0.8, 0.115) -- (1.6, 0.46); % Red line now stops here
        
        % [Gap here: No red lines between the second node, u_c, rN/2, and Nu_c^-1]
        
        % Right side segments resuming at Nu_c^-1
        \draw[thick, red] (3.4, 2.08) -- (3.8, 2.60);
        \draw[thick, red] (3.8, 2.60) -- (4.2, 3.17);
        \draw[thick, red] (4.2, 3.17) -- (4.5, 3.65);
        
        % Final vertical constraint line going straight up from the rN node
        \draw[thick, red] (4.5, 3.65) -- (4.5, 5.3);

        % --- All 8 Vertex Nodes (Uniform Bullets) ---
        \foreach \p in {
            (0.8, 0.115), % 1. 1/r
            (1.6, 0.46),  % 2. rN / (r^2N-1)
            (2.2, 0.87),  % 3. u_c
            (2.8, 1.41),  % 4. rN / 2
            (3.4, 2.08),  % 5. N u_c^{-1}
            (3.8, 2.60),  % 6. r(r^2N-2) / (r^2N-1)
            (4.2, 3.17),  % 7. (r^2N-1) / rN
            (4.5, 3.65)   % 8. rN (Corner Node)
        } {
            \node at \p {$\bullet$}; 
        }

        % --- Mathematical Label Placements ---
        \node[above, yshift=4pt] at (0.8, 0.115) {$\frac{1}{r}$};
        \node[above, yshift=4pt] at (1.6, 0.46) {$\frac{rN}{r^2N-1}$};
        \node[above left, xshift=2pt, yshift=4pt] at (2.2, 0.87) {$u_c$};
        \node[above, yshift=4pt] at (2.8, 1.41) {$\frac{rN}{2}$};
        \node[above left, xshift=2pt, yshift=4pt] at (3.4, 2.08) {$Nu_c^{-1}$};
        \node[above left, xshift=6pt, yshift=6pt] at (3.8, 2.60) {$\frac{r(r^2N-2)}{r^2N-1}$};
        \node[above left, xshift=6pt, yshift=6pt] at (4.2, 3.17) {$\frac{r^2N-1}{rN}$};
        \node[below right, xshift=1pt, yshift=-1pt] at (4.5, 3.65) {$rN$};

        % --- Adjusted Red Continuation Dots (Sits neatly along the curve) ---
        \node[red, font=\bfseries, rotate=22] at (1.9, 0.73) {$\dots$};
        \node[red, font=\bfseries, rotate=50] at (3.1, 1.83) {$\dots$}; % Tucked tighter to y = 0.18*x^2

    \end{tikzpicture}
    \caption{This figure shows the restriction 
    on $(\sigma_2',\tau_2',z_2')$ imposed by the Heaviside functions. 
    The Heaviside functions only allow the region bounded by the red lines. 
    There are infinite number of such red lines shown by the dots.
    The $x'$-axis has been rescaled appropriately to show close points on the 
    curve separately.}
    \label{fig:Heaviside_r=3}
\end{figure}
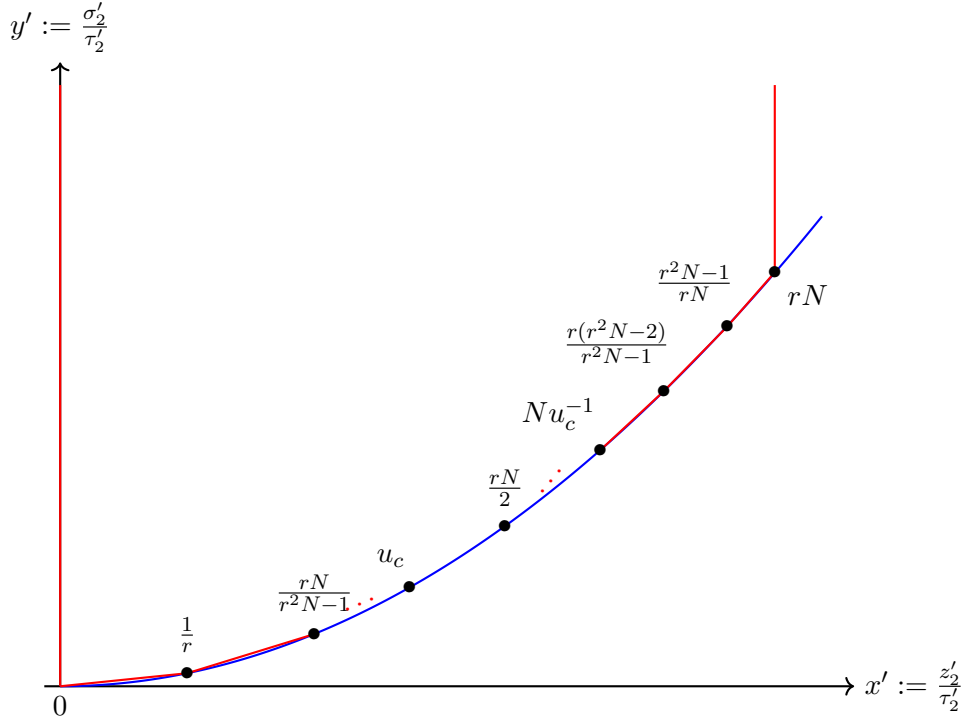

The lines described in \refb{e474x} 
have been shown in red in Figure \ref{fig:Heaviside_r=3}. 
The image of the
condition $x'>0$ translates to the condition that $(x',y')$ should lie above the red lines
appearing in Figure \ref{fig:Heaviside_r=3}, to the right of the line $x'=0$ and to the 
left of the
line $x'=rN$.
This shows that 
part of the parabola $y'=x^{\prime 2}$ 
between $x'=0$ and $x'=u_c$ and between
$x'=u_c'$ and $x'=r$ are removed by the Heaviside functions except for the isolated
points that appear in the set $S_1\cup S_2$ in \eqref{eq:def_S123}.

Our next task is to translate the condition on $(x',y')$ given above 
to a condition on the
integers $a,b,c,d$ using \refb{edefprimenew}.
A given $(a,b,c,d)$ will map the $\CR_N$ chamber to one of the
chambers in the $(x',y')$ plane. If this chamber lies in the allowed region
then the corresponding $(a,b,c,d)$ is allowed. Otherwise
it is excluded from the sum in \refb{ef3exp}.
To examine this, we need to determine the image of the vertices \eqref{eq:vert_RN_xy} of the $\CR_N$ chamber under $\begin{pmatrix}
    a&b\\c&d
\end{pmatrix}$. It is easy to see that a point $(x,x^2)$ on the parabola maps under \eqref{e471y} to the point 
\begin{eqsp}
    x'=\frac{c+dx}{a+bx}~,\quad y'= \left(\frac{c+dx}{a+bx}\right)^2~,
\end{eqsp}
on the parabola. Next, it is easy to check that the vertices \eqref{eq:vert_RN_xy} of the $\CR_N$-chamber from $-N$ to $-Nv_c$ and from $-v_c^{-1}$ to $\infty$ are shifted by one in the anticlockwise direction and the vertices in the range  $-Nv_c<x<-v_c^{-1}$ are shifted by one in the clockwise direction by 
\begin{eqsp}
    \tilde{g}_N:=\begin{pmatrix}
        0&1/\sqrt{N}\\-\sqrt{N}&-\sqrt{N}
    \end{pmatrix}~.
\end{eqsp}
Thus, to get the image of all (a subset of for $N\geq 7$ \cite{Sen:2007vb}) vertices of the $\CR_N$-chamber under $\begin{pmatrix}
        a&b\\c&d
    \end{pmatrix}$, we can find the image of $0$ and $-Nv_c<2<-v_c^{-1}$ under  $\begin{pmatrix}
        a&b\\c&d
    \end{pmatrix}\tilde{g}_N^n$. 
These are all points on the parabola $y'=x^{\prime 2}$. Allowed values of $(a,b,c,d)$
are those for which all the vertices lie in the allowed region in the $(x',y')$ plane.
From Figure \ref{fig:Heaviside_r=3} we see that this requires the $x'$ values 
of all the vertices to
either coincide with the set of points in the sets $S_1$ or $S_2$ or lie in the range
$u_c<x<u_c'$ in which case it is in the set $S_3$. 
This concludes the proof of the lemma.  
\end{proof}
\begin{remark}\label{rem:lemma_nec}
The condition \eqref{erequirement} is a necessary condition for \eqref{eq:heaviside_const_n} for all $N$ and also sufficient for $N\leq 6$. For $N\geq 7$, all the vertices cannot be reached by the action of $\tilde{g}_N$ \cite{Sen:2007vb}. Moreover, for $N>4,4|N$, $2$ is a vertex as well as an accumulation point \cite{Sen:2007vb}. Since 2 is a vertex for $N=2,3,4$, it is enough to require 
\begin{eqsp}
    \frac{{}_n c}{{}_n a}\in S_1\cup S_2\cup S_3~,
\end{eqsp}
in Lemma \ref{lemma:Heaviside_matr_const}.
\end{remark}
We now decompose the sum into $r\leq 2/\sqrt{N}$ and $r> 2/\sqrt{N}$ terms. So we define 
\begin{eqsp}
    \CF_4=\CF_4^{<}+\CF_4^{>}~,
\end{eqsp}
where 
\begin{eqsp} \label{e460xy}
\CF_4^{<}&=f_{-1}g_{-1} \sum_{r=1}^{2/\sqrt{N}} r\, 
\sum_{\big{(}\begin{smallmatrix} a & b\cr c & d\end{smallmatrix}\big{)}\in G^N_r\backslash \mr{P}\Gamma_1(N)^+}\hskip .1in 
 \bigg\{\prod_{n=-\infty}^\infty 
H(a_nc_n\tau_2 + b_nd_n\sigma_2 + (a_nd_n+b_nc_n)z_2)
\bigg\}
\\ & \hskip 1in \times \
e^{2\pi i \{(-a^2-c^2/N+r ac)\tau+(-b^2-d^2/N+r bd)\sigma+
(-2ab-2cd/N + r(ad+bc)) z\}}  \, ,
\\
\CF_4^{>}&=f_{-1}g_{-1} \sum_{r>2/\sqrt{N}} r\, 
\sum_{\big{(}\begin{smallmatrix} a & b\cr c & d\end{smallmatrix}\big{)}\in G^N_r\backslash \mr{P}\Gamma_1(N)^+}\hskip .1in 
 \bigg\{\prod_{n=-\infty}^\infty 
H(a_nc_n\tau_2 + b_nd_n\sigma_2 + (a_nd_n+b_nc_n)z_2)
\bigg\}
\\ & \hskip 1in \times \
e^{2\pi i \{(-a^2-c^2/N+r ac)\tau+(-b^2-d^2/N+r bd)\sigma+
(-2ab-2cd/N + r(ad+bc)) z\}}  \, .
\end{eqsp}
Note that for $N=2,3,4$, $\CF^{<}_4$ only has $r=1$ terms and the sum is empty for $N>4$. 
\begin{prop}\label{prop:F_3>_finite}
Let $N=2,3$. Then the sum $\CF_{4}^{<}$ is a finite sum. 
\begin{proof}
Our aim is to show that for each $N=2,3$, there are finitely many $\mathrm{P}\Gamma_1(N)^+$ matrices which satisfy all the constraints \eqref{eq:heaviside_const_n}. 
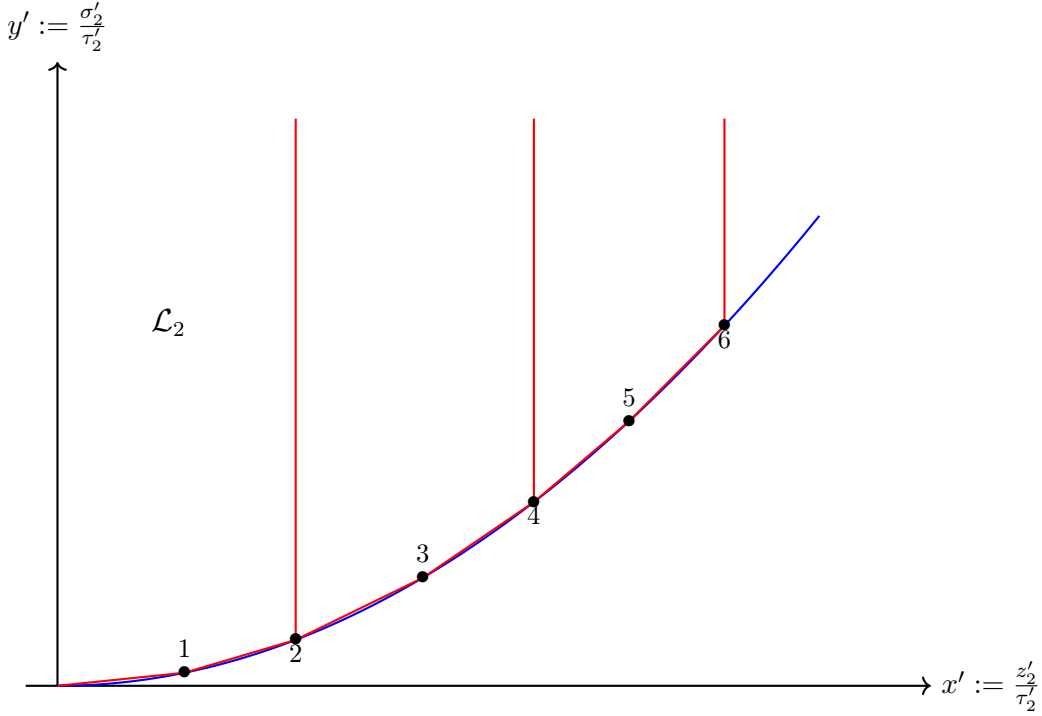
\begin{figure}[htbp]
    \centering
    \begin{tikzpicture}[scale=1.5, xscale=1.4] 

        % --- Axes ---
        % Horizontal Axis (x')
        \draw[->, thick] (-0.2, 0) -- (5.5, 0) node[right] {$x' := \frac{z'_2}{\tau'_2}$};
        
        % Vertical Axis (y') - Standard black arrow line on the left
        \draw[->, thick] (0, 0) -- (0, 5.5) node[above ] {$y' := \frac{\sigma'_2}{\tau'_2}$};

        % --- The Continuous Blue Curve ---
        \draw[blue, thick, domain=0:4.8, samples=200] plot (\x, {0.18*\x*\x});

        % --- Red Polygonal Path Lines ---
        % Tracing continuously from 0 up through node 6
        \draw[thick, red] (0,0) -- (0.8, 0.115);
        \draw[thick, red] (0.8, 0.115) -- (1.5, 0.405);
        \draw[thick, red] (1.5, 0.405) -- (2.3, 0.952);
        \draw[thick, red] (2.3, 0.952) -- (3.0, 1.62);
        \draw[thick, red] (3.0, 1.62) -- (3.6, 2.333);
        \draw[thick, red] (3.6, 2.333) -- (4.2, 3.175);

        % --- Vertical Red Walls (From nodes 2, 4, and 6) ---
        \draw[thick, red] (1.5, 0.405) -- (1.5, 5.0); % From node 2
        \draw[thick, red] (3.0, 1.62)  -- (3.0, 5.0); % From node 4
        \draw[thick, red] (4.2, 3.175) -- (4.2, 5.0); % From node 6

        % --- Node Bullets & Numerical Labels ---
        \foreach \p/\n/\pos in {
            (0.8, 0.115)/1/above,
            (1.5, 0.405)/2/below,
            (2.3, 0.952)/3/above,
            (3.0, 1.62)/4/below,
            (3.6, 2.333)/5/above,
            (4.2, 3.175)/6/below%
        } {
            \node at \p {$\bullet$};
            \node[\pos, yshift=2pt] at \p {\small \n};
        }

        % --- Region Label ---
        \node at (0.7, 3.2) {\large $\mathcal{L}_2$};

    \end{tikzpicture}
    \caption{The three chambers with one vertex at $\infty$ for $N=2,r=3$. There are
other chambers, not related to these by transformation by $g_{r,N},g_{r,N}^{-1}$, which lie close to the parabola, with
vertices lying between $u_c$ and $Nu_c^{-1}$.}
    \label{fig:chamber-L1}
\end{figure}
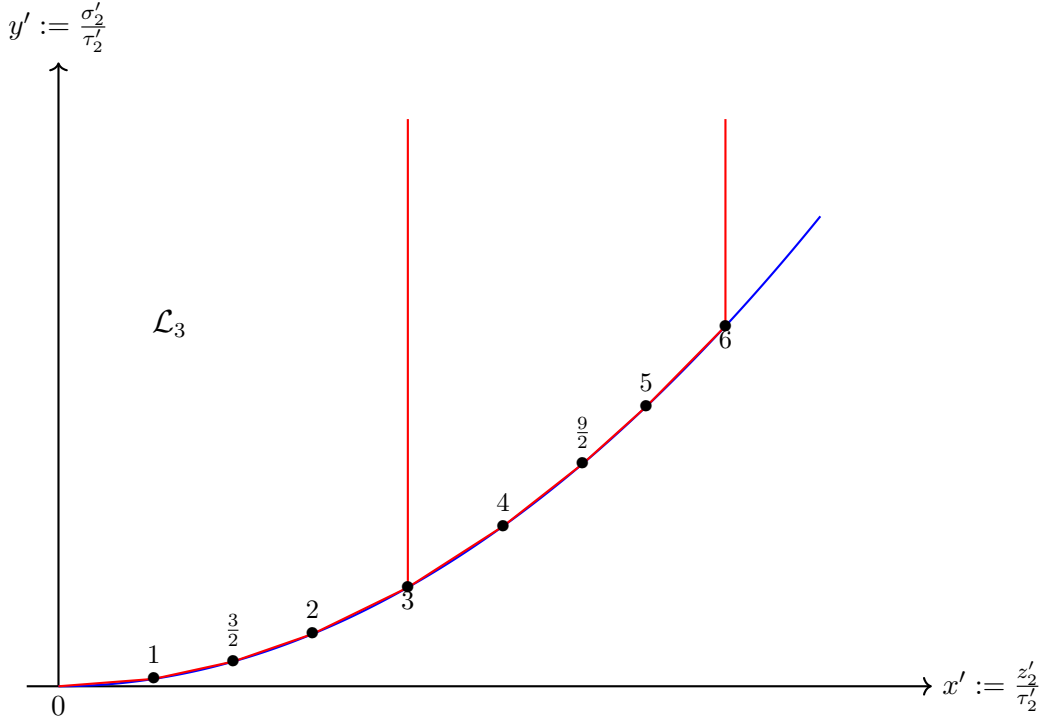
\begin{figure}[htbp]
    \centering
    \begin{tikzpicture}[scale=1.5, xscale=1.4] 

        % --- Axes ---
        % Horizontal Axis (x')
        \draw[->, thick] (-0.2, 0) -- (5.5, 0) node[right] {$x' := \frac{z'_2}{\tau'_2}$};
        \node[below] at (0,0) {$0$};
        
        % Vertical Axis (y')
        \draw[->, thick] (0, 0) -- (0, 5.5) node[above] {$y' := \frac{\sigma'_2}{\tau'_2}$};

        % --- The Continuous Blue Curve ---
        \draw[blue, thick, domain=0:4.8, samples=200] plot (\x, {0.18*\x*\x});

        % --- Red Polygonal Path Lines ---
        \draw[thick, red] (0,0) -- (0.6, 0.065);
        \draw[thick, red] (0.6, 0.065) -- (1.1, 0.218);
        \draw[thick, red] (1.1, 0.218) -- (1.6, 0.461);
        \draw[thick, red] (1.6, 0.461) -- (2.2, 0.871);
        \draw[thick, red] (2.2, 0.871) -- (2.8, 1.411);
        \draw[thick, red] (2.8, 1.411) -- (3.3, 1.96);
        \draw[thick, red] (3.3, 1.96) -- (3.7, 2.464);
        \draw[thick, red] (3.7, 2.464) -- (4.2, 3.175);

        % --- Vertical Red Walls (From nodes 3 and 6) ---
        \draw[thick, red] (2.2, 0.871) -- (2.2, 5.0); % From node 3
        \draw[thick, red] (4.2, 3.175) -- (4.2, 5.0); % From node 6

        % --- Node Bullets & Mathematical Labels ---
        \foreach \p/\n/\pos in {
            (0.6, 0.065)/1/above,
            (1.1, 0.218)/\frac{3}{2}/above,
            (1.6, 0.461)/2/above,
            (2.2, 0.871)/3/below,
            (2.8, 1.411)/4/above,
            (3.3, 1.96)/\frac{9}{2}/above,
            (3.7, 2.464)/5/above,
            (4.2, 3.175)/6/below%
        } {
            \node at \p {$\bullet$};
            \node[\pos, yshift=2pt] at \p {\small $\n$};
        }

        % --- Region Label ---
        \node at (0.7, 3.2) {\large $\mathcal{L}_3$};

    \end{tikzpicture}
    \caption{The three chambers with one vertex at $\infty$ for $N=3,r=2$. There are
other chambers, not related to these by transformation by $g_{r,N},g_{r,N}^{-1}$, which lie close to the parabola, with
vertices lying between $u_c$ and $Nu_c^{-1}$.}
    \label{fig:chamber-L3}
\end{figure}
Note that for $r\leq 2\sqrt{N}$, the set $S_3$ is empty. Next, from the identity 
\begin{eqsp}
    \gamma_{N}\tilde{g}_{N}\gamma_{N}^{-1}=g_{1,N}^{-1}~,
\end{eqsp}
the fact that $\gamma_N$ maps $\CR_N\leftrightarrow\CL_N$, and that $\tilde{g}_N$ permutes the vertices of $\CR_N$, we conclude that elements of $S_1\cup S_2$ for $r=1$ are the vertices of the $\CL_N$-chamber since the elements of $S_1,S_2$ are the action of $g_{1,N}^n,n\in\IZ$ on $\infty,0$ respectively. 
Explicitly, for $N=2,3$, we have:
\begin{eqsp}
    N=2:\quad & S_1\cup S_2=\left\{0,1,2,\infty\right\}~,
    \\
    N=3:\quad & S_1\cup S_2=\left\{0,1,\frac{3}{2},2,3,\infty\right\}~,
\end{eqsp}
Hence only those $\mathrm{P}\Gamma_1(N)^+$-matrices which map the $\CR_N$-chamber
to the $\CL_N$-chamber are allowed in the sum. Since $S_1\cup S_2$ is a finite set, there are only finitely many such $\mathrm{P}\Gamma_1(N)^+$-matrices and hence the sum is finite.
\end{proof}
\end{prop}
For the convergence of $\CF_4^{>}$, the idea is to generalize the analysis of the proof of Proposition \ref{prop:F_3>_finite} to $r\geq 2$ using Lemma \ref{lemma:Heaviside_matr_const}. Let us define 
\begin{eqsp} \label{e573}
    \CS_{12}&:=\left\{\begin{pmatrix}
        a&b\\c&d
    \end{pmatrix}\in\mathrm{P}\Gamma_1(N)^+:\left\{{{}_nc\over {}_na}\right\}\bigcap (S_1\cup S_2)\neq \varnothing\right\}~,
    \\
    \CS_{3}&:=\left\{\begin{pmatrix}
        a&b\\c&d
    \end{pmatrix}\in\mathrm{P}\Gamma_1(N)^+:{{}_nc\over {}_na}\in S_3\right\}~.
\end{eqsp}
For $N=2,3$, we have 
\begin{eqsp}\label{eq:nc/naN23}
    N&=2:\quad \left\{{{}_nc\over {}_na},n\in\IZ\right\}=\left\{\frac{d}{b},\frac{c-2 d}{a-2
   b},\frac{c-d}{a-b},\frac{c}{a}\right\}~,
   \\
   N&=3:\quad \left\{{{}_nc\over {}_na},n\in\IZ\right\}=\left\{\frac{2 c-3 d}{2 a-3
   b},\frac{c-d}{a-b},\frac{c}{a},\frac{d}{b},\frac{c-
   3 d}{a-3 b},\frac{c-2 d}{a-2 b}\right\}~.
\end{eqsp}
Then we can write 
\begin{eqsp}
    \CF_4^>=\CF_4^{>12}+\CF_4^{>3}~,
\end{eqsp}
where 
\begin{align}
\CF_4^{>12}&:=f_{-1}g_{-1} \sum_{r>2/\sqrt{N}} r\, 
\sum_{\big{(}\begin{smallmatrix} a & b\cr c & d\end{smallmatrix}\big{)}\in G^N_r\backslash \CS_{12}}
\exp\left[2\pi i \left\{\left(-a^2-\frac{c^2}{N}+r ac\right)\tau+\left(-b^2-\frac{d^2}{N}+r bd\right)\sigma\right.\right.\nonumber\\&\hspace{6cm}+
\left.\left.\left(-2ab-\frac{2cd}{N} + r(ad+bc)\right) z\right\}\right]  \, .\label{eq:CF3>12_def}
\\
\CF_4^{>3}&:=f_{-1}g_{-1} \sum_{r>2/\sqrt{N}} r\, 
\sum_{\big{(}\begin{smallmatrix} a & b\cr c & d\end{smallmatrix}\big{)}\in G^N_r\backslash \CS_3}
\exp\left[2\pi i \left\{\left(-a^2-\frac{c^2}{N}+r ac\right)\tau+\left(-b^2-\frac{d^2}{N}+r bd\right)\sigma\right.\right.\nonumber\\&\hspace{6cm}+
\left.\left.\left(-2ab-\frac{2cd}{N} + r(ad+bc)\right) z\right\}\right]  \, .
\end{align}
\begin{prop}
The series $\CF_4^{>12}$ converges absolutely and uniformly on compact subsets of the $\CR_N$-chamber for $N=2,3$.     
\end{prop}
\begin{proof}
We first show that, for a fixed $r$, the inner sum is a finite sum. 
For a fixed $r$, consider a typical term corresponding to $\begin{pmatrix} a & b\cr c & d\end{pmatrix}\in G_r^N\backslash \CS_{12}$. Let $x',y'$ be as in \eqref{expypdef}. Then 
since one element of $\left\{{}_nc\over {}_na\right\}$ is an element of $S_1\cup S_2$, we can left 
multiply $\begin{pmatrix} a & b\cr c & d\end{pmatrix}$ 
by $g_{r,N}^n$ for some $n$ such that $(x',y')$ belongs to a chamber with one vertex at $\infty$. 
There are exactly
$r$ chambers with one vertex at infinity, the other vertices being at 
\begin{eqsp}
 N&=2:\quad   x'=\left\{0,1,2\right\},~~\{2,3,4\},~~\dots,~~\{2r-2,2r-1,2r\}~,
 \\
 N&=3:\quad   x'=\left\{0,1,\frac{3}{2},2,3\right\},~~\left\{3,4,\frac{9}{2},5,6\right\},~~\left\{3r-3,3r-2,3r-\frac{3}{2},3r-1,3r\right\}~.
\end{eqsp}
This has been shown in Figure~\ref{fig:chamber-L1} for the case $N=2,r=3$ and in Figure \ref{fig:chamber-L3} for the case $N=3,r=2$. There are $\mathrm{P}\Gamma_1(N)^+$ matrices which permute these vertices of the $\CL_N$-chamber:
\begin{eqsp}\label{eq:ver_perm_matr}
    N&=2:\quad 
\begin{pmatrix} 0 & -\frac{1}{\sqrt{2}} \\ \sqrt{2} & 0 \end{pmatrix}, \begin{pmatrix} 1 & 1 \\ 0 & 1 \end{pmatrix}, \begin{pmatrix} \sqrt{2} & \frac{1}{\sqrt{2}} \\ \sqrt{2} & \sqrt{2} \end{pmatrix}, \begin{pmatrix} 1 & 0 \\ 2 & 1 \end{pmatrix} 
\\
N&=3:\quad 
\begin{pmatrix} 1 & 0 \\ 3 & 1 \end{pmatrix}, \begin{pmatrix} 0 & -\frac{1}{\sqrt{3}} \\ \sqrt{3} & 0 \end{pmatrix}, \begin{pmatrix} 1 & 1 \\ 0 & 1 \end{pmatrix}, \begin{pmatrix} \sqrt{3} & \frac{2}{\sqrt{3}} \\ \sqrt{3} & \sqrt{3} \end{pmatrix}, \begin{pmatrix} 2 & 1 \\ 3 & 2 \end{pmatrix}, \begin{pmatrix} \sqrt{3} & \frac{1}{\sqrt{3}} \\ 2\sqrt{3} & \sqrt{3} \end{pmatrix}~.
\end{eqsp}
Taking into account the cyclic permutation of the vertices by these matrices, we get at most 
$2Nr$ independent matrices $\begin{pmatrix} a & b\cr c & d\end{pmatrix}$  
that are not related by left multiplication by the matrices \refb{e555}.  The actual number might be less. 
Let $(x',y')$
be the image of $(x,y)\in\CR_N$ in the chamber $\CL_N$. 
We denote the images of $(\tau,\sigma,z)$ under the matrices in \eqref{eq:ver_perm_matr} by 
\begingroup
\allowdisplaybreaks
\begin{equation}\label{eq:tsz_primes_def}
\begin{split}
&N=2:\quad 
\begin{pmatrix} 0 & -\frac{1}{\sqrt{2}} \\ \sqrt{2} & 0 \end{pmatrix}:\begin{pmatrix}
    \tau&z\\z&\sigma
\end{pmatrix}\longmapsto \begin{pmatrix}
    \tau^{21}&z^{21}\\z^{21}&\sigma^{21}
\end{pmatrix}:=\begin{pmatrix} \frac{\sigma}{2} & -z \\ -z & 2\tau \end{pmatrix}~, 
\\&N=2:\quad \begin{pmatrix} 1 & 1 \\ 0 & 1 \end{pmatrix}:\begin{pmatrix}
    \tau&z\\z&\sigma
\end{pmatrix}\longmapsto \begin{pmatrix}
    \tau^{22}&z^{22}\\z^{22}&\sigma^{22}
\end{pmatrix}:=\begin{pmatrix} 2z + \sigma + \tau & z + \sigma \\ z + \sigma & \sigma \end{pmatrix}~,
\\&N=2:\quad \begin{pmatrix} \sqrt{2} & \frac{1}{\sqrt{2}} \\ \sqrt{2} & \sqrt{2} \end{pmatrix}:\begin{pmatrix}
    \tau&z\\z&\sigma
\end{pmatrix}\longmapsto \begin{pmatrix}
    \tau^{23}&z^{23}\\z^{23}&\sigma^{23}
\end{pmatrix}:=\begin{pmatrix} \frac{1}{2}(4z + \sigma + 4\tau) & 3z + \sigma + 2\tau \\ 3z + \sigma + 2\tau & 2(2z + \sigma + \tau) \end{pmatrix}, 
\\&N=2:\quad \begin{pmatrix} 1 & 0 \\ 2 & 1 \end{pmatrix}:\begin{pmatrix}
    \tau&z\\z&\sigma
\end{pmatrix}\longmapsto \begin{pmatrix}
    \tau^{24}&z^{24}\\z^{24}&\sigma^{24}
\end{pmatrix}:=\begin{pmatrix} \tau & z + 2\tau \\ z + 2\tau & 4z + \sigma + 4\tau \end{pmatrix}
\\
&N=3:\quad 
\begin{pmatrix} 1 & 0 \\ 3 & 1 \end{pmatrix}:\begin{pmatrix}
    \tau&z\\z&\sigma
\end{pmatrix}\longmapsto \begin{pmatrix}
    \tau^{31}&z^{31}\\z^{31}&\sigma^{31}
\end{pmatrix}:=\begin{pmatrix} \tau & z + 3\tau \\ z + 3\tau & 6z + \sigma + 9\tau \end{pmatrix},
\\
&N=3:\quad \begin{pmatrix} 0 & -\frac{1}{\sqrt{3}} \\ \sqrt{3} & 0 \end{pmatrix}:\begin{pmatrix}
    \tau&z\\z&\sigma
\end{pmatrix}\longmapsto \begin{pmatrix}
    \tau^{32}&z^{32}\\z^{32}&\sigma^{32}
\end{pmatrix}:=\begin{pmatrix} \frac{\sigma}{3} & -z \\ -z & 3\tau \end{pmatrix}, \\
&N=3:\quad \begin{pmatrix} 1 & 1 \\ 0 & 1 \end{pmatrix}:\begin{pmatrix}
    \tau&z\\z&\sigma
\end{pmatrix}\longmapsto \begin{pmatrix}
    \tau^{33}&z^{33}\\z^{33}&\sigma^{33}
\end{pmatrix}:=\begin{pmatrix} 2z + \sigma + \tau & z + \sigma \\ z + \sigma & \sigma \end{pmatrix}, 
\\
&N=3:\quad \begin{pmatrix} \sqrt{3} & \frac{2}{\sqrt{3}} \\ \sqrt{3} & \sqrt{3} \end{pmatrix}:\begin{pmatrix}
    \tau&z\\z&\sigma
\end{pmatrix}\longmapsto \begin{pmatrix}
    \tau^{34}&z^{34}\\z^{34}&\sigma^{34}
\end{pmatrix}:=\begin{pmatrix} 4z + \frac{4\sigma}{3} + 3\tau & 5z + 2\sigma + 3\tau \\ 5z + 2\sigma + 3\tau & 3(2z + \sigma + \tau) \end{pmatrix}, \\
&N=3:\quad \begin{pmatrix} 2 & 1 \\ 3 & 2 \end{pmatrix}:\begin{pmatrix}
    \tau&z\\z&\sigma
\end{pmatrix}\longmapsto \begin{pmatrix}
    \tau^{35}&z^{35}\\z^{35}&\sigma^{35}
\end{pmatrix}:=\begin{pmatrix} 4z + \sigma + 4\tau & 7z + 2\sigma + 6\tau \\ 7z + 2\sigma + 6\tau & 12z + 4\sigma + 9\tau \end{pmatrix}, 
\\
&N=3:\quad \begin{pmatrix} \sqrt{3} & \frac{1}{\sqrt{3}} \\ 2\sqrt{3} & \sqrt{3} \end{pmatrix}:\begin{pmatrix}
    \tau&z\\z&\sigma
\end{pmatrix}\longmapsto \begin{pmatrix}
    \tau^{36}&z^{36}\\z^{36}&\sigma^{36}
\end{pmatrix}:=\begin{pmatrix} \frac{1}{3}(6z + \sigma + 9\tau) & 5z + \sigma + 6\tau \\ 5z + \sigma + 6\tau & 3(4z + \sigma + 4\tau) \end{pmatrix}~.
\end{split}
\end{equation}
\endgroup
The images in the other chambers with a vertex at
$\infty$ can be generated from $\CL_N$ by the action of the matrices
\be
\begin{pmatrix} 1 & 0\cr nN & 1\end{pmatrix}\, , \qquad 1\le n\le r-1\, .
\ee
Indeed, using \refb{esl2zonxy} we see that under this map,
\begin{eqsp}
    N=2:\quad &(0,1,2,\infty)\mapsto (2n,2n+1,2n+2,\infty)~,
    \\
    N=3:\quad & \left(0,1,\frac{3}{2},2,3,\infty\right)\mapsto \left(3n,3n+1,3n+\frac{3}{2},3n+2,3n+3,\infty\right)~.
\end{eqsp}
Then we have 
\begin{eqsp}\label{eq:CSr_def}
G_r^2\backslash\CS_{12}&\subset \left\{\begin{pmatrix} 1 & 0\cr 2n & 1\end{pmatrix}\begin{pmatrix} 0 & -\frac{1}{\sqrt{2}} \\ \sqrt{2} & 0 \end{pmatrix}, \begin{pmatrix} 1 & 0\cr 2n & 1\end{pmatrix}\begin{pmatrix} 1 & 1 \\ 0 & 1 \end{pmatrix}, \begin{pmatrix} 1 & 0\cr 2n & 1\end{pmatrix}\begin{pmatrix} \sqrt{2} & \frac{1}{\sqrt{2}} \\ \sqrt{2} & \sqrt{2} \end{pmatrix},\right. \\&\hspace{6cm}\left.\begin{pmatrix} 1 & 0\cr 2n & 1\end{pmatrix}\begin{pmatrix} 1 & 0 \\ 2 & 1 \end{pmatrix} :0\leq n\leq r-1\right\}~,
\\
G_r^3\backslash\CS_{12}&\subset \left\{\begin{pmatrix} 1 & 0\cr 3n & 1\end{pmatrix}\begin{pmatrix} 1 & 0 \\ 3 & 1 \end{pmatrix}, \begin{pmatrix} 1 & 0\cr 3n & 1\end{pmatrix}\begin{pmatrix} 0 & -\frac{1}{\sqrt{3}} \\ \sqrt{3} & 0 \end{pmatrix}, \begin{pmatrix} 1 & 0\cr 3n & 1\end{pmatrix}\begin{pmatrix} 1 & 1 \\ 0 & 1 \end{pmatrix},\begin{pmatrix} 1 & 0\cr 3n & 1\end{pmatrix}\begin{pmatrix} \sqrt{3} & \frac{2}{\sqrt{3}} \\ \sqrt{3} & \sqrt{3} \end{pmatrix}, \right.
\\
&\hspace{4.5cm}\left.\begin{pmatrix} 1 & 0\cr 3n & 1\end{pmatrix}\begin{pmatrix} 2 & 1 \\ 3 & 2 \end{pmatrix}, \begin{pmatrix} 1 & 0\cr 3n & 1\end{pmatrix}\begin{pmatrix} \sqrt{3} & \frac{1}{\sqrt{3}} \\ 2\sqrt{3} & \sqrt{3} \end{pmatrix}:0\leq n\leq r-1\right\}~,
\end{eqsp}
where the containment is because of the fact that there may be further
identification between these matrices by left multiplication by powers of 
\refb{e555}. The left multiplication of
$\begin{pmatrix} a & b\cr c & d\end{pmatrix}$ by $\begin{pmatrix} 1 & 0\cr nN & 1\end{pmatrix}$
maps 
\be
(\tau_2',\sigma_2',z_2') \to (\tau_2', \sigma'_2+n^2N^2\tau'_2+ 2nN z'_2, z'_2+nN\tau'_2)\, .
\ee
Using this and \eqref{eq:CSr_def}, we thus have 
\begin{eqsp}\label{e487x}
\left|\CF_4^{>12}\right|&\leq f_{-1}g_{-1} \sum_{r=2}^{\infty} r 
\sum_{n=0}^{r-1}\sum_{i=1}^{2N}\left[e^{2\pi \{\tau^{Ni}_2+(\sigma^{Ni}_2+n^2N^2\tau^{Ni}_2+2nNz^{Ni}_2)/N-r(z^{Ni}_2+nN\tau^{Ni}_2)\}}\right]
\\
&= f_{-1}g_{-1} \sum_{r=2}^{\infty} r 
\sum_{n=0}^{r-1}\sum_{i=1}^{2N} 
\left[e^{2\pi\{\sigma^{Ni}_2/N+\tau^{Ni}_2-nN(r-n)\tau^{Ni}_2
+ (2n-r)\, z^{Ni}_2\}}\right]~.
\end{eqsp}
We now use the following bound for $r\geq 2$:
\begin{eqsp} \label{e496yx}
e^{2\pi\{\sigma^{Ni}_2/N+\tau^{Ni}_2-nN(r-n)\tau^{Ni}_2
+ (2n-r)\, z^{Ni}_2\}}&= e^{2\pi \{\sigma^{Ni}_2/N+\tau^{Ni}_2- n (r-n) (N\tau^{Ni}_2- z^{Ni}_2)
- ((n+1)(r-n-1)+1) z^{Ni}_2\}}  
\\
&\leq 
 e^{2\pi(\sigma^{Ni}_2/N+\tau^{Ni}_2- C_{Ni}\, r)}\ \hbox{for $0\le n\le r-1$}~,
\end{eqsp}
where 
\begin{eqsp}
    C_{Ni} := \min\{N\tau^{Ni}_2-z^{Ni}_2, z^{Ni}_2\}\, .
\end{eqsp}
Explicitly, we have 
\begin{equation}\label{eq:CNi_explicit}
\begin{aligned}
C_{21} &= \min \left\{ \sigma_{2}+z_{2} \,,\, -z_{2} \right\} , \\
C_{22} &= \min \left\{ 3z_{2}+\sigma_{2}+2\tau_{2} \,,\, z_{2}+\sigma_{2} \right\} , \\
C_{23} &= \min \left\{ 2\tau_{2}+z_{2} \,,\, 3z_{2}+\sigma_{2}+2\tau_{2} \right\} , \\
C_{24} &= \min \left\{ -z_{2} \,,\, 2\tau_{2}+z_{2} \right\} , \\
C_{31} &= \min \left\{ -z_{2} \,,\, 3\tau_{2}+z_{2} \right\} , \\
C_{32} &= \min \left\{ \sigma_{2}+z_{2} \,,\, -z_{2} \right\} , \\
C_{33} &= \min \left\{ 5z_{2}+3\tau_{2}+2\sigma_{2} \,,\, z_{2}+\sigma_{2} \right\} , \\
C_{34} &= \min \left\{ 7z_{2}+6\tau_{2}+2\sigma_{2} \,,\, 5z_{2}+2\sigma_{2}+3\tau_{2} \right\} , \\
C_{35} &= \min \left\{ 5z_{2}+6\tau_{2}+\sigma_{2} \,,\, 7z_{2}+2\sigma_{2}+6\tau_{2} \right\} , \\
C_{36} &= \min \left\{ z_{2}+3\tau_{2} \,,\, 5z_{2}+\sigma_{2}+6\tau_{2} \right\} .
\end{aligned}
\end{equation}
From this, and \eqref{echamber}, we easily check that $C_{Ni}>0$
for $\begin{pmatrix}
    \tau&z\\z&\sigma
\end{pmatrix}\in\CR_N$ for $N=2,3$. Thus we get 
\begin{eqsp} \label{e497yx}
\left|\CF_4^{>12}\right|&\leq f_{-1}g_{-1} \sum_{r=2}^{\infty} r 
\sum_{n=0}^{r-1}\sum_{i=1}^{2N} 
\left[e^{2\pi\{\sigma^{Ni}_2/N+\tau^{Ni}_2-nN(r-n)\tau^{Ni}_2
+ (2n-r)\, z^{Ni}_2\}}\right]
\\&=\sum_{r=2}^{\infty} r^2\sum_{i=1}^{2N}
\left[e^{2\pi(\sigma^{Ni}_2/N+\tau^{Ni}_2- C_{Ni}\, r)}\right]
\\&<\infty~.    
\end{eqsp}
The absolute and uniform convergence on compact subsets of the $\CR_N$-chamber follows from the continuity of $C_{Ni}$ in the $\CR_N$-chamber.
\end{proof}
We are left to prove the convergence of $\CF_4^{>3}$. 
With this in mind, we prove the following lemma.
\begin{lemma}\label{lemma:coset_reps}
We have 
\begin{eqsp}
    G^N_r\backslash\CS_3= \tilde{\CS}_3:=\left\{\begin{pmatrix}
        a&b\\c&d
    \end{pmatrix}\in\CS_3: {2\over r} \leq  {c\over a} < {rN\over 2}\right\}~. 
\end{eqsp}
\end{lemma}
\begin{proof}
It follows from \refb{expypdef} and \refb{e555} that on the parabola $y'=x^{\prime 2}$,
if we denote by $(x_n,y_n=x_n^2)$ the image of $(x',y'=x^{\prime 2})$ under the action of
$g_{r,N}^n,n>0$, then we have
\be\label{eurecurx}
x_{n+1}= N(r - x_n^{-1}), \qquad x_{n-1} = N(rN-x_n)^{-1} \, .
\ee
Using $u_c+Nu_c^{-1}=rN$, this gives,
\be\label{einequalx}
x_{n+1}-u_c = N\left({x_n-u_c\over  x_n u_c}\right) >  x_n-u_c\, ,\qquad
Nu_c^{-1}-x_{n+1} = {Nu_c^{-1}-x_n\over x_nu_c^{-1}}<Nu_c^{-1}-x_n 
\ee
as long as $u_c<x_n<Nu_c^{-1}$. 
This shows that 
at every step, $x$ is driven away from $u_c$ towards $Nu_c^{-1}$.
Conversely, the action of $g_{r,N}^{-1}$ will drive $x$ towards $u_c$. 
Furthermore, considering the first equation in \refb{e471y} with $y'=x'{}^2$ as a function $f_{g_r}:\IR\to\IR$, we have
\begin{eqsp}
    f_{g_r}(x)=rN-\frac{N}{x}~.
\end{eqsp}  
We see that $f_{g_r}$ is a continuous function on $(u_c,Nu_c^{-1})$. Moreover, 
\begin{eqsp} \label{egderivative}
    f_{g_r}'(x)=\frac{1}{x^2}>0~,\quad x\neq 0~,
\end{eqsp}
showing that $f_{g_r}(x)$ is a monotonically increasing function of $x$. Also,
\begin{eqsp} \label{egspecial}
   f_{g_r}(u_c)=u_c~,\quad  f_{g_r}(2/r)=rN/2~,\quad f_{g_r}(Nu_c^{-1})=Nu_c^{-1}~.
\end{eqsp}
This shows that the region $u_c<x\le 2/r$ is mapped to the region
$u_c<x\le rN/2$ under $g_{r,N}$ and the region $rN/2<x<Nu_c^{-1}$ is mapped to the region
$2/r<x<Nu_c^{-1}$ under $g_{r,N}^{-1}$.  Furthermore, using \refb{egderivative} and \refb{egspecial}
it is easy to see that any point in the interval $2/r\leq x<rN/2$ is mapped outside this range by the action of $g_{r,N}$
and $g_{r,N}^{-1}$. 
%Now, since $g$ drives $x$
%towards $u_{c}^{-1}$ and $g^{-1}$ drives $x$ towards $u_c$, 
This shows that any $x$ in the
range $u_c<x<Nu_c^{-1}$ can be mapped to a unique point in the range 
$2/r\leq x<rN/2$ by successive action of $g_{r,N}$ or $g_{r,N}^{-1}$.
This completes the proof of the lemma.     
\end{proof}
We now want to find a bound for the absolute value of the summand of $\CF_4^{>3}$. Let us write 
\begin{eqsp}
    \left|e^{2\pi i \{(-a^2-c^2/N+r ac)\tau+(-b^2-d^2/N+r bd)\sigma+
(-2ab-2cd/N + r(ad+bc)) z\}}\right|=e^{2\pi B}~,
\end{eqsp}
where 
\begin{eqsp}
B&:= \sigma_2'/N+\tau_2'-r z_2' 
\\&= (a^2 + c^2/N - r ac)\tau_2 + (b^2 + d^2/N - rbd)\sigma_2 - 
\{  r (ad+bc) - 2ab - 2cd/N\} z_2\, .
\end{eqsp}
We have the following upper bound for $B$.
\begin{lemma}\label{lemma:bound_B}
For $N=2,3$, $\begin{pmatrix} a & b\cr c & d\end{pmatrix}\in \tilde{\CS}_3$, we have
\begin{eqsp}\label{ebound}
    B\leq \frac{L}{2N}\left(\frac{2}{r}-\frac{rN}{2}+N\left(1-\frac{2}{3}\delta_{N,2}\delta_{r,2}-\frac{11}{18}\delta_{N,3}\delta_{r,2}\right)\right)C_N(\Omega)\mu(a,b,c,d)<0~,
\end{eqsp}
for some $L>0$, where $C_N(\Omega),\mu(a,b,c,d)$ is defined in \eqref{eq:def_C2},\eqref{eq:def_C3} and \eqref{eq:def_mu}.
\end{lemma}
\begin{proof}
We first express $B$ in various ways. We have 
\begin{eqsp}\label{eq:B_exp_N2}
B&=\frac{ac}{N} \left( {aN\over c} + {c\over a} - rN\right)\, (\tau_2 + z_2)
+ \frac{bd}{N} \left( {bN\over d} +{d\over b}
- rN\right)\, (\sigma_2+z_2) \\&+ \,\frac{(a-b)(c-d)}{N} \left(  {N(a-b)\over (c-d)}+ 
{c-d\over a-b}-rN \right) (-z_2)  \\&=    \frac{(a-b)(c-d)}{N} \left(  {N(a-b)\over (c-d)}+ 
{c-d\over a-b}-rN \right) (2\tau_2+z_2) + \frac{bd}{N} \left( {bN\over d} +{d\over b}
- rN\right)(3z_2+2\tau_2+\sigma_2) \\&+ \frac{(a-2b)(c-2d)}{N}\left( \frac{N(a-2b)}{(c-2d)} + \frac{(c-2d)}{(a-2b)} - rN \right)(-\tau_2-z_2)
\end{eqsp}
which will be useful for $N=2$ and 
\begin{eqsp}\label{eq:B_exp_N3}
B&=\frac{ac}{N} \left( {aN\over c} + {c\over a} - rN\right)\, (\tau_2 + z_2)
+ \frac{bd}{N} \left( {bN\over d} +{d\over b}
- rN\right)\, (\sigma_2+z_2) \\&+ \,\frac{(a-b)(c-d)}{N} \left(  {N(a-b)\over (c-d)}+ 
{c-d\over a-b}-rN \right) (-z_2)  \\
&=\frac{bd}{2N}\left( \frac{bN}{d} + \frac{d}{b} - rN\right)(5z_2+3\tau_2+2\sigma_2) \\&+ \frac{(a-b)(c-d)}{N}\left( {N(a-b)\over (c-d)}+ 
{c-d\over a-b}-rN \right)(2z_2+3\tau_2) \\&+ \frac{(2a-3b)(2c-3d)}{2N}\left( \frac{N(2a-3b)}{(2c-3d)} + \frac{(2c-3d)}{(2a-3b)} - rN \right)(-\tau_2-z_2)
\\&=\frac{bd}{2N}\left( \frac{bN}{d} + \frac{d}{b} - rN\right)(7z_2+6\tau_2+2\sigma_2) \\&+\frac{(a-2b)(c-2d)}{N}\left( \frac{N(a-2b)}{(c-2d)} + \frac{(c-2d)}{(a-2b)} - rN \right)(-2z_2-3\tau_2) \\&+\frac{(2a-3b)(2c-3d)}{2N}\left( \frac{N(2a-3b)}{(2c-3d)} + \frac{(2c-3d)}{(2a-3b)} - rN \right)(z_2+2\tau_2)  
\\&=\frac{bd}{N} \left( {bN\over d} +{d\over b}
- rN\right)(5z_2+6\tau_2+\sigma_2) \\&+ \frac{(a-3b)(c-3d)}{N}\left( \frac{N(a-3b)}{(c-3d)} + \frac{(c-3d)}{(a-3b)} - rN \right)(-z_2-2\tau_2)\\& + \frac{(a-2b)(c-2d)}{N}\left( \frac{N(a-2b)}{(c-2d)} + \frac{(c-2d)}{(a-2b)} - rN \right)(z_2+3\tau_2)
\end{eqsp}
which will be useful for $N=3$.
For $\begin{pmatrix} a & b\cr c & d\end{pmatrix}
\in \tilde{\CS}_3$,
we get
\be\label{e4102xy}
\frac{c}{a}+\frac{aN}{c}-rN \le {2\over r}+{rN\over 2} - rN = {2\over r} - {rN\over 2}\, .
\ee
As in \cite{Bhand:2025ghn}, using \eqref{eq:nc/naN23}, all the denominators in the coefficients in \eqref{eq:B_exp_N2} and \eqref{eq:B_exp_N3} are nonzero and the two terms in each of the pairs $(ab,cd),(a(a-b),c(c-d)),(c(c-2d),a(a-2b)),(c(2c-3d),a(2a-3b)),(c(c-3d),(a(a-3b))$ have the same sign.  Using the fact that $N|c$, we find that
\begin{align}\label{eq:B1}
\frac{d}{b}+\frac{bN}{d}-rN=\frac{c}{a}-\frac{aN}{c}-rN+\frac{1}{ab}-\frac{N}{cd}&\leq \frac{2}{r}-\frac{rN}{2}+1~, 
\\
{N(a-b)\over (c-d)}+ 
{c-d\over a-b}-rN=\frac{c}{a}-\frac{aN}{c}-rN+\frac{1}{a(a-b)}-\frac{N}{c(c-d)}&\leq \frac{2}{r}-\frac{rN}{2}+1~, 
\\\frac{N(a-2b)}{(c-2d)} + \frac{(c-2d)}{(a-2b)} - rN=\frac{c}{a}-\frac{aN}{c}-rN+\frac{2N}{c(c-2d)} - \frac{2}{a(a-2b)}&\leq \frac{2}{r}-\frac{rN}{2}+2~,
\\
\frac{N(2a-3b)}{(2c-3d)} + \frac{(2c-3d)}{(2a-3b)} - rN=\frac{c}{a}-\frac{aN}{c}-rN+\frac{3N}{c(2c-3d)} - \frac{3}{a(2a-3b)}&\leq \frac{2}{r}-\frac{rN}{2}+3~, 
\\
\frac{N(a-3b)}{(c-3d)} + \frac{(c-3d)}{(a-3b)} - rN=\frac{c}{a}-\frac{aN}{c}-rN+\frac{3N}{c(c-3d)} - \frac{3}{a(a-3b)} &\leq \frac{2}{r}-\frac{rN}{2}+3~. 
\label{eq:B5}
\end{align}
For $r=2$ we need a stricter bound. As in \cite{Bhand:2025ghn}, we need to find bounds on 
\begin{eqsp}\label{eq:cdab_bounds_req}
  & \frac{1}{ab}-\frac{N}{cd} ,\quad  \frac{1}{a(a-b)}-\frac{N}{c(c-d)} ,\quad  \frac{N}{c(c-2d)} - \frac{1}{a(a-2b)} ~,\\& \frac{N}{c(2c-3d)} - \frac{1}{a(2a-3b)} ,\quad \frac{N}{c(c-3d)} - \frac{1}{a(a-3b)}~,   
\end{eqsp}
the first three of which is required for $N=2,3$ and all five are required for $N=3$. 
We perform the analysis case by case. 
We first observe that for $\begin{pmatrix}
    a&b\\c&d
\end{pmatrix}\in\mr{P}\Gamma_1(N)^+$, we have $ad,bc,ab\in\IZ$ and $cd\in N\IZ$. Next, $\left|\frac{1}{ab}-\frac{N}{cd}\right|$ is maximized when $|ab|$ is minimized (maximized) and $|cd|$ is maximized (minimized). So, we find the maximum of $\left|\frac{1}{ab}-\frac{N}{cd}\right|$ with the constraint that $\begin{pmatrix}
    a&b\\c&d
\end{pmatrix}\in\tilde{\CS}_3$. We start with $ab=\pm 1$. Since $\frac{2}{r}\leq \frac{c}{a}<\frac{rN}{2}$, we get for $r=2$
\begin{eqsp}
    1\leq \pm bc<N~.
\end{eqsp}
But since $N|c$, this inequality has no solution. We then take $ab=\pm 2$. Then we get 
\begin{eqsp}
    1\leq \pm \frac{bc}{2}<N\implies bc=\pm N~.
\end{eqsp}
This gives $ad=1\pm N$. We get
\begin{eqsp}
    \left|\frac{1}{ab}-\frac{N}{cd}\right|=\left|\pm\frac{1}{2}-\frac{2}{1\pm N}\right|~.
\end{eqsp}
Note that for $N=2$, $ab=-2$ is not allowed since in that case 
\begin{eqsp}
    \frac{d}{b}=\frac{1}{2}<u_c\implies\begin{pmatrix}
        a&b\\c&d
    \end{pmatrix}\not\in S_3~. 
\end{eqsp}
Thus, when $|ab|$ is minimized and $|cd|$ is maximized, we obtain 
\begin{eqsp}\label{eq:abcdB1}
    \left|\frac{1}{ab}-\frac{N}{cd}\right|=\begin{cases}
        \frac{1}{6},&N=2,ab=2,\\0,&N=3,ab=2,
        \\
        \frac{1}{2},&N=3,ab=-2.
    \end{cases}
\end{eqsp}
Next, we consider the case when $|cd|$ is minimized and $|ab|$ is maximized. We need to consider $N=2,3$ cases separately. Take first $N=2$ and consider first the case $cd=\pm 2$. This is not allowed since for this case, we get 
\begin{eqsp}
    1\leq \frac{\pm 2}{ad}=\frac{\pm 2}{1+bc}<2~,
\end{eqsp}
which has no solution since $0\neq bc\in 2\IZ$. So, we consider $cd=\pm 4$. This gives 
\begin{eqsp}
    1\leq \frac{\pm 4}{1+bc}<2~.
\end{eqsp}
This has no solution for $cd=4$ since $bc$ is not allowed to be $\pm 2$ because for $bc=\pm 2$ gives $ab=\frac{3}{2},\frac{1}{2}\not\in\IZ$. So we try $cd=-4$, which gives $bc=-4$. We get 
\begin{eqsp}\label{eq:abcdN2B2}
    \left|\frac{1}{ab}-\frac{2}{cd}\right|=\frac{1}{6}~,\quad cd=-4~.
\end{eqsp}
For $N=3$, $cd= 3$ is not allowed for the same reason as $N=2,cd=N$. For $cd=-3$, we get
\begin{eqsp}
    1\leq -\frac{3}{1+bc}<3\implies bc=-3~.
\end{eqsp}
This gives 
\begin{eqsp}\label{eq:abcdB2N3}
 \left|\frac{1}{ab}-\frac{3}{cd}\right|=\frac{1}{2}~,\quad cd=-3~.   
\end{eqsp}
Thus combining \eqref{eq:abcdB1},\eqref{eq:abcdN2B2} and \eqref{eq:abcdB2N3}, we can write 
\begin{eqsp}
    \frac{1}{ab}-\frac{N}{cd}\leq \frac{1}{6}\delta_{N,2}+\frac{1}{2}\delta_{N,3}~,\quad r=2,\begin{pmatrix}
        a&b\\c&d
    \end{pmatrix}\in \tilde{\CS}_3~.
\end{eqsp}
To find a bound on $\left|\frac{1}{a(a-b)}-\frac{N}{c(c-d)}\right|$, we define $b'=b-a,d'=d-c$ so that $ad'-b'c=1$. Then from the previous calculation, we get 
\begin{eqsp}
    \frac{1}{a(a-b)}-\frac{N}{c(c-d)}\leq \frac{1}{6}\delta_{N,2}+ \frac{1}{2}\delta_{N,3}~,\quad r=2, \begin{pmatrix}
        a&b\\c&d
    \end{pmatrix}\in \tilde{\CS}_3~. 
\end{eqsp}
For $\left|\frac{N}{c(c-2d)} - \frac{1}{a(a-2b)}\right|$, we define $d''=2d-c,b''=2b-a$ so that $ad''-b''c=2$. As before $ab''=\pm 1$ and is not allowed. Moreover $ab''=-2$ is not allowed because the solution $b''c=-N$ to the inequality 
\begin{eqsp}
1\leq -\frac{b''c}{2}<N~,    
\end{eqsp}
gives 
\begin{eqsp}
    ad''=0,\quad N=2~,
    \\
    \frac{d''}{b''}=\frac{1}{2},\quad N=3~,
\end{eqsp}
both of which contradict the requirement the $\frac{d''}{b''}\in S_3$. Thus, we get 
\begin{eqsp}\label{eq:ab''bound}
   \left|\frac{1}{ab''}-\frac{N}{cd''}\right| = \left|\frac{1}{2}-\frac{2}{2+ N}\right|~.
\end{eqsp}
Thus, for the case when when $|ab''|$ is minimized and $|cd''|$ is maximized, we obtain 
\begin{eqsp}\label{eq:ab''cd''B1}
    \left|\frac{1}{ab''}-\frac{N}{cd''}\right|=\begin{cases}
        0,&N=2,ab''=2,\\\frac{1}{10},&N=3,ab''=2.
    \end{cases}
\end{eqsp}
Now, we analyze the case when $|ab''|$ is maximized and $|cd|$ is minimized. Let us start with $N=2$. For $N=2$, we notice that $4|cd''$. So we take $cd''=\pm 4$. The inequality 
\begin{eqsp}
  1\leq \frac{cd''}{ad''}=\frac{cd''}{2+b''c}<2\implies (cd'',b''c)=(4,2),(-4,-6)~. 
\end{eqsp}
This gives
\begin{eqsp}\label{eq:ab''cd''N2B2}
N=2:\quad      \left|\frac{1}{ab''}-\frac{2}{cd''}\right|=\begin{cases}
         0,&cd''=4,
         \\
         \frac{1}{3},&cd=-4.
     \end{cases}
\end{eqsp}
For $N=3$, $cd''=\pm 3$, there is no solution to the inequality 
\begin{eqsp}
    1\leq \frac{\pm 3}{ad''}=\frac{\pm 3}{2+b''c}<3~.
\end{eqsp}
So, we try $cd''=\pm 6$. We have the solution 
\begin{eqsp}
    1\leq \frac{cd''}{ad''}=\frac{cd''}{2+b''c}<3\implies (cd'',b''c)=(6,3),(-6,-6)~.
\end{eqsp}
But $cd''=6$ is not allowed since in this case $ab''=\frac{5}{2}\not\in\IZ$.
This gives 
\begin{eqsp}\label{eq:ab''cd''N3B2}
 N=3:\quad    \left|\frac{1}{ab''}-\frac{3}{cd''}\right|=\frac{1}{4}~,\quad cd''=6~.
\end{eqsp}
Combining \eqref{eq:ab''cd''B1}, \eqref{eq:ab''cd''N2B2} and \eqref{eq:ab''cd''N3B2} we can write 
\begin{eqsp}
    \frac{N}{c(c-2d)} - \frac{1}{a(a-2b)}&\leq \frac{1}{3}\delta_{N,2}+\frac{1}{4}\delta_{N,3}~,\quad r=2,\begin{pmatrix}
        a&b\\c&d
    \end{pmatrix}\in \tilde{\CS}_3~.
\end{eqsp}
For $N=3$, we need bound on $\frac{N}{c(2c-3d)} - \frac{1}{a(2a-3b)}$ and $\frac{N}{c(c-3d)} - \frac{1}{a(a-3b)}$. Let us start with $\frac{N}{c(2c-3d)} - \frac{1}{a(2a-3b)}$. Define $b'''=3b-2a,d'''=3d-2c$ so that $ad'''-b'''c=3$.  $ab'''=\pm 1$ is not allowed as before. Next, $ab'''=-2$ is also not allowed since the inequality 
\begin{eqsp}
    1\leq -\frac{b'''c}{2}<3\implies b'''c=-3~.
\end{eqsp}
This gives $ad'''=0$ which is a contradiction since $c/a,d'''/b'''\in S_3$. Then we are left with $ab'''=2$. This gives $b'''c=3$ which implies 
\begin{eqsp}
    \left|\frac{1}{ab'''}-\frac{3}{cd'''}\right|=\frac{1}{6}~,\quad ab'''=2~.
\end{eqsp}
Thus, in the case when $|ab'''|$ is minimized and $|cd'''|$ is maximized, we can write 
\begin{eqsp}\label{eq:cd'''B1N3}
 \frac{1}{ab'''}-\frac{3}{cd'''}\leq \frac{1}{6}~.    
\end{eqsp}
Next, we consider the case $cd'''=\pm 3$. This case is not allowed since the inequality 
\begin{eqsp}\label{eq:cd'''inequality}
    1\leq \frac{cd'''}{ad'''}=\frac{cd'''}{3+b'''c}<3~,
\end{eqsp}
has no solution for $cd'''=3$ and the solution $b'''c=-6$ for $cd'''=-3$ implies $d'''/b'''=1/2\not\in S_3$. The case $cd'''=6$ gives the solution $b'''c=3$ to the inequality \eqref{eq:cd'''inequality}. This gives
\begin{eqsp}\label{eq:cd'''B2N3}
    \left|\frac{1}{ab'''}-\frac{3}{cd'''}\right|=\frac{1}{6}~.
\end{eqsp}
The case $cd'''=-6$ gives the solution $b'''c=-6,-9$ to the inequality \eqref{eq:cd'''inequality}. In these cases, we get $ab'''=-3,-9$. So we get 
\begin{eqsp}\label{eq:cd'''B3N3}
    \left|\frac{1}{ab'''}-\frac{3}{cd'''}\right|=\begin{cases}
       \frac{1}{6},&cd'''=6,
       \\
       \frac{7}{18},&cd'''=-6.
   \end{cases}~.
\end{eqsp}
Combining \eqref{eq:cd'''B1N3}, \eqref{eq:cd'''B2N3} and \eqref{eq:cd'''B3N3} we can write 
\begin{eqsp}
   \frac{N}{c(2c-3d)} - \frac{1}{a(2a-3b)} \leq\frac{7}{18}\delta_{N,3}~.
\end{eqsp}
Finally, to get a bound on $\frac{N}{c(c-3d)} - \frac{1}{a(a-3b)}$, we write $b''''=3b-a,d''''=3d-c$ so that $ad''''-b''''c=3$ and we are back to the previous case. This gives us 
\begin{eqsp}
    \frac{N}{c(c-3d)} - \frac{1}{a(a-3b)}\leq\frac{7}{18}\delta_{N,3}~. 
\end{eqsp}
Thus for $r=2$, we can put the uniform upper bound of
\begin{eqsp}
    \frac{2}{r}-\frac{rN}{2}+N\left(1-\frac{2}{3}\delta_{N,2}\delta_{r,2}-\frac{1}{2}\delta_{N,3}\delta_{r,2}\right)~,
\end{eqsp}
on each of the terms in \eqref{eq:B1}--\eqref{eq:B5}. Note that for $N=2,3$, we have 
\begin{eqsp}
    \frac{2}{r}-\frac{rN}{2}+N\left(1-\frac{2}{3}\delta_{N,2}\delta_{r,2}-\frac{1}{2}\delta_{N,3}\delta_{r,2}\right)<0~,\quad\text{for}~~r\geq2~.
\end{eqsp}
Thus using the expression \eqref{eq:B_exp_N2} and \eqref{eq:B_exp_N3}, we get the following bound for $B$:
\begin{eqsp}
    B\leq \frac{L}{2N}\left(\frac{2}{r}-\frac{rN}{2}+N\left(1-\frac{2}{3}\delta_{N,2}\delta_{r,2}-\frac{1}{2}\delta_{N,3}\delta_{r,2}\right)\right)C_N(\Omega)\mu(a,b,c,d)<0~,
\end{eqsp}
for some $L>0$, where $C_N(\Omega),\mu(a,b,c,d)$ is defined in \eqref{eq:def_C2},\eqref{eq:def_C3} and \eqref{eq:def_mu}.
This proves the lemma.
\end{proof}
\begin{prop}
The series $\CF_4^{>3}$ converges absolutely and uniformly on compact subsets of the $\CR_N$-chamber.    
\end{prop}
\begin{proof}
Using Lemma \ref{lemma:coset_reps} and Lemma \ref{lemma:bound_B}, we have 
\begin{eqsp}
\left|\CF_4^{>3}\right|&\leq f_{-1}g_{-1} \sum_{r=2}^{\infty} r 
\sum_{\big{(}\begin{smallmatrix} a & b\cr c & d\end{smallmatrix}\big{)}\in G^N_r\backslash \CS_{3}}
e^{2\pi B}   \\
&\leq f_{-1}g_{-1} \sum_{r=2}^{\infty} r 
\sum_{\big{(}\begin{smallmatrix} a & b\cr c & d\end{smallmatrix}\big{)}\in \tilde{\CS}_{3}} 
\exp\left[2\pi \frac{L}{2N}\left(\frac{2}{r}-\frac{rN}{2}+N\left(1-\frac{2}{3}\delta_{N,2}\delta_{r,2}-\frac{1}{2}\delta_{N,3}\delta_{r,2}\right)\right)\right.\\&\hspace{10cm}C_N(\Omega)\mu(a,b,c,d)\bigg{]}\, .
\end{eqsp}
Following the same arguments as in \cite{Bhand:2025ghn} to bound the number of matrices with maximum of the absolute values of its entries equal to $M$, we get,
\begin{eqsp}
\left|\CF_4^{>3}\right|&\leq f_{-1}g_{-1} \sum_{r=2}^{\infty} r 
\sum_{M=1}^\infty 4(2M+1)^2 \, 
\exp\left[2\pi \frac{LM}{2N}\left(\frac{2}{r}-\frac{rN}{2}\right.\right.\\&\left.\left.\hspace{6cm}+N\left(1-\frac{2}{3}\delta_{N,2}\delta_{r,2}-\frac{1}{2}\delta_{N,3}\delta_{r,2}\right)\right)C_N(\Omega)\right]\\& <\infty\, .
\end{eqsp}
The absolute and uniform convergence on compact subsets of $\CR_N$ follows from the continuity of $C_N(\Omega)$.
\end{proof}
This proves the convergence of the sum appearing in \refb{eguessfin} for $N=2,3$.

\section{Properties of $\wt F_k(\Omega)$} \label{ssection5}

In this section we shall prove the results described in the Introduction. Since the analysis is similar to \cite{Bhand:2025ghn}, we will only highlight the differences from \cite{Bhand:2025ghn}. 
\par We start by considering the subtraction function:
\begin{eqsp} \label{eSdef}
S_k(\Omega) &:=\varepsilon_N \sum_{\big{(}\begin{smallmatrix} a & b\cr c & d\end{smallmatrix}\big{)}\in 
\mr{P}\Gamma_1(N)}
\left(e^{\pi i \{ac\tau + bd\sigma + (ad+bc)z\}} - e^{-\pi i \{ac\tau + bd\sigma + (ad+bc)z\}}\right)^{-2}  \\ & 
\hskip 1in 
\times\ f_+(a^2\tau +b^2\sigma +2abz) \ g_+(c^2\tau+d^2\sigma+2cd z) \\ 
&+\varepsilon_N 
\sum_{\big{(}\begin{smallmatrix} a & b\cr c & d\end{smallmatrix}\big{)}\in \mr{P}\Gamma_1(N)}
\sum_{r>0}  g_{-1} r  H(ac\tau_2 + bd\sigma_2 + (ad+bc)z_2) \\ &\hskip 1in \times
H\left(-ac\tau_2 - bd\sigma_2 - (ad+bc)z_2 + rNa^2\tau_2 + rNb^2\sigma_2+2rNab z_2
\right) \,
\\&\hskip 1in  \times 
e^{2\pi i \{-(c^2\tau+d^2\sigma+2cdz)/N+r( ac\tau+bd\sigma+(ad+bc)z )\}}f_+(a^2\tau+b^2\sigma+2abz) 
\\ &+\varepsilon_N 
\sum_{\big{(}\begin{smallmatrix} a & b\cr c & d\end{smallmatrix}\big{)}\in \mr{P}\Gamma_1(N)}
\sum_{r>0}  f_{-1} r  H(ac\tau_2 + bd\sigma_2 + (ad+bc)z_2) \\ &\hskip 1in \times\
H\left(-ac\tau_2 - bd\sigma_2 - (ad+bc)z_2 - rc^2\tau_2 + rd^2\sigma_2+2rcd z_2
\right) \,
\\&\hskip 1in  \times 
e^{2\pi i \{-(a^2\tau+b^2\sigma+2abz)+r( ac\tau+bd\sigma+(ad+bc)z )\}}g_+(c^2\tau+d^2\sigma+2cdz) \\ 
&+  f_{-1}g_{-1} \sum_{r>0} r\, 
\sum_{\big{(}\begin{smallmatrix} a & b\cr c & d\end{smallmatrix}\big{)}\in G^N_r\backslash \mr{P}\Gamma_1(N)^+}\hskip .1in 
 \bigg\{\prod_{n=-\infty}^\infty 
H(a_nc_n\tau_2 + b_nd_n\sigma_2 + (a_nd_n+b_nc_n)z_2)
\bigg\}
\\ & \hskip 1in \times \
e^{2\pi i \{(-a^2-c^2/N+r ac)\tau+(-b^2-d^2/N+r bd)\sigma+
(-2ab-2cd/N + r(ad+bc)) z\}}  \, .
\end{eqsp}
Let us discuss the proof of the
meromorphicity of $S_k(\Omega)$ on $\IH_2$ with double poles at 
\be\label{epoleint_thm}
-m_1\tau + n_1\sigma + m_2 + j\, z = 0, \qquad
m_1/N,n_1,m_2,j\in \mathbb{Z}, \qquad m_1n_1  +
{j^2\over 4} = {1\over 4}\,~ .
\ee
The same argument as in \cite[Lemma 5.1]{Bhand:2025ghn} shows that the hypersurface \eqref{epoleint_thm} is related to the hypersurface $z=0$ by the $\Gamma_1(N)$-action $\Omega\to\gamma\Omega\gamma^t$ for some $\gamma\in\Gamma_1(N)$ and a translation $\Omega\to\Omega+\begin{pmatrix}
    h_1&h_2\\h_2&Nh_3
\end{pmatrix}, h_i\in\IZ$. Thus in view of \eqref{esl3zoffi}, as in \cite{Bhand:2025ghn}, it suffices to prove that $S_k(\Omega)$ has a 
double pole at $z=0$ and then use the $\Gamma_1(N)$-invariance and translation 
invariance of $S_k(\Omega)$. As we approach the $z_2=0$ line from the $z_2<0$ or the $z_2>0$ side, the terms corresponding to 
$\begin{pmatrix}
    1&0\\0&1
\end{pmatrix}$ in the first sum, $\begin{pmatrix}
    1&0\\rN&1
\end{pmatrix}$,
$\begin{pmatrix}
    1&0\\0&1
\end{pmatrix}$ in the second and third sum and $\begin{pmatrix}
    0&-1/\sqrt{N}\\\sqrt{N}&0
\end{pmatrix}$, 
$\begin{pmatrix}
    1&0\\0&1
\end{pmatrix}$ in the fourth sum in \eqref{eSdef}
gives  
\begin{eqsp}\label{eq:S_pole_def}
    S_{k,\mr{pole}}(\Omega)&:=\left(e^{i\pi z}-e^{-i\pi z}\right)^{-2}f_+(\tau)g_+(\sigma) 
    \\&+g_{-1}e^{-2\pi i\sigma/N}f_+(\tau)\sum_{r>0}r 
    \left\{e^{-2\pi irz} H(-z_2)+ e^{2\pi irz} H(z_2)\right\} \\ &
    +f_{-1}e^{-2\pi i\tau}g_+(\sigma)\sum_{r>0}r 
    \left\{e^{-2\pi irz} H(-z_2)+ e^{2\pi irz} H(z_2)\right\}
    \\ &
    + f_{-1}g_{-1} \, e^{-2\pi i\tau} \, e^{-2\pi i\sigma/N} \, 
    \sum_{r>0}r \, \left\{e^{-2\pi irz} H(-z_2)+ e^{2\pi irz} H(z_2)\right\}
    \\&=\frac{(e^{\pi i z}- e^{-\pi i z})^{-2}}{ f^{(k)}(\tau)g^{(k)}(\sigma)} \, .
\end{eqsp}
Note that $S_{k,\mr{pole}}(\Omega)$ has a double pole at $z=0$ that exactly cancels
the double pole of $1/\Phi_{k}$ at $z=0$ given in \refb{epolestructure}. 
Then we can write\footnote{Despite the use of subscript ``hol'', $\CF_{i,\mr{hol}}$ continue to have poles on the subspaces \eqref{epoleint_thm} except at $z=0$.}
\be \label{espoleresults}
S_k(\Omega) = S_{k,\rm pole}(\Omega) + \CF_{1,\rm hol}(\Omega) +  \CF_{2,\rm hol}(\Omega) +  \CF_{3,\rm hol}(\Omega)+\CF_{4,\rm hol}(\Omega)\, ,
\ee
where
\begin{eqsp}\label{eq:S_hol}
\CF_{1,\rm hol}(\Omega)&:=    \varepsilon_N \sum_{\big{(}\begin{smallmatrix} 1 & 0\cr 0 & 1\end{smallmatrix}\big{)}\neq \big{(}\begin{smallmatrix} a & b\cr c & d\end{smallmatrix}\big{)}\in \mr{P}\Gamma_1(N)}\hskip .1in 
\left(e^{\pi i \{ac\tau + bd\sigma + (ad+bc)z\}} - e^{-\pi i \{ac\tau + bd\sigma + (ad+bc)z\}}\right)^{-2}  \\ & 
\hskip 1in 
\times\ f_+(a^2\tau +b^2\sigma +2abz) \ g_+(c^2\tau+d^2\sigma+2cd z)~,  \\ 
\CF_{2,\rm hol}(\Omega)&:=  \varepsilon_Ng_{-1} \sum_{r>0} r 
\sum_{\big{(}\begin{smallmatrix}
    1&0\\rN&1
\end{smallmatrix}\big{)},
\big{(}\begin{smallmatrix}
    1&0\\0&1
\end{smallmatrix}\big{)}\neq \big{(}\begin{smallmatrix} a & b\cr c & d\end{smallmatrix}\big{)}\in \mr{P}\Gamma_1(N)} %\hskip .1in 
 H(ac\tau_2 + bd\sigma_2 + (ad+bc)z_2)  \\ & \hskip 1in \times\
H\left(-ac\tau_2 - bd\sigma_2 - (ad+bc)z_2 + ra^2\tau_2 + rb^2\sigma_2+2rab z_2
\right) \,
\\ & \hskip 1in \times 
e^{2\pi i \{(-c^2+d^2\sigma+2cdz)/N+
r(ac\tau+bd\sigma+ r(ad+bc) )z\}}f_+(a^2\tau+b^2\sigma+2abz) ~,
\\
\CF_{3,\rm hol}(\Omega)&:=\varepsilon_N 
f_{-1} \sum_{r>0} r 
\sum_{\big{(}\begin{smallmatrix}
    1&0\\rN&1
\end{smallmatrix}\big{)},
\big{(}\begin{smallmatrix}
    1&0\\0&1
\end{smallmatrix}\big{)}\neq \big{(}\begin{smallmatrix} a & b\cr c & d\end{smallmatrix}\big{)}\in \mr{P}\Gamma_1(N)} 
\sum_{r>0}  f_{-1} r  H(ac\tau_2 + bd\sigma_2 + (ad+bc)z_2) \\ &\hskip 1in \times\
H\left(-ac\tau_2 - bd\sigma_2 - (ad+bc)z_2 - rc^2\tau_2 + rd^2\sigma_2+2rcd z_2
\right) \,
\\&\hskip 1in  \times 
e^{2\pi i \{-(a^2\tau+b^2\sigma+2abz)+r( ac\tau+bd\sigma+(ad+bc)z )\}}g_+(c^2\tau+d^2\sigma+2cdz) 
 \\ 
\CF_{4,\rm hol}(\Omega)&:=  
f_{-1}g_{-1} \sum_{r>0} r\, 
\sum_{\gamma_N,\mathds{1}_2\neq \big{(}\begin{smallmatrix} a & b\cr c & d\end{smallmatrix}\big{)}\in G^N_r\backslash \mr{P}\Gamma_1(N)^+}
 \bigg\{\prod_{n=-\infty}^\infty 
H(a_nc_n\tau_2 + b_nd_n\sigma_2 + (a_nd_n+b_nc_n)z_2)
\bigg\}
\\ & \hskip 1in \times \
e^{2\pi i \{(-a^2-c^2/N+r ac)\tau+(-b^2-d^2/N+r bd)\sigma+
(-2ab-2cd/N + r(ad+bc)) z\}} \, .
\end{eqsp}
Following the arguments in \cite{Bhand:2025ghn} almost verbatim, one can prove the convergence of $\CF_{1,\mr{hol}},\CF_{2,\mr{hol}}$ and $\CF_{3,\mr{hol}}$ on the hypersurface $z=0$. In the proof of the convergence of $\CF_{4,\mr{hol}}$, we need to analyze the bounds $C_{Ni}$ defined in \eqref{eq:CNi_explicit}. For $z=0$, we see that 
\begin{eqsp}
    C_{21}=C_{24}=C_{31}=C_{32}=0~.
\end{eqsp}
Thus, we need to analyze the exponents of the terms corresponding to $i=1,4$ for $N=2$ and $i=1,2$ for $N=3$ in \eqref{e497yx} more carefully. Using \eqref{eq:tsz_primes_def}, the exponent in the relevant terms  of \eqref{e487x} for $z_2=0$ can be written as
\begin{eqsp}
 (N,i)&=(2,1):\quad   2\pi \left(\tau_2+\frac{\sigma_2}{2}-n(r-n)\sigma_2\right)~,
 \\
 (N,i)&=(2,4):\quad 2\pi\left(\frac{\sigma}{2}+(2(n+1)(n+1-r)+1)\right)~,
 \\
 (N,i)&=(3,1):\quad 2\pi\left(\frac{\sigma}{2}+(3(n+1)(n+1-r)+1)\right)~,
 \\
 (N,i)&=(3,2):\quad 2\pi \left(\tau_2+\frac{\sigma_2}{3}-n(r-n)\sigma_2\right)~.
\end{eqsp}
Thus, only the $n=0,r-1$ term causes a problem in the convergence of this sum.  From \eqref{eq:CSr_def}, we see that these problematic terms correspond to the matrices 
\begin{eqsp}
 \begin{pmatrix}
     0&-\frac{1}{\sqrt{N}}\\\sqrt
     N&0
 \end{pmatrix}~,\quad \begin{pmatrix}
     1&0\\N(r-1)&1
 \end{pmatrix}\begin{pmatrix}
     1&0\\N&1
 \end{pmatrix}=\begin{pmatrix}
     1&0\\rN&1
 \end{pmatrix}=g_{r,N}\begin{pmatrix}
     0&\frac{1}{\sqrt{N}}\\-\sqrt
     N&0
 \end{pmatrix}~.   
\end{eqsp}
Since these terms have been removed in the definition of $\CF_{4,\mr{hol}}$, we conclude that $\CF_{4,\mr{hol}}$ converges in the at $z= 0$.
\par At this point, we notice that we could have proven the convergence of the generating function in the $\CL_N$-chamber by following the same arguments as in Section \ref{stildeFconverge}. Indeed, we can make the transformation 
\begin{eqsp}
    \Omega\to\gamma_N\Omega\gamma_N^{t}~,\quad \gamma_N=\begin{pmatrix}
        0&-\frac{1}{\sqrt{N}}\\\sqrt{N}&0
    \end{pmatrix}~,
\end{eqsp}
which maps $\CR_N$-chamber to $\CL_N$-chamber leaving the common boundary $z_2=0$ fixed, to all equations in Section \ref{stildeFconverge}. Then proceeding as above, we can prove the convergence of $\CF_{i,\mr{hol}},i=1,2,3,4$ on $\CL_N\cup\{z=0\}$. Moreover, the functions definitely agree at $z_2=0$. This means that $S_k-S_{k,\mr{pole}}$ is a holomorphic function on $\CR_N\cup\CL_N\cup\{z=0\}$. The $\Gamma_1(N)$-invariance then implies that $S_k(\Omega)$ is meromorphic on $\IH_2$ with double poles on the hypersurface \eqref{epoleint_thm}.
\par This analysis simplifies immensely for the case when $\Phi_k(\Omega)$ and hence also $S_k(\Omega)$ is formally invariant under $\Omega\to\gamma_N\Omega\gamma_N^t$.
By simply applying the transformation $\gamma_N$ to the $\CR_N$-chamber, we get the convergence in the $\CL_N$-chamber and the equality at $z_2=0$ follows from the fact that $\gamma_N$ leaves $z_2=0$ line fixed. Thus, $S_k(\Omega)-S_{k,\mr{pole}}$ defines a holomorphic function on $\CR_N\cup\CL_N\cup\{z=0\}$. 
\par The same arguments as in \cite{Bhand:2025ghn} proves that 
\begin{eqsp}
    \wt F_{k}(\Omega)=\frac{1}{\Phi_k}-S_k(\Omega)~,
\end{eqsp}
does not have any singularity in the domain $\mr{det}\,\mr{Im}\,\Omega>\frac{1}{4}$. In fact the generating function $\wt F_k$ admits a meromorphic continuation to all of 
$\IH_2$ with double poles at 
\begin{eqnarray}
&& n_2 (\tau\sigma - z^2) - m_1\tau + n_1\sigma + m_2 + j\, z = 0, \nonumber \\
&&
m_1/N,n_1,m_2,n_2,j\in \mathbb{Z}, \qquad n_2\ge 1, \qquad m_1n_1 + m_2 n_2 +
{j^2\over 4} = {1\over 4}\, .
\end{eqnarray}
To prove \eqref{etildeFFequality}, we first notice that since $\wt F_k(\Omega)$ is does not have any poles in the domain $\mr{det}\,\mr{Im}\,\Omega>\frac{1}{4}$, for any contour  $\wt\CC\subset\left\{\Omega\in\IH_2:\mr{det}\,\mr{Im}\,\Omega>\frac{1}{4}\right\}$ of the form 
\begin{eqsp}
    0\leq \tau_1,z_2<1~,\quad 0\leq \sigma_1<N~,
\end{eqsp}
we get from \eqref{edeftildeF}
\begin{eqsp}
    \wt d^*(m,n,\ell)=\frac{(-1)^{\ell+1}}{N}\int_{\wt\CC}d\tau d\sigma dz\,e^{-2\pi i(m\tau+n\sigma/N+\ell z)}\wt F_k(\Omega)~.
\end{eqsp}
Moreover, we are allowed to deform the contour $\wt\CC$ to the contour $\CC_{m,n,\ell}$ defined in \eqref{eq 3.3} without changing the answer. Arguing as in 
\cite{Bhand:2025ghn}, we conclude that 
\begin{eqsp}
    \int_{\CC_{m,n,\ell}}d\tau d\sigma dz\,e^{-2\pi i(m\tau+n\sigma/N+\ell z)}S_k(\Omega)=0~.
\end{eqsp}
This implies that 
\begin{eqsp}
\wt d^*(m,n,\ell)&=\frac{(-1)^{\ell+1}}{N}\int_{\wt\CC}d\tau d\sigma dz\,e^{-2\pi i(m\tau+n\sigma/N+\ell z)}\wt F_k(\Omega)\\&=\frac{(-1)^{\ell+1}}{N}\int_{\CC_{m,n,\ell}}d\tau d\sigma dz\,e^{-2\pi i(m\tau+n\sigma/N+\ell z)}\frac{1}{\Phi_k(\Omega)}\\&=d^*(m,n,\ell)~.
\end{eqsp}
\section{Conclusion and future directions}\label{sec:conc}
In this paper, we have generalized the results of \cite{Bhand:2025ghn} and constructed the generating function of single-centered black hole degeneracy for general CHL models. We also proved the convergence of the generating function for $N=2,3$. The results of this paper leave ample opportunity for future work. We list some possible future directions. 
\subsection{Generalizing the proof of convergence to $N>3$}
Our proof of convergence of the generating function is complete for $N=2,3$. We present some general steps in the proof of convergence for general $N$ but some of the steps in the proof fail for $N>3$. The main technical bottleneck is the infinity of walls of the fundamental chamber $\CR_N$ for $N\geq 4$. The first place this will cause trouble is in the bound on $\tau_2',\sigma_2',z_2'$ derived in \eqref{eq:tau2,sig2_bound} and \eqref{eq:z2'_bound}. One might be able to overcome this issue by writing infinitely many expression for $\tau_2',\sigma_2',z_2'$ as linear combinations of walls of the $\CR_N$ chamber as we did for $N=2,3$. But this will require a detailed knowledge of all the walls of the $\CR_N$-chamber. 
\par As mentioned in  Remark \ref{rem:lemma_nec}, Lemma \ref{lemma:Heaviside_matr_const} is an if and only if statement only for $N\leq 6$. For $N\geq 7$, Lemma \ref{lemma:Heaviside_matr_const} is only a necessary condition. We will need a more detailed knowledge of the walls of $\CR_N$ to get an equivalent condition for the nonvanishing of the Heaviside functions. This affects the proof of the convergence of $\CF_4$. Thus, some novel methods are needed to establish the convergence of $\wt F_k$ in full generality. 
\par Finally, we prove the convergence of the sum in the $\CR_N$-chamber and then use the formal $\Gamma_1(N)$-invariance to analytically continue the generating function to all of $\IH_2$. This requires that $\CR_N$ be a fundamental chamber, in the sense that any other chamber can be reached from $\CR_N$ by a $\Gamma_1(N)$-transformation. This is proved for $N\leq 3$ in \cite{Sen:2007vb}. For $N\geq 4$, this remains to be proven.  
\subsection{Transformation properties of the generating function}
It was noted in \cite{Bhand:2025ghn} that the generating function $\wt F_k$ might exhibit transformation properties of a mock Siegel modular form. The theory of mock Siegel modular forms is not well developed in the mathematics literature, except for a few relevant papers \cite{bringmann2012kohnen,bringmann2016almost,raum}. In analogy with mock modular and mock Jacobi forms, we expect that a weight $k$ mock Siegel modular form for $\mathrm{Sp}(2,\IZ)$ must be a holomorphic function $\Psi(\Omega)$ on $\IH_2$ satisfying 
\begin{eqsp}
    \Psi(\gamma\Omega\gamma^t)=\Psi(\Omega)~,\quad \Psi(\Omega+A)=\Psi(\Omega)~,
\end{eqsp}
for $\gamma\in\mr{SL}(2,\IZ)$ and $A$ a symmetric $2\times 2$ integer matrix, such that there exists a nonholomorphic function $\Psi^*$ such that the completion $\wh{\Psi}=\Psi+\Psi^*$ transforms as a Siegel modular form of weight $k$ under $\mathrm{Sp}(2,\IZ)$. Let us discuss an approach of proving such a transformation property for the simplest case of $\CM=K3,N=1$. The generating function $\wt F$ given in \cite{Bhand:2025ghn} is 
\begin{eqnarray}
\label{eguessfinintroN=1}
\wt F(\Omega) &=& {1\over \Phi_{10}(\Omega)}
- {1\over 2} \sum_{\big{(}\begin{smallmatrix} a & b\cr c & d\end{smallmatrix}\big{)}\in 
\mr{PSL}(2,\mathbb{Z})}
\left(e^{\pi i \{ac\tau + bd\sigma + (ad+bc)z\}} - e^{-\pi i \{ac\tau + bd\sigma + (ad+bc)z\}}\right)^{-2} \nonumber \\ && 
\hskip 1in 
\times\ f_+(a^2\tau +b^2\sigma +2abz) \ f_+(c^2\tau+d^2\sigma+2cd z) \nonumber \\ 
&-&  f_{-1} \sum_{r>0} r 
\sum_{\big{(}\begin{smallmatrix} a & b\cr c & d\end{smallmatrix}\big{)}\in  \mr{PSL}(2,\mathbb{Z})}
 H(ac\tau_2 + bd\sigma_2 + (ad+bc)z_2) \nonumber \\ && \hskip 1in \times\
H\left(-ac\tau_2 - bd\sigma_2 - (ad+bc)z_2 + ra^2\tau_2 + rb^2\sigma_2+2rab z_2
\right) \,
\nonumber \\ && \hskip 1in \times 
e^{2\pi i \{-(c^2\tau+d^2\sigma+2cdz)+r( ac\tau+bd\sigma+(ad+bc)z )\}}f_+(a^2\tau+b^2\sigma+2abz) 
\nonumber \\ 
&-&  f_{-1}^2 \sum_{r>0} r 
\sum_{\big{(}\begin{smallmatrix} a & b\cr c & d\end{smallmatrix}\big{)}\in G_r\backslash \mr{PSL}(2,\mathbb{Z})}\hskip .1in 
 \bigg\{\prod_{n=-\infty}^\infty 
H(a_nc_n\tau_2 + b_nd_n\sigma_2 + (a_nd_n+b_nc_n)z_2)
\bigg\}
\nonumber \\ && \hskip 1in \times \
e^{2\pi i \{(-a^2-c^2+r ac)\tau+(-b^2-d^2+r bd)\sigma+
(-2ab-2cd + r(ad+bc)) z\}}  \, ,
\end{eqnarray}
where 
\begin{eqsp}
    f(\tau)=\frac{1}{\eta(\tau)^{24}}~,
\end{eqsp}
where $\eta$ is the Dedekind eta function.
Let us focus on the first term:
\begin{eqsp}
\sum_{\big{(}\begin{smallmatrix} a & b\cr c & d\end{smallmatrix}\big{)}\in 
\mr{PSL}(2,\mathbb{Z})}
&\left(e^{\pi i \{ac\tau + bd\sigma + (ad+bc)z\}} - e^{-\pi i \{ac\tau + bd\sigma + (ad+bc)z\}}\right)^{-2} \\ & 
\hskip 1in 
\times\ f_+(a^2\tau +b^2\sigma +2abz) \ f_+(c^2\tau+d^2\sigma+2cd z)   \end{eqsp}
Let us write 
\begin{eqsp}
    P(\Omega):=\left(e^{i\pi z}-e^{-i\pi z}\right)^{-2}f_+(\tau)f_+(\sigma)~.
\end{eqsp}
Then we can write this term using the slash operator 
\begin{eqsp}
    \sum_{\gamma\in \mr{PSL}(2,\IZ)}(P\underset{-10}{|}g_4(\gamma))(\Omega)~,
\end{eqsp}
where the weight-$k$ slash operator on functions on $\IH_2$ is defined as
\begin{eqsp}
    (P\underset{k}{|}M)(\Omega):=\mr{det}\,(C\Omega+D)^{-k}P((A\Omega+B)(C\Omega+D)^{-1})~,\quad M=\begin{pmatrix}
        A&B\\C&D
    \end{pmatrix}\in \mr{Sp}(2,\IZ)~,
\end{eqsp}
and
\begin{eqsp}
    g_4(\gamma)=\begin{pmatrix}
        \gamma&0\\0&(\gamma^t)^{-1}
    \end{pmatrix}\in\mr{Sp}(2,\IZ)~,\quad \gamma\in\mr{SL}(2,\IZ)~.
\end{eqsp}
Note that $g_4:\mr{SL}(2,\IZ)\longrightarrow \mr{Sp}(2,\IZ)$ is a group homomorphism, i.e., $g_4(\gamma\gamma')=g_4(\gamma)g_4(\gamma')$. 
Let $\Gamma_\infty$ be the stabilizer subgroup of the cusp $i\infty$ of $\mr{SL}(2,\IZ)$:
\begin{eqsp}
    \Gamma_\infty:=\left\{\begin{pmatrix}
        1&r\\0&1
\end{pmatrix}:r\in\IZ\right\}~.
\end{eqsp}
Then we can write 
\begin{eqsp}
\sum_{\gamma\in \mr{PSL}(2,\IZ)}(P\underset{-10}{|}\gamma)(\Omega)=\sum_{\gamma\in\Gamma_\infty\backslash\mr{PSL}(2,\IZ)}\sum_{\Tilde{\gamma}\in\Gamma_\infty}((P\underset{-10}{|}g_4(\Tilde{\gamma}))\underset{-10}{|}g_4(\gamma))(\Omega)~.    
\end{eqsp}
Recall that the index $m$ averaging operator  is given by \cite{Dabholkar:2012nd}
\begin{eqsp}
    \mr{Av}^{(m)}_{\zeta,q}[h(\zeta)]:=\sum_{r \in \mathbb{Z}} q^{m r^2} \zeta^{2 m r} h\left(q^r \zeta\right)~,\quad \zeta:=e^{2\pi iz},q:=e^{2\pi i\tau}~.
\end{eqsp}
A simple calculation shows that 
\begin{eqsp}
  \mr{Av}^{(n)}_{\zeta,p}[h(\zeta)]q^n=\sum_{\tilde{\gamma}\in\Gamma_\infty}(h(\zeta)q^n)\underset{k}{|}g_4(\tilde{\gamma})~,\quad p:=e^{2\pi i\sigma}~.  
\end{eqsp}
Next, recall that for $s=(\alpha,\beta)\in\IQ^2$ the order 2 Appell-Lerch sum is given by \cite{Dabholkar:2012nd} 
\begin{eqsp}
   \mathcal{B}_m^s(\tau, z)=\mathbf{e}\left(-m \alpha z_s(\tau)\right) \mathrm{Av}_{\zeta,q}^{(m)}\left[\mathcal{R}_{-2 m \alpha}^{(2)}\left(\zeta / \zeta_s(\tau)\right)\right]~,
\end{eqsp}
where $z_s(\tau)=\alpha\tau+\beta$ and 
\begin{eqsp}
\mathcal{R}_c^{(2)}(\zeta)&=\sum_{\ell \in \mathbb{Z}} \frac{|\ell-c|+\operatorname{sgn}(\operatorname{Im}(z))(\ell-c)}{2} \zeta^{\ell}
\\&=\zeta^{\lfloor c\rfloor+1}\left(\frac{1}{(\zeta-1)^2}+\frac{c-\lfloor c\rfloor}{\zeta-1}\right)~.
\end{eqsp}
We have 
\begin{eqsp}
  \sum_{\tilde{\gamma}\in\Gamma_\infty}(P|g_4(\Tilde{\gamma})(\Omega)&=\sum_{n\geq 0}f_n\sum_{\tilde{\gamma}\in\Gamma_\infty}\left.\left[\frac{\zeta}{(1-\zeta)^2}f_+(\sigma)q^n\right]\right|g_4(\Tilde{\gamma}) 
  \\&=f_+(\sigma)\sum_{n\geq 0}f_n\CB^0_n(\sigma,z)q^n~,
\end{eqsp}
The completion $\wh{\CB}^s_n(\sigma,z)$ of $\CB^s_n(\sigma,z)$ to a nonholomorphic function which transforms as a Jacobi form of weight $-10$ and index $n$ was found in \cite{Dabholkar:2012nd}. Thus a candidate for the completion of the first term is 
\begin{eqsp}
    \sum_{n\geq 0}f_n\sum_{\gamma\in\Gamma_\infty\backslash\mr{PSL}(2,\IZ)}\left[f_+(\sigma)\wh{\CB}^0_n(\sigma,z)q^n\right]\underset{-10}{\Big|}g_4(\gamma)~.
\end{eqsp}
Note that the combination $f_+(\sigma)\wh{\CB}^0_n(\sigma,z)$ does not transform as a Jacobi form because $f_+(\sigma)$ does not transform as a modular form and hence we cannot invoke the (generalized version, which also needs to be proven, of the) converse theorem of \cite{raum2013formal}. But the hope is that the completions of the other terms in \eqref{eguessfinintroN=1} combine into a nonhomolomrphic function which transforms as a Siegel modular form. This remains an interesting open problem. 
\\\\
\noindent{\bf Acknowledgments:}
I would like to thank Ajit Bhand and Ashoke Sen for collaboration at initial stages of this work. I am   especially indebted to Ashoke Sen for several helpful discussions. I would like to acknowledge the hospitality provided by Galileo Galilei Institute, Florence during the workshop on ``Pathways to Quantum Black Holes: from Effective Theories to Exact Methods'' where part of this work was done.
This work was supported by the US Department of Energy under grant 
DE-SC0010008. 
\newpage
\appendix
\section{Bound State Metamorphosis}\label{app:BSM}
In this appendix, we summarize the results of \cite{1104.1498} and discuss the identification of bound states used in the construction of the generating function in Section \ref{section3}.

The moduli space of a CHL model is parametrized by a number of scalars and the axio-dilaton moduli. We fix the value of other moduli fields and consider the space parametrized by the axio-dilaton moduli, which we denote by $\tau_\infty$. This space is the complex upper half plane $\IH$. The walls in the moduli space correspond to matrices \cite{Sen:2007vb,Sen:2007qy}
\begin{equation}
\left(\begin{array}{ll}
a & b \\
c & d
\end{array}\right) \in \mr{SL}(2, \mathbb{Z}), \quad c d \in N \mathbb{Z} .
\end{equation}
For nonvanishing $c, d$, the wall is a circle which intersects the real $\tau_{\infty}$-axis at
\begin{equation}
\frac{a}{c} \quad \text { and }\quad  \frac{b}{d}~.
\end{equation}
When $c$ or $d$ vanishes, then the wall is a straight line connecting an integer on the real $\tau_{\infty}$-axis and $i\infty$. We choose orientation of the line connecting points $A,B$ to be the direction from the starting point $A$ to the end point $B$ and the region to the left/right of these walls is with respect to this orientation.
\par These walls in the moduli space are in 1-1 correspondence with walls in the Siegel upper half space described by linear poles of $\Phi_k^{-1}$ :
\begin{equation}
-m_1 \tau_2+n_1 \sigma_2+j z_2=0, \quad m_1 \in N \mathbb{Z}, n_1, j \in \mathbb{Z}, \quad m_1 n_1+\frac{j^2}{4}=\frac{1}{4}~ .
\end{equation}
The precise map is \cite{1104.1498}
\begin{equation}
\left(\frac{b}{d}, \frac{a}{c}\right) \longleftrightarrow a b \sigma_2+c d \tau_2+(a d+b c) z_2=0 ~.
\end{equation}
We will denote the bound state of $(Q,P),(Q',P')$ by 
\begin{eqsp}
(Q,P)\oplus (Q',P')~.    
\end{eqsp} 
Let us recall the rules for the existence of bound states \cite{1104.1498}. 
Given a charge vector ($Q, P$) and a point $\tau_{\infty}^\circ$ in the $\tau_{\infty}$-plane, a bound state correspond to matrices 
\begin{eqsp}
    \begin{pmatrix}
        a&b\\c&d
    \end{pmatrix}\in \Gamma_1(N)~.
\end{eqsp}
More precisely, given $\begin{pmatrix}
        a&b\\c&d
    \end{pmatrix}\in\Gamma_1(N)$, the bound state
\begin{equation}\label{eq:bound_gen_charge}
(a(d Q-b P), c(d Q-b P)) \oplus(b(-c Q+a P), d(-c Q+a P))~,
\end{equation}
exists if one of the conditions below is satisfied:
\begin{enumerate}[(1)]
    \item $(d Q-b P) \cdot(-c Q+a P)>0$ and $\tau_{\infty}^{\circ}$ lies to the left of the wall connecting $\frac{b}{d}$ to $\frac{a}{c}$.
\item $(d Q-b P) \cdot(-c Q+a P)<0$ and $\tau_{\infty}^{\circ}$ lies to the right of the wall connecting $\frac{b}{d}$ to $\frac{a}{c}$.
\end{enumerate}
Mapping the wall corresponding to $\left(\begin{array}{ll}a & b \\ c & d\end{array}\right)$ to the wall

\begin{equation}
a b \sigma_2+c d \tau_2+(a d+b c) z_2=0,
\end{equation}
the bound state in \eqref{eq:bound_gen_charge} exist in the chamber
\begin{equation}
a b \sigma_2+c d \tau_2+(a d+b c) z_2>0~,\quad \text{when}\quad (d Q-b P) \cdot(-c Q+a P)>0~,
\end{equation}
and in the chamber
\begin{equation}
a b \sigma_2+c d \tau_2+(a d+b c) z_2<0~,\quad\text{when}\quad (d Q-b P) \cdot(-c Q+a P)<0~.
\end{equation}
Let us now discuss the identifications from bound state metamorphosis (BSM). The proposal in \cite{1104.1498} is the following: 
\begin{enumerate}[(1)]
    \item A bound state\footnote{This bound state corresponds to the matrix $\begin{pmatrix}
        \pm1&0\\0&\pm1
    \end{pmatrix}$.} $(Q,0)\oplus (0,P)$ with $P^2=-2$ should be identified with the bound state\footnote{This bound state corresponds to the matrix $\begin{pmatrix}
        1&-r\\0&1
    \end{pmatrix}$.} $(Q+rP,0)\oplus (-rP,P)$. The general rule states that for $r:=Q\cdot P>0$, $(Q,0)\oplus (0,P)$ exists to the left of the wall connecting $0$ and $i\infty$ and $(Q+rP,0)\oplus (-rP,P)$ exists to the right of the wall connecting $-r$ to $i\infty$. When $r=Q\cdot P<0$ then $(Q,0)\oplus (0,P)$ exists to the right of the wall connecting $0$ and $i\infty$ and $(Q+rP,0)\oplus (-rP,P)$ exists to the left of the wall connecting $-r$ to $i\infty$. These bound states exist only if they exist in both these chambers. 
    \item A bound state $(Q,0)\oplus (0,P)$ with $Q^2=-2/N$ must be identified with the bound state\footnote{This bound state corresponds to the matrix $\begin{pmatrix}
        1&0\\-rN&1
    \end{pmatrix}$.} $(Q,-rNQ)\oplus (0,P+rNQ)$. The chamber in which these bound states exist can be determined by the rules above depending on sign of $r$ and these bound states exist only in the intersection of these chambers. 
    \item A bound state $(Q,0)\oplus(0,P)$ with $Q^2=-2/N,P^2=-2$ must be identified with bound states $(Q+rP,0)\oplus (-rP,P),(Q,-rNQ)\oplus (0,P+rNQ)$ and they exist in the common chamber of existence of all these bound states. There are further identification of these bound states with the bound states corresponding to matrices obtained by taking alternate products of
\begin{equation}
\left(\begin{array}{cc}
1 & -r \\
0 & 1
\end{array}\right) \text { and }\left(\begin{array}{cc}
1 & 0 \\
r N & 1
\end{array}\right) \text { or }\left(\begin{array}{cc}
1 & 0 \\
-r N & 1
\end{array}\right) \text { and }\left(\begin{array}{ll}
1 & r \\
0 & 1
\end{array}\right)~.
\end{equation}
Note that for $N=1$ and $r=1$, these matrices do not give new bound states.     
\end{enumerate}
By applying $S$-duality transformations, we can now explain the identification used to construct the generating function. Again, we discuss the identifications in three cases.
\begin{enumerate}[(1)]
    \item Consider charges $\left(\bN^{\prime}, \bM^{\prime}\right)$ with $\bM^{\prime 2}=-2$, the bound state $\left(\bN^{\prime}, 0\right) \oplus\left(0, \bM^{\prime}\right)$ must be identified with the bound state $\left(\bN^{\prime}+r \bM^{\prime}, 0\right) \oplus\left(-r \bM^{\prime}, \bM^{\prime}\right)$ where $r=\bM^{\prime} \cdot \bN^{\prime}$. Applying the S-duality transformation for the $\Gamma_1(N)$ matrix

\begin{equation}
\gamma=\left(\begin{array}{ll}
d & b \\
c & a
\end{array}\right) \in \Gamma_1(N)~,
\end{equation}
we get that we must identify the bound states 
\begin{eqsp}
&\left(d \bN^{\prime}, c \bN^{\prime}\right) \oplus\left(b \bM^{\prime}, a \bM^{\prime}\right)~~\text{and}~~
\\&\left(d\left(\bN^{\prime}+r \bM^{\prime}\right), c\left(\bN^{\prime}+r \bM^{\prime}\right)\right) \oplus\left((b-d r) \bM^{\prime},(a-c r) \bM^{\prime}\right)~.    
\end{eqsp}
Let us call
\begin{equation}
\bN=\bN^{\prime}+r \bM^{\prime}, \quad \bM=\bM^{\prime}~.
\end{equation}
Then, we have
\begin{equation}
\bM^2=\bM^{\prime 2}, \quad \bN^2=\bN^{\prime 2}, \quad \bM \cdot \bN=-\bM^{\prime} \cdot \bN^{\prime}~.
\end{equation}
Moreover, BSM dictates that we must identify the bound state 
\begin{eqsp}
\left(d \bN^{\prime}, c \bN^{\prime}\right) \oplus\left(b \bM^{\prime}, a \bM^{\prime}\right)~~\text{with}~~ (d \bN, c \bN) \oplus((b-d r) \bM,(a-c r) \bM)~.   
\end{eqsp}
Equivalently, we must identify the bound states corresponding to the matrices
\begin{eqsp}
    \begin{pmatrix}
        a&b\\c&d
    \end{pmatrix}\quad \text{and}\quad \begin{pmatrix}
        a-cr&b-dr\\c&d
    \end{pmatrix}=\begin{pmatrix}
        1&-r\\0&1
    \end{pmatrix}\begin{pmatrix}
        a&b\\c&d
    \end{pmatrix}~.
\end{eqsp}
To determine the chambers in which these bound states exist, let us define
\begin{equation}
\binom{Q}{P}:=\left(\begin{array}{ll}
d & b \\
c & a
\end{array}\right)\binom{\bN^{\prime}}{\bM^{\prime}} \Rightarrow\binom{\bN^{\prime}}{\bM^{\prime}}=\left(\begin{array}{cc}
a & -b \\
-c & d
\end{array}\right)\binom{Q}{P}
\end{equation}
so that $\bM^{\prime}=d P-c Q, \bN^{\prime}=-b P+a Q$. Thus, the bound state $\left(d \bN^{\prime}, c \bN^{\prime}\right) \oplus\left(b \bM^{\prime}, a \bM^{\prime}\right)$ exists in the chamber
\begin{equation}
a c \tau_2+b d \sigma_2+(a d+b c) z_2>0~,
\end{equation}
when $\bM^{\prime} \cdot \bN^{\prime}=r>0$ and in the chamber
\begin{equation}
a c \tau_2+b d \sigma_2+(a d+b c) z_2<0~,
\end{equation}
when $\bM^{\prime} \cdot \bN^{\prime}=-r<0$. Similarly, the bound state $(d \bN, c \bN) \oplus((b-d r) \bM,(a-c r) \bM)$ exists in the chamber
\begin{equation}
a c \tau_2+b d \sigma_2+(a d+b c) z_2-c^2 r \tau_2-d^2 r \sigma_2-2 c d r z_2<0~,
\end{equation}
when $\bM \cdot \bN=-\bM^{\prime} \cdot \bN^{\prime}=-r<0$ and in the chamber
\begin{equation}
a c \tau_2+b d \sigma_2+(a d+b c) z_2-c^2 r \tau_2-d^2 r \sigma_2-2 c d r z_2>0~,
\end{equation}
when $\bM \cdot \bN=-\bM^{\prime} \cdot \bN^{\prime}=-r>0$.\par
\item For charges ($\bM^{\prime}, \bN^{\prime}$) with $\bN^{\prime 2}=-2 / N$, we need to identify the bound states $\left(\bN^{\prime}, 0\right) \oplus\left(0, \bM^{\prime}\right)$ with the bound state $\left(\bN^{\prime},-r N \bN^{\prime}\right) \oplus\left(0, \bM^{\prime}+r N \bN^{\prime}\right)$ where $r=\bM^{\prime} \cdot \bN^{\prime}$. Applying $S$-duality transformation for the
$\Gamma_1(N)$ matrix
\begin{equation}
\gamma=\left(\begin{array}{ll}
d & b \\
c & a
\end{array}\right) \in \Gamma_1(N),
\end{equation}
we get that we must identify the bound states
\begin{eqsp}
&\left(d \bN^{\prime}, c \bN^{\prime}\right) \oplus\left(b \bM^{\prime}, a \bM^{\prime}\right)~~\text{and}~~ \\&\left((d-r N b) \bN^{\prime},(c-arN) \bN^{\prime}\right) \oplus\left(d\left(\bM^{\prime}+r N \bN^{\prime}\right), c\left(\bM^{\prime}+r N \bN'\right)\right)~.
\end{eqsp}
Let us call
\begin{equation}
\bM=\bM^{\prime}+r N \bN^{\prime}~, \quad \bN=\bN^{\prime}~ .
\end{equation}
Then, we have
\begin{equation}
\bM^2=\bM^{\prime 2}, \quad \bN^2=\bN^{\prime 2}, \quad \bM \cdot \bN=-\bM^{\prime} \cdot \bN^{\prime} .
\end{equation}
Moreover, BSM dictates that we must identify bound state 
\begin{eqsp}
\left(d \bN^{\prime}, c \bN^{\prime}\right) \oplus\left(b \bM^{\prime}, a \bM^{\prime}\right)~~\text{and}~~  (d \bM, c \bM) \oplus\left((d-r N b) \bN,(c-arN) \bN\right)~.  
\end{eqsp}
Equivalently, we must identify the bound states corresponding to the matrices
\begin{eqsp}
    \begin{pmatrix}
        a&b\\c&d
    \end{pmatrix}\quad \text{and}\quad \begin{pmatrix}
        a&b\\c-arN&d-brN
    \end{pmatrix}=\begin{pmatrix}
        1&0\\-rN&1
    \end{pmatrix}\begin{pmatrix}
        a&b\\c&d
    \end{pmatrix}~.
\end{eqsp}
The bound state
$\left(d \bN^{\prime}, c \bN^{\prime}\right) \oplus\left(b \bM^{\prime}, a \bM^{\prime}\right)$ exists in the chamber
\begin{equation}
a c \tau_2+b d \sigma_2+(a d+b c) z_2>0~,
\end{equation}
when $\bM^{\prime} \cdot \bN^{\prime}=r>0$ and in the chamber
\begin{equation}
a c \tau_2+b d \sigma_2+(a d+b c) z_2<0~,
\end{equation}
when $\bM^{\prime} \cdot \bN^{\prime}=-r<0$. Similarly, the bound state $(d \bM, c \bM) \oplus\left((d-r N b) \bN,(c-a r N) \bN\right)$ exists in the chamber
\begin{equation}
a c \tau_2+b d \sigma_2+(a d+b c) z_2-a^2 r N \tau_2-b^2 r N \sigma_2-2 a b r N z_2<0~,
\end{equation}
when $\bM \cdot \bN=\bM^{\prime} \cdot \bN^{\prime}=-r<0$ and in the chamber
\begin{equation}
a c \tau_2+b d \sigma_2+(a d+b c) z_2-a^2 r N \tau_2-b^2 r N \sigma_2-2 a b r N z_2>0~,
\end{equation}
when $\bM \cdot \bN=\bM^{\prime} \cdot \bN^{\prime}=-r>0$.
\item Finally, the proposal for $\bM^{\prime 2}=-2, \bN'^2=-2 / N$, then we identify bound states $\left(\bN^{\prime}, 0\right) \oplus\left(0, \bM^{\prime}\right),\left(\bN^{\prime}+r \bM^{\prime}, 0\right) \oplus\left(-r \bM^{\prime}, \bM'\right)$ and $\left(\bN^{\prime},-r N \bN^{\prime}\right) \oplus\left(0, \bM^{\prime}+r N \bN^{\prime}\right)$. Moreover, we also identify all bound states corresponding to $\Gamma_1(N)$ matrices obtained by taking alternate products of
\begin{equation}
\left(\begin{array}{cc}
1 & -r \\
0 & 1
\end{array}\right) \text { and }\left(\begin{array}{cc}
1 & 0 \\
r N & 1
\end{array}\right) \text { or }\left(\begin{array}{cc}
1 & 0 \\
-r N & 1
\end{array}\right) \text { and }\left(\begin{array}{ll}
1 & r \\
0 & 1
\end{array}\right)~.
\end{equation}
Applying $S$-duality transformation for the
$\Gamma_1(N)$ matrix
\begin{equation}
\gamma=\left(\begin{array}{ll}
d & b \\
c & a
\end{array}\right) \in \Gamma_1(N),
\end{equation}
we get that we must identify the bound states
\begin{eqsp}
&\left(d \bN^{\prime}, c \bN^{\prime}\right) \oplus\left(b \bM^{\prime}, a \bM^{\prime}\right)~~\text{and}~~
\\
&\left(d\left(\bN^{\prime}+r \bM^{\prime}\right), c\left(\bN^{\prime}+r \bM^{\prime}\right)\right) \oplus\left((b-d r) \bM^{\prime},(a-c r) \bM^{\prime}\right)~~\text{and}
\\
&\left((d-r N b) \bN^{\prime},(c-arN) \bN^{\prime}\right) \oplus\left(d\left(\bM^{\prime}+r N \bN^{\prime}\right), c\left(\bM^{\prime}+r N \bN'\right)\right)~.
\end{eqsp}
Let us call
\begin{eqsp}
&\bN=\bN'+r\bM'~,\quad \bM=\bM'~,\\
&\wt{\bM}=\bM^{\prime}+r N \bN^{\prime}~, \quad \wt{\bN}=\bN^{\prime}~ .
\end{eqsp}
Then, we have
\begin{equation}
\bM^2=\wt{\bM}^2=\bM^{\prime 2}, \quad \bN^2=\wt{\bN}^2=\bN^{\prime 2}, \quad \bM \cdot \bN=\wt{\bM}\cdot\wt{\bN}=-\bM^{\prime} \cdot \bN^{\prime}~ .
\end{equation}
Moreover, BSM dictates that we must identify bound state 
\begin{eqsp}
&\left(d \bN^{\prime}, c \bN^{\prime}\right) \oplus\left(b \bM^{\prime}, a \bM^{\prime}\right)~~\text{and}~~
\\
&(d \bN, c \bN) \oplus\left((b-dr) \bM,(a-cr) \bM\right)~~\text{and}
\\
&(d \wt{\bM}, c \wt{\bM}) \oplus\left((d-r N b) \wt{\bN},(c-arN) \wt{\bN}\right)~.  
\end{eqsp}
Equivalently, we must identify the bound states corresponding to the matrices
\begin{eqsp}
    &\begin{pmatrix}
        a&b\\c&d
    \end{pmatrix}\quad \text{and}
    \\&\begin{pmatrix}
        a-cr&b-dr\\c&d
    \end{pmatrix}=\begin{pmatrix}
        1&-r\\0&1
    \end{pmatrix}\begin{pmatrix}
        a&b\\c&d
    \end{pmatrix}\quad \text{and}
    \\
    &\begin{pmatrix}
        a&b\\c-arN&d-brN
    \end{pmatrix}=\begin{pmatrix}
        1&0\\-rN&1
    \end{pmatrix}\begin{pmatrix}
        a&b\\c&d
    \end{pmatrix}~.
\end{eqsp}
The bound state
$\left(d \bN^{\prime}, c \bN^{\prime}\right) \oplus\left(b \bM^{\prime}, a \bM^{\prime}\right)$ exists in the chamber
\begin{equation}
a c \tau_2+b d \sigma_2+(a d+b c) z_2>0~,
\end{equation}
when $\bM^{\prime} \cdot \bN^{\prime}=r>0$ and in the chamber
\begin{equation}
a c \tau_2+b d \sigma_2+(a d+b c) z_2<0~,
\end{equation}
when $\bM^{\prime} \cdot \bN^{\prime}=-r<0$. Similarly, the bound state $(d \wt{\bM}, c \bm{\bM}) \oplus\left((d-r N b) \bm{\bN},(c-a r N) \bm{\bN}\right)$ exists in the chamber
\begin{equation}
a c \tau_2+b d \sigma_2+(a d+b c) z_2-a^2 r N \tau_2-b^2 r N \sigma_2-2 a b r N z_2<0~,
\end{equation}
when $\wt{\bM} \cdot \wt{\bN}=\bM^{\prime} \cdot \bN^{\prime}=-r<0$ and in the chamber
\begin{equation}
a c \tau_2+b d \sigma_2+(a d+b c) z_2-a^2 r N \tau_2-b^2 r N \sigma_2-2 a b r N z_2>0~,
\end{equation}
when $\wt{\bM} \cdot \wt{\bN}=\bM^{\prime} \cdot \bN^{\prime}=-r>0$. Since both these transformations change the sign of $\bM'\cdot\bN'$, this means that the term $\begin{pmatrix}
    a&b\\c&d
\end{pmatrix}$ in \eqref{ef1analysis1}, \eqref{ef1analysis2} must be identified with the terms $\begin{pmatrix}
        1&-r\\0&1
    \end{pmatrix}\begin{pmatrix}
        a&b\\c&d
    \end{pmatrix}$ and $\begin{pmatrix}
        1&0\\-rN&1
    \end{pmatrix}\begin{pmatrix}
        a&b\\c&d
    \end{pmatrix}$ in \eqref{ef1analysis3}, \eqref{ef1analysis4}. In \eqref{eq:f-1g-1_pre}, these terms correspond to matrices
    \begin{eqsp}
        \gamma_N\begin{pmatrix}
        1&-r\\0&1
    \end{pmatrix}\begin{pmatrix}
        a&b\\c&d
    \end{pmatrix}\quad \text{and}\quad \gamma_N\begin{pmatrix}
        1&0\\-rN&1
    \end{pmatrix}\begin{pmatrix}
        a&b\\c&d
    \end{pmatrix}~,
    \end{eqsp}
where 
\begin{eqsp}
 \gamma_N=\left(\begin{array}{cc}
0 & -1 / \sqrt{N} \\
\sqrt{N} & 0
\end{array}\right)~.   
\end{eqsp}
Moreover, we need to introduce the corresponding Heaviside functions for these two terms. 
Next, we must also identify the contributions 
\begin{eqsp}
    \begin{pmatrix}
        a&b\\c&d
    \end{pmatrix}\quad \text{and}\quad \left(\begin{array}{cc}
1 & -r \\
0 & 1
\end{array}\right)\left(\begin{array}{cc}
1 & 0 \\
r N & 1
\end{array}\right)\begin{pmatrix}
        a&b\\c&d
    \end{pmatrix}\quad  \text{and}\quad \left(\begin{array}{cc}
1 & 0 \\
-r N & 1
\end{array}\right)\left(\begin{array}{ll}
1 & r \\
0 & 1
\end{array}\right)\begin{pmatrix}
        a&b\\c&d
    \end{pmatrix}~, 
\end{eqsp}
where we do not need any factors of $\gamma_N$ because these transformations do not change the sign of $\bM'\cdot\bN'$. 
Note that
\begin{eqsp}
\left(\begin{array}{cc}
1 & 0 \\
r N & 1
\end{array}\right)&=-\gamma_N\left(\begin{array}{cc}
1 & -r \\
0 & 1
\end{array}\right) \gamma_N~.  
\end{eqsp}
Thus we have 
\begin{eqsp}
\left(\begin{array}{cc}
1 & -r \\
0 & 1
\end{array}\right)\left(\begin{array}{cc}
1 & 0 \\
r N & 1
\end{array}\right)&=-\left(\gamma_N\left(\begin{array}{cc}
1 & 0 \\
r N & 1
\end{array}\right)\right)^2~,
\\
\left(\begin{array}{cc}
1 & 0 \\
-r N & 1
\end{array}\right)\left(\begin{array}{ll}
1 & r \\
0 & 1
\end{array}\right)&=-\left(\gamma_N\left(\begin{array}{cc}
1 & -r \\
0 & 1
\end{array}\right)\right)^2~.
\end{eqsp}
Now noting that 
\begin{eqsp}
  \gamma_N\begin{pmatrix}
        1&0\\-rN&1
    \end{pmatrix}&=-\left(\gamma_N\begin{pmatrix}
        1&-r\\0&1
    \end{pmatrix}\right)^{-1}~,  
\end{eqsp}
we see that BSM amounts to identifications
\begin{eqsp}
    \begin{pmatrix}
        a&b\\c&d
    \end{pmatrix}\quad \text{and}\quad g_{r,N}^n\begin{pmatrix}
        a&b\\c&d
    \end{pmatrix}~,\quad g_{r,N}:=\gamma_N\begin{pmatrix}
        1&-r\\0&1
    \end{pmatrix}=\begin{pmatrix} 0 & -1/\sqrt{N}\\\sqrt{N}& -r\sqrt{N}\end{pmatrix},\quad n\in\IZ~.
\end{eqsp}
\end{enumerate}
\bibliography{main.bib}

\end{document}